\documentclass[11pt]{article}

\usepackage{amsmath,amsthm,amscd,amssymb,graphicx,subfigure,tikz}

\usepackage{latexsym}

\usepackage{mathrsfs}

\newcommand{\bbN}{{\mathbb{N}}}

\newcommand{\bbR}{{\mathbb{R}}}

\newcommand{\bbE}{{\mathbb{E}}}

\newcommand{\bbZ}{{\mathbb{Z}}}

\newcommand{\eps}{\varepsilon}

\newcommand{\beq}{\begin{equation}}

\newcommand{\eeq}{\end{equation}}

\newcommand{\ba}{\begin{align}}

\newcommand{\ea}{\end{align}}

\newcommand{\EE}{\mathbb E}

\DeclareMathOperator{\Tr}{Tr}
\DeclareMathOperator{\Var}{Var}

\newcommand{\lozr}{
--++(1,1)--++(1,-1)--++(-1,-1)
--++(-1,1)[fill=red!40!white]}

\newcommand{\lozd}{--++(1,-1)--++(0,2)--++(-1,1)--++(0,-2)
[fill=white!60!green]}

\newcommand{\lozu}{--++(1,1)--++(0,2)--++(-1,-1)--++(0,-2)
[fill=white!60!blue]}

\allowdisplaybreaks

\numberwithin{equation}{section}

\newtheorem{theorem}{Theorem}[section]

\newtheorem{proposition}[theorem]{Proposition}

\newtheorem{lemma}[theorem]{Lemma}

\newtheorem{corollary}[theorem]{Corollary}

\theoremstyle{definition}

\newtheorem{example}[theorem]{Example}
\newtheorem{assumption}[theorem]{Assumption}

\theoremstyle{remark}

\newtheorem{remark}{Remark}[section]

\title{On global fluctuations for non-colliding processes}
\author{Maurice Duits \footnote{Department of Mathematics, Royal Institute of Technology (KTH), Stockholm Lindstedsv\"agen 25, SE-10044, Sweden. Email: duits@kth.se.  Supported  by the  Swedish Research Council (VR) Grant no.\ 2012-3128.} }
\date{\today}

\begin{document}
\maketitle
\begin{abstract}
We study the global fluctuations for a class of determinantal point processes coming from  large systems of non-colliding processes and non-intersecting paths. Our main assumption is that the point processes are constructed by biorthogonal families that satisfy finite term recurrence relations.  The central observation of the paper is that the fluctuations of multi-time or multi-layer linear statistics can be efficiently expressed in terms of  the associated recurrence matrices. As a consequence, we prove  that different models that share the same asymptotic behavior of the recurrence matrices, also share the same asymptotic behavior for the global fluctuations.  An important special case is when the recurrence matrices have limits along the diagonals,  in which case we prove  Central Limit Theorems for the linear statistics. We then show that these results prove Gaussian Free Field fluctuations for the random surfaces associated to these systems. To illustrate the results, several examples will be discussed, including non-colliding processes for which the invariant measures are the classical orthogonal polynomial ensembles and random lozenge tilings of a hexagon. 
\end{abstract}
\section{Introduction}

Random surfaces appearing  in various models of integrable probability, such as random matrices and random tilings,  are known to have a rich structure. A particular feature, one that has received much attention in recent years, is the Gaussian free field that is expected to appear as a universal field describing the global fluctuations of such random surfaces. Using  the integrable structure, this has been rigorously verified in a number of models in the literature. For a partial list of reference see \cite{Bwigner,BB,BF,BG,D,Kenyon,Kuan,Petrov}.  The results so far indicate the universality of the Gaussian Free Field in this context is rather robust. In this paper, we will be interested in the global fluctuations for a particular class of models, namely that of non-colliding processes and ensembles of non-intersecting paths.  For those models, we will provide a general principle that leads to  Gaussian Free Field type fluctuations. 
 
Non-colliding process and non-intersecting path ensembles form an important class of two dimensional random interacting systems. For instance, Dyson showed \cite{Dyson}  that the Gaussian Unitary Ensemble is the invariant measure for a system of non-colliding Ornstein-Uhlenbeck processes.   Replacing the Ornstein-Uhlenbeck process  by its radially squared version  defines a similar stochastic dynamics for the Laguerre Unitary Ensemble \cite{KOC}.  In Section \ref{sec:example} we will recall  a generalization to non-colliding processes for which the classical orthogonal polynomial ensembles (continuous and discrete) are the invariant measures. Another source of models that lead  non-intersecting paths  are random tilings of planar domains. Lozenge tilings of a hexagon on a triangular lattice, for example,  can alternatively be described by discrete random walks that start and end at opposite sites of the hexagon \cite{G1,Jhahn}. In that way,  a probability measure on all tilings of the hexagon induces a non-intersecting path ensemble.

By the tandem of the  Karlin-McGregor or Gessel-Lindstr\"om-Viennot Theorem and the Eynard-Mehta Theorem, it follows that many non-colliding process are determinantal point processes, see, e.g., \cite{Jdet}. This makes them especially tractable for asymptotic analysis. A natural way to study the random surfaces appearing in this way is to analyze   linear statistics for such determinantal point processes. 
 \emph{The purpose of this paper is to prove  Central Limit Theorems for multi-time or multi-layer linear statistics for a certain class of determinantal point processes. } The conditions under which the results hold are easily verified in the classical ensembles. In particular,  we will show that this Central Limit Theorem confirms the universality conjecture for the Gaussian free field in these models. We will illustrate our results by considering several examples, including dynamic  extension of Unitary Ensembles and other ensembles related to (bi)orthogonal polynomials.\\

In the remaining part of the  Introduction we give an example of what type of results we will prove by discussing  a well-known model. Consider the left picture as shown in Figure \ref{fig:BB} showing  $n$ Brownian bridges $t \mapsto \gamma_j(t)$ that start for $t=0$ at the origin and return there at $t=1$. We also condition the bridges never to collide. 
The $\gamma_j(t)$ turn out to have the same distribution as the ordered eigenvalues of an $n\times n$ Hermitian matrix for which the real and imaginary part of each entry independently (up to the symmetry) performs a Brownian bridge.  Hence, at any given time $t \in (0,1)$ the locations $\gamma_j(t)$ have the exact same distribution of the appropriately scaled eigenvalues of a matrix taken randomly from the Gaussian Unitary Ensemble (=GUE).  For more details and background  on this, by now classical, example we refer to \cite{AGZ,Dyson}.  See also \cite{Handbook} for a general reference on Random Matrix Theory.

 Because of the non-colliding condition, we can view the paths as level lines for a random surface. More precisely, if we define the height function by 
\begin{equation}\label{eq:defheight}
h_n(t,x)=\# \{ j \mid \gamma_j(t) \leq x\},
\end{equation}
i.e. $h_n(t,x)$ counts the number of paths directly below a given point $(t,x)$, then the trajectories are the lines where the stepped surface defined by the graph of $h_n(t,x)$ makes a jump by one. The  question is what happens with $h_n$ when the size of the system grows large, i.e. $n \to \infty$. 
\begin{figure}[t]
\begin{center}
\includegraphics[scale=.4]{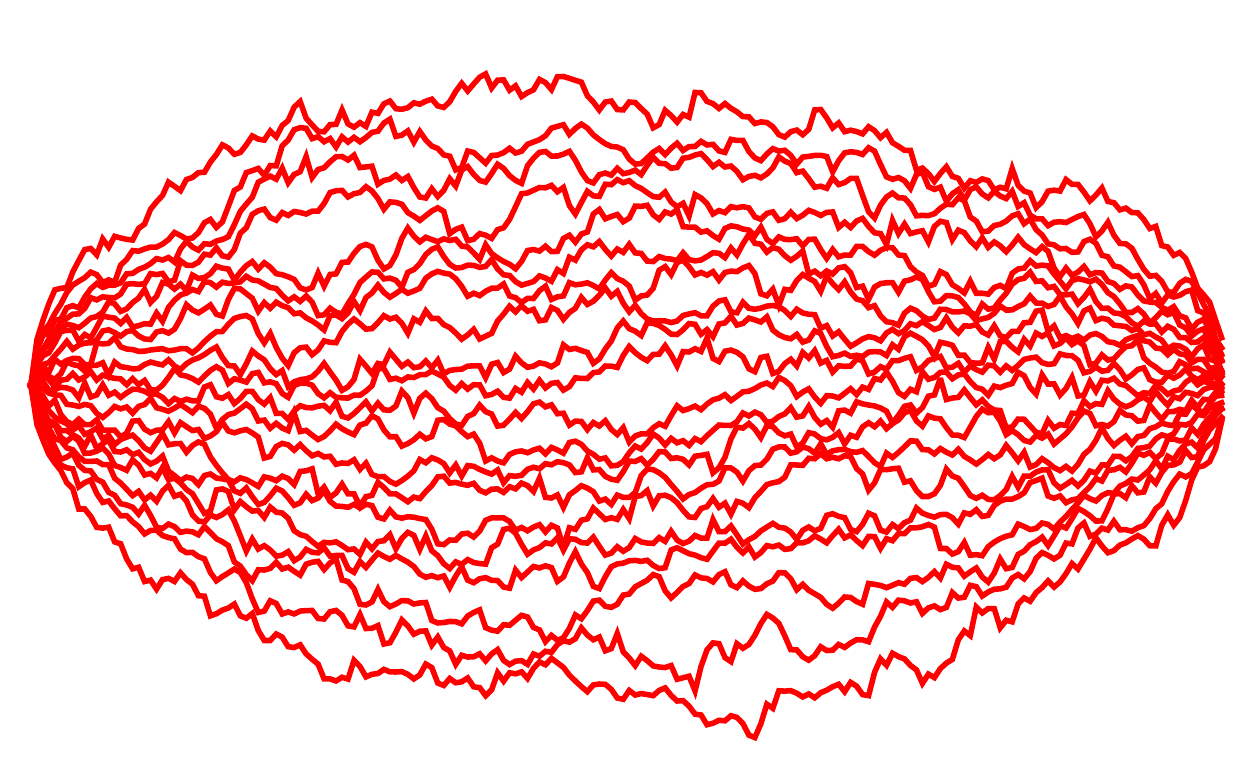}
\includegraphics[scale=.4]{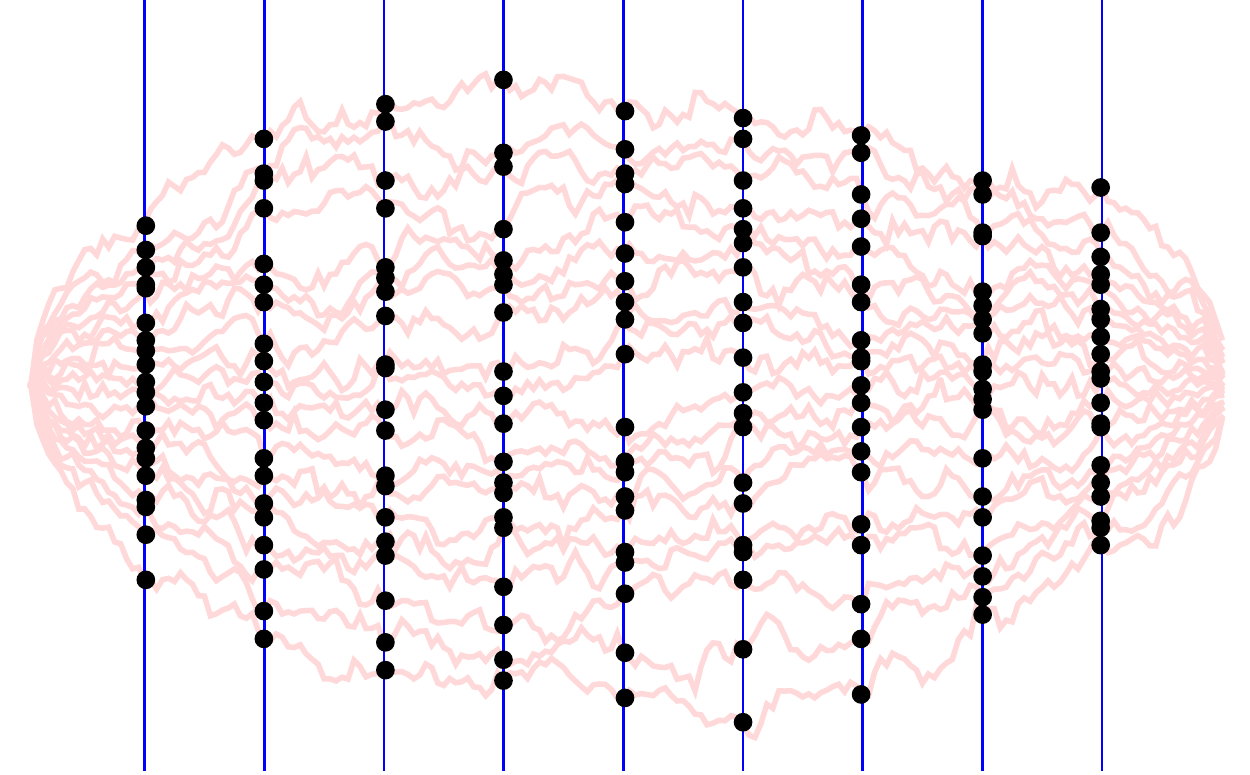}
\end{center}
\caption{The left picture shows a typical configuration of non-colliding Brownian bridges that are conditioned never to collide. In the right picture we take a number of vertical sections of the bridge. }  \label{fig:BB}
\end{figure}

It turns out that as $n \to \infty$ the normalized height function $\frac1n h_n$ has an almost sure limit, also called the limit shape.  Indeed, when $n \to \infty$ the paths will all accumulate on a region $\mathcal E$ that is usually referred to as the disordered region. We will assume that the original system is rescaled such that $\mathcal E$ does not depend on $n$ and is a non-empty open domain in $\mathbb R^2$.  In fact, after a  proper rescaling, the domain $\mathcal E$ is  the ellipse
\begin{equation}
\label{eq:ellipse}\mathcal E= \{(t,x)\, \mid  \, x^2 \leq 4t(1-t)\}.
\end{equation}
It is  well-known that the eigenvalue distribution of a GUE matrix converges to that of a semi-circle law (see, e.g., \cite{AGZ}). This implies that  we have the following limit for the height function
$$\lim_{n\to \infty} \frac{1}{n}  \EE h_n(t,x)= \frac1 {2\pi t(1-t)} \int_{-\sqrt{4 t(1-t)}}^x \sqrt{4 t(1-t)-s^2}{\rm d} s.$$
The next question is about the fluctuations of the random surface, i.e. the behavior of $h_n(t,x) -\EE h_n(t,x)$. For a fixed point $(t,x)$ we note that  $h_n(t,x)$ is a counting statistic counting the number of eigenvalues  of a suitably normalized GUE matrix  in a given semi-infinite interval $(-\infty,x]$. The variance for such a statistic is known to be growing logarithmically $\sim c \ln n$ as $n \to \infty$. Moreover, by dividing by the variance we find that 
$$\frac{h_n(t,x)-\EE h_n(t,x)}{\sqrt{\Var h_n(t,x)}} \to N(0,1),$$
as $n \to \infty$. The principle behind this result goes back to Costin-Lebowitz \cite{CL} and was later extended by Soshnikov \cite{SoshCLT}. However, if we consider the correlation between several points that are macroscopically far apart,
$$\EE \left[\prod_{j=1}^R\left( h_n(t_j,x_j)-\EE h_n(t_j,x_j)\right)\right],$$ we
obtain a finite limit as $n \to \infty$. When $n\to \infty$ the random surface defined by the graph of $h_n(t,x)-\EE h_n(t,x)$ converges to a rough surface. The pointwise limit does not make sense (due to the growing variance) but it has a limit as a generalized function. This generalized function is, up to coordinate transform, known in the  literature as the Gaussian free field. 

Since the Gaussian Free Field is a random generalized function, it is natural  to pair it  with a test function $\phi$, i.e.
\begin{equation}\label{eq:wrongpairing}
\langle h_n,\phi\rangle=\iint h_n(t,x) \phi(t,x) {\rm d} t {\rm d} x,
\end{equation} (as we will show in Section \ref{sec:GFF} the Dirichlet pairing is more appropriate, also including a coordinate transform, but this simpler pairing does show the essential idea). Then by  writing $h_n(t,x)= \sum_{j=1}^n \chi_{(-\infty,\gamma_j(t)]}(x)$ and by a discretization of the time integral we obtain
\begin{multline}\label{eq:intro:discr}
\langle h_n,\phi\rangle=\sum_{j=1}^n \int \int_{-\infty}^ {\gamma_j(t)} \phi(t,x) {\rm d} x {\rm d} t
= \sum_{m=1}^N  \sum_{j=1}^n  (t_{m+1}-t_m) \int_{-\infty}^ {\gamma_j(t_m)} \phi(t_m,x) {\rm d}x,
\end{multline} 
where  we choose  $N\in \bbN$ time points $t_m \in (0,1)$  such that 
$$0=t_0<t_1 < \ldots t_N<t_{N+1}=1,$$
and typically want the mesh $\sup_{m=0,\ldots,N}(t_{m+1}-t_m)$ to be small. The fact of the matter is that the right-hand side is a  \emph{linear statistic} for the  point process on $\{1,\ldots,N\}\times \bbR$ defined by the locations 
\begin{equation}\label{eq:pointprocess}
\{(m,\gamma_j(t_m)) \}_{j=1,m=1}^{n,N}.
\end{equation}
In other words, the pairing in \eqref{eq:wrongpairing} naturally leads us  to studying  linear statistics $X_n(f)$ defined by 
\begin{equation}\label{eq:introlinstat}
X_n(f) = \sum_{m=1}^N \sum_{j=1}^n f(m,\gamma_j(t_m)),
\end{equation}
for a function $f$ on $\{1, \ldots, N\} \times \bbR$.  The central question of the paper is to ask for the limiting behavior, as $n \to \infty$,  of the fluctuations of $X_n(f)-\EE X_n(f)$ for sufficiently smooth functions $f$.  A particular consequence of the main results is the following.
\begin{proposition}\label{prop:intro}
Let $f:\{1,\ldots,N\}\times \bbR \to \bbR$ such that $x\mapsto f(m,x)$ is continuously differentiable and grows at most polynomially for $x \to \pm \infty$. Then the linear statistic \eqref{eq:introlinstat} for the point process \eqref{eq:pointprocess} satisfies
\begin{equation}\label{eq:CLTintro} X_n(f)-\EE X_n(f) \to N(0,\sigma_f^2)
\end{equation}
as $n \to \infty$,
where $$
\sigma(f)^2 = \sum_{m_1,m_2=1}^N  \sum_{k=1}^\infty{\rm e}^{-|\tau_{m_1}-\tau_{m_2} |k} k f_k^{(m_1)} f_{k}^{(m_2)}, $$
with $\tau_m=\frac12\ln t_m/(1-t_m)$ and 
$$f_k^{(m)}= \frac{1}{\pi} \int_0^\pi f\left(m,2 \sqrt{t_m(1-t_m)}–\cos \theta\right) \cos k \theta {\rm d} \theta,$$
for $k \in \bbN$. 
\end{proposition}
This Central Limit Theorem is a special case of a more general theorem that we will state in the next section. The main point of the present paper is to show that such results follow from a  general principle for models that have a determinantal structure for which the integrating functions (i.e. the Hermite polynomials in the above example) satisfy a finite term recurrence.    The proof of Proposition \ref{prop:intro} will be discussed in Example \ref{eq:OU} (see also \cite{Bwigner} for a similar statement in the context of stochastic evolutions for Wigner matrices).  The precise connection to the Gaussian Free Field will be explained in Section \ref{sec:GFF}.

\subsubsection*{Overview of the rest of the paper}
In Section \ref{sec:mainresults} we will formulate the model that we will consider and state our main results, including various corollaries. The proofs of those corollaries will also be given in Section \ref{sec:mainresults}, but the proofs of the main results, Theorem ~\ref{th:main0}, ~\ref{th:extend}, ~\ref{th:main0growing} and~\ref{th:extendgrowing} will be given in Section \ref{sec:proofs}. Our approach is a connection to recurrence matrices, which will be explained  in Section \ref{sec:rec}. Then in Section~\ref{sec:fred} we will analyze the asymptotic behavior of a general Fredholm determinant from which the proofs of the main results are special cases. Finally, in Section \ref{sec:example} we will provide  ample  examples to illustrate our results. 


 \section{Statement of results}\label{sec:mainresults}
 
 In this section we will discuss the general  model that we will consider and state our main results. Some proofs are postponed to later sections. 
 
 \subsection{The model}
 
 Determinantal point processes that come from non-colliding process and non-intersecting paths have a particular form.  In this paper, we will therefore consider probability measures on $(\mathbb R^n)^N$ that can be written as the product of several determinants,
 \begin{multline}
\label{eq:productmeasures}\frac{1}{Z_n} \det \left(\phi_{j,1}(x_{1,k})\right)_{j,k=1}^n \prod_{m=1}^{N-1} \det\left( T_m(x_{m,i},x_{m+1,j}) \right)_{i,j=1}^n\\
\times  \det \left(\psi_{j,N}(x_{N,k})\right)_{j,k=1}^n
\prod_{m=1}^N \prod_{j=1}^n {\rm d} \mu_m(x_{m,j}),
\end{multline}
 where $Z_n$ is a normalizing constant, ${\rm d} \mu_m$ are Borel measures on $\mathbb R$, $\phi_{j,1} \in \mathbb L_2(\mu_N)$  and $ \psi_{j,N} \in \mathbb L_2(\mu_1)$. The function  $T_m$  is such that the integral operator $\mathcal T_m :\mathbb L_2(\mu_m) \to \mathbb L_2(\mu_{m+1})$ 
 defined by 
$$\mathcal T_m f(y)= \int f(x) T_m(x,y) {\rm d} \mu_m(x),$$ is a bounded  operator.

The form of \eqref{eq:productmeasures} may look very special at first, but it appears often  when dealing with non-colliding processes and non-intersecting paths. See, e.g., \cite{Jdet} and the references therein. The key is the Karlin-McGregor Theorem in the continuous setting or the Gessel-Lindstr\"om-Viennot Theorem in the discrete setting, that say that the transition probabilities of non-colliding processes can be expressed as determinants of a matrix constructed out of the transition probability  for a single particle. We will discuss several explicit examples in Section \ref{sec:example}.

It is standard that without loss of generality we can assume that 
\begin{equation} \label{eq:biophimpsi}
 \int \psi_{j,N}(x) \mathcal T_{N-1} \mathcal T_{N-2} \cdots \mathcal T_{1} \phi_{k,1}(x) {\rm d}  \mu_{N}(x)= \delta_{jk},
 \end{equation}
for $j,k=1, \ldots,n$.  To show this,    we first recall Andrei\'ef's identity: For any measure $\nu$ and $f_j,g_j\in \mathbb L_2(\nu)$ for $j=1, \ldots, n$ we have 
\begin{multline} \label{eq:adnr} \int \cdots \int \det \left(f_j(x_k)\right)_{j,k=1}^n \det \left(g_j(x_k)\right)_{j,k=1}^n {\rm d} \nu(x_1) \cdots {\rm d} \nu(x_n)\\= n!\det \left(\int f_j(x)g_k(x) {\rm d} \nu(x)\right)_{j,k=1}^n.
\end{multline}
 Then, by iterating \eqref{eq:adnr}, we see that $Z_n$ can be expressed as the determinant of the Gram-matrix associated to $\phi_{i,1}$ and $\psi_{j,N}$, i.e.
$$Z_n= (n!)^N \det \left(\int \psi_{i,N}(x) \mathcal T_{N-1} \mathcal T_{N-2} \cdots \mathcal T_{1} \phi_{j,1}(x) {\rm d}  \mu_{N}(x)\right)_{i,j=1}^n.$$
 Since $Z_n$ can not vanish, it means that Gram-matrix is non-singular. The fact of the matter is that   by linearity of the determinant, we can replace the $\phi_{j,1}$'s and $\psi_{k,N}$'s in the determinants in \eqref{eq:productmeasures} by any other linear combinations of those functions, as long as the resulting family is linearly independent. A particular choice, for example by using the singular value decomposition of the original Gram-matrix,  is to make sure that the new Gram-matrix  becomes the identity. In other words, without loss of generality we can assume that  we $\phi_{j,1}$ and $\psi_{k,N}$ are such that  \eqref{eq:biophimpsi} holds. This also shows that in that case $Z_n=(n!)^N$.


An important role in the analysis is played by the functions
\begin{equation} \label{eq:defphimpsim}
 \phi_{j,m}= \mathcal T_{m-1} \cdots \mathcal T_1 \phi_{j,1}, \qquad \psi_{j,m}= \mathcal T_m^* \cdots   \mathcal T_{N-1}^*  \psi_{j,N},
 \end{equation}
 for $m=1, \ldots, N$, where $\mathcal T_m^*$ stands for the adjoint of $\mathcal T_m$ which is given by   
$$\mathcal T_m^* f(x)= \int f(y) T_m(x,y) {\rm d} \mu_{m+1}(y)$$
Note that it follows from \eqref{eq:biophimpsi} that
 \begin{equation} \label{eq:biophimpsim}
 \int \phi_{j,m}(x) \psi_{k,m}(x) {\rm d}  \mu_{m}(x)= \delta_{jk},
 \end{equation}
for $j,k=1, \ldots,n$ and $m=1, \ldots, N$. 
The marginals in \eqref{eq:productmeasures} for the points $\{(m,x_{j,m})\}_{j=1}^n$ for a fixed $m$ are given by the measure 
\begin{equation}
\label{eq:BE}
\frac{1}{n!} \det \left( \phi_{j,m}(x_{m,k})\right)_{j,k=1}^n \det \left( \psi_{j,m}(x_{m,k})  \right)_{j,k=1}^n {\rm d} \mu_m(x_{1,m}) \cdots  {\rm d} \mu_m(x_{n,m}). 
\end{equation}
A probability measure of this type is known in the literature as a  biorthogonal ensemble \cite{BorBio}.

It is well-known that, by  the Eynard-Mehta Theorem,  measures of the form \eqref{eq:productmeasures} are examples of determinantal point processes. We recall that  a determinantal point process is a point process for which the $k$-point correlation functions can be written as  $k\times k$ determinants constructed out of a single function of two variables, called the correlation kernel.  More precisely, there exists a $K_{n,N}$ such that for any test function $g$ we have
\begin{multline}\label{eq:detstruct}
\mathbb E\left[\prod_{m=1,j=1}^{N,n} (1+ g(m,x_{m,j}))\right]\\
=\sum_{\ell=0} ^\infty \int_{\Lambda^\ell} g(\eta_1) \cdots g(\eta_\ell) \det \left(K_{n,N}(\eta_i,\eta_j)\right)_{i,j=1}^\ell  {\rm d} \nu(\eta_1) \cdots {\rm d} \nu(\lambda_\ell),
\end{multline}
where $\Lambda= \{1,2,\ldots, N\}\times \bbR$ and $\nu$ is a measure on $\Lambda$, called the reference measure. 
For the point process defined by \eqref{eq:productmeasures} this kernel has the form
\begin{multline}
\label{eq:generalformKn}
K_{n,N}(m_1,x_1,m_2,x_2)\\
= \begin{cases}
\sum_{k=1}^n \phi_{j,m_1}(x_1)\psi_{j,m_2}(x_2), & \text{if } m_1 \geq m_2,\\
\sum_{k=1}^n \phi_{j,m_1}(x_1)\psi_{j,m_2}(x_2)- T_{m_1,m_2}(x_1,x_2), & \text{if } m_1 <m_2,
\end{cases}
\end{multline}
with reference measure $\nu=\sum_{m=1}^N \delta_m\times \mu_m$. Here  $T_{m_1,m_2}(x_1,x_2)$ stands for the integral kernel for the integral operator $\mathcal T_{m_1}\mathcal T_{m_1+1}\cdots \mathcal T_{m_2-1}$ For more details and background on determinantal point process for extended kernels we refer to \cite{BorDet,Jdet,L,Sosh}.  

For a determinantal point process, all information is in one way or the other encoded in the correlation kernel. For that reason, a common approach to various results for determinantal point processes goes by an analysis of the kernel and its properties. However, in this paper we use an alternative approach for analyzing the global fluctuations.  We   follow the idea of \cite{BD} and assume that the biorthogonal families admit a recurrence.

\begin{assumption}\label{assumption}  We assume that $\{\phi_{j,1}\}_{j=1}^N$ and  $\{\psi_{j,N}\}_{j=1}^N$   can be extended to families  $\{\phi_{j,1}\}_{j=1}^\infty$ and $\{\psi_{j,N}\}_{j=1}^\infty$  such that   the functions defined by$$ \phi_{j,m}= \mathcal T_{m-1} \cdots \mathcal T_1 \phi_{1,m}, \qquad \psi_{j,m}= \mathcal T_m^* \cdots   \mathcal T_{N-1}^*  \psi_{j,N},$$ for $m=1, \ldots, N$,  have the properties
\begin{enumerate}
\item Biorthogonality: $$\int \phi_{j,N}(x) \mathcal T_{N-1} \cdots \mathcal T_1 \psi_{k,1}(x) {\rm d}\mu_N(x) = \delta_{jk},$$
for $j,k=1,2, \ldots$ 
\item  Recursion:   for each $m\in\{1, \ldots, N\}$ there is a  \emph{banded} matrix $\mathbb J_m$ such that 
\begin{equation}\label{eq:jacobirec} x \begin{pmatrix} \phi_{0,m} (x)\\  \phi_{1,m} (x)\\ \phi_{2,m} (x)\\ \vdots \end{pmatrix}= \mathbb J_{m} \begin{pmatrix} \phi_{0,m} (x)\\  \phi_{1,m} (x)\\ \phi_{2,m} (x)\\ \vdots \end{pmatrix}.
\end{equation}
\end{enumerate}
We will denote the bandwidth by $\rho$, i.e. 
$$(\mathbb J_m)_{k,l}= 0 , \text{     if   } |k-l| \geq \rho.$$ 
We will assume that $\rho$ does not depend on $m$ or $n$ (but $J_m$ may also depend on $n$).
\end{assumption} 
 
  Note that \eqref{eq:jacobirec} and the banded structure of $\mathbb J_m$ means that the functions $\phi_{j,m}$ satisfy finite term recurrence relation
$$x\phi _{k,m}(x)=\sum_{|j|\leq \rho} \left(\mathbb J_{m}\right)_{k,k+j} \phi_{k+j,m}(x).$$ The number of terms in the recurrence equals the number of non-trivial diagonals, which is at most $2 \rho+1$.  Also note that by biorthogonality we have    $$(\mathbb J_m)_{kl}= \int x \phi_{k,m}(x)\psi_{l,m}(x) {\rm d} \mu_m(x).$$
Finally, we mention that  although the arguments in this paper can likely by adjusted to allow for a varying bandwidth (but keeping the bandwidth uniformly bounded in $m,n,N$), in the relevant examples we always have a fixed bandwidth independent of $n,m$. 

An important special class of examples that we will study in this paper is when the biorthogonal families are related to  orthogonal polynomials.  If  each $\mu_m$ has finite moments, then we can define $p_{j,m}$ as the polynomial of degree $j$ with positive leading coefficient such that 
$$\int p_{j,m} (x) p_{k,m}(x) {\rm d} \mu_m(x) = \delta_{jk}.$$
As we will see in Section \ref{sec:example} in many examples we end up with  a measure \eqref{eq:productmeasures} with $\phi_{j,1}=c_{j,1} p_{j-1,N}$, $\psi_{j,N}=p_{j-1,N}/c_{j,N}$ and 
$$T_m(x,y) =\sum_{j=1}^\infty \frac{c_{j,m+1}}{c_{j,m}} p_{j-1,m}(x)p_{j-1,m+1}(y),$$
for some coefficients $c_{j,m} \neq 0$.  In that case, we find
\begin{equation}\label{eq:opchoice}
\phi_{j,m}(x)= c_{j,m} p_{j-1,m}, \qquad \text{and} \qquad \psi_{j,m}(x)= \frac{1}{c_{j,m}} p_{j-1,m} (x).
\end{equation}
Such examples satisfy Assumption \ref{assumption}. Indeed,  it is classical that the orthogonal polynomials satisfy a three-term recurrence
$$xp_{k,m}(x)=a_{k+1,m} p_{k+1,m}(x)+b_{k,m} p_{k,m}(x) + a_{k,m} p_{k-1,m}(x),$$ for coefficients $a_{k,m} >0$ and $b_{k,m} \in \bbR$.   This recurrence can be written as
\begin{equation} x \begin{pmatrix} p_{0,m} (x)\\  p_{1,m} (x)\\ p_{2,m} (x)\\ \vdots \end{pmatrix}= \mathcal  J_{m} \begin{pmatrix} p_{0,m} (x)\\  p_{1,m} (x)\\ p_{2,m} (x)\\ \vdots \end{pmatrix}.
\end{equation}
 The matrix $\mathcal J_{m}$ is then a symmetric tridiagonal matrix containing the recurrence coefficients, also called the Jacobi matrix/operator associated to $\mu_{m}$. It is not hard to check that in this situation, Assumption \ref{assumption} is satisfied with 
 \begin{equation} \label{eq:fromJacobitoJ}
 (\mathbb J_m)_{kl} =  \frac{c_{k,m}}{c_{l,m}} (\mathcal J_m)_{kl}.
 \end{equation}
 
 We stress that such a recurrence relation is not special for orthogonal polynomials only, but appear often in the presence of an  orthogonality condition. For instance, multiple orthogonal polynomial ensembles \cite{Kuijlaars}  also appear in the context of non-colliding processes, such  as external source models. Multiple orthogonal polynomials  satisfy recurrence relations  involving more terms then only three, see e.g. \cite{CV}.

Finally, note that in the example in the Introduction, it was needed to rescale the process as $n \to \infty$. This rescaling means that all the parameters will depend on $n$. \emph{Therefore  we will allow $\mu_m$,  $\phi_{j,m}$ and $\psi_{j,m}$ to depend on $n$, but for clarity reasons we will suppress this dependence in the notation.}

\subsection{Fluctuations of linear statistics for fixed $N$}

We will study linear statistics for the determinantal point process. That is, for a function $f:\{0,1,\ldots,N\} \times  \bbR \to \mathbb R$ we define
$$X_n(f)= \sum_{m=0}^N \sum_{j=1}^n f(m,x_{j,m})$$
where $\{(m,x_{j,m})\}_{j=1,m=1}^{n,N}$ are sampled from a probability measure of the form  \eqref{eq:productmeasures} satisfying Assumption \ref{assumption}. As we will see, the linear statistics $X_n(f)$ admit  a  useful representation  in terms of  the recurrence matrices $\mathbb J_m$.  In fact, one of the main points of the paper is that for studying linear statistics, this representation appears to be more convenient than the representation in terms of the correlation kernel $K_{n,N}$. In many interesting examples the asymptotic study of the relevant parts of $\mathbb J_m$  are trivial, whereas the  asymptotic analysis (in all the relevant regimes) of the kernel is usually  tedious.

The central observation of this paper, is that the fluctuation of the linear statistic depend strongly on just small part of the operators $\mathbb J_m$, More precisely, the coefficients $(\mathbb J_m)_{n+k,n+l}$ for fixed $k,l$ are dominant in the fluctuations for large $n$.  The other coefficients only play a minor role.  Two different models for which these coefficients behave similarly, have the same fluctuations. This is the content of the first main result. 
\begin{theorem}\label{th:main0}
Consider two probability measures of the form \eqref{eq:productmeasures} satisfying Assumption \ref{assumption} and denote expectations with $\mathbb E$ and $\tilde {\mathbb E}$ and the banded matrices by $\mathbb J_m$ and $\tilde {\mathbb J}_m$.
Assume that  for any $k, l \in \bbZ$  the sequence 
$\{(\tilde {\mathbb J}_{m})_{n+k,n+l}\}_n$ is bounded and
\begin{equation}
\label{eq:cond1main0}\lim_{n\to \infty} \left(\left(\mathbb J_{m}\right)_{n+k,n+l}-\left(\tilde { \mathbb J}_{m}\right)_{n+k,n+l}\right) = 0,
\end{equation}
Then for any function $f:  \{0,1,\ldots,N\}\times \bbR \to \mathbb R$  such that $f(m,x)$ is a polynomial in $x$, we have for $k \in \bbN$, 
\begin{equation}\label{eq:concthm0}
\EE\left[\left( X_n(f)-\EE X_n(f)\right)^k \right]-\widetilde \EE\left[\left( X_n(f)-\widetilde \EE X_n(f)\right)^k \right]\to 0,
\end{equation}
as $n \to \infty$. 
\end{theorem} 
The proof of this theorem will be given in Section \ref{sec:proofs}.

This result is a genuine universality result,  in the sense that  there is no specification of a limit. If two  families of models have the same asymptotic behavior of the recurrence matrices, then the fluctuations are also the same. As a consequence,  after computing the limiting behavior for a particular example, we obtain the same result for all comparable processes. 

The natural question is then what the typical behaviors are that one observes in the models of interest. As we will illustrate in Section \ref{sec:GFF}, one important example is when the recurrence coefficients have a limiting value or, more precisely, the matrices $\mathcal J_m$ have limits along the diagonals. The fluctuations in that case are described by the following theorem.

\begin{theorem}\label{th:main1}
Consider a probability measure of the form \eqref{eq:productmeasures} satisfying  Assumption \ref{assumption}. Assume that there exists $a_j^{(m)} \in \bbR$ such that  \begin{equation}\label{eq:mainthcond1}
\lim_{n\to \infty} \left(\mathbb J_{m}\right)_{n+k,n+l} = a_{k-l}^{(m)},
\end{equation}
for $k,l \in \bbZ$ and $m=1, \ldots, N$.
 Then for any function $f:\{1,\ldots,N\} \times \bbR \to \mathbb R$  such that $f(m,x)$ is a polynomial in $x$ , we have
\begin{multline}\label{eq:CLTfinite}
X_n(f)  - \bbE X_n(f)  \to\\
 N\left(0,2 \sum_{m_1=1}^N\sum_{m_2=m_1+1}^N  \sum_{k=1}^\infty k  f^{(m_1)}_k f^{(m_2)}_{-k}+\sum_{m=1}^N \sum_{k=1}^\infty k f^{(m)}_kf^{(m)}_{-k}\right),
\end{multline}
where 
\begin{equation}\label{eq:defpk}
 f_k^{(m)}= \frac{1}{2 \pi {\rm i}} \oint_{|z|=1} f(m,a^{(m)}(z)) \frac{{\rm d} z}{z^{k+1}}.
 \end{equation}
 and $a^{(m)}(z)= \sum_j a_j^{(m)} z^j$.
\end{theorem} 
\begin{remark}
Note that each $\mathbb J_m$ is banded and hence only finitely many $a_j^{(m)}$ are non-zero. In particular, each $a^{(m)}(z)$ is a Laurent polynomial. 
\end{remark}
The proof of this theorem will be given in Section \ref{sec:proofs}.

The latter result in particular applies when we are in the situation of orthogonal polynomials \eqref{eq:opchoice}. In that case, the following corollary will be particularly useful to us. 
\begin{corollary}\label{cor:th:main1}
Consider a probability measure of the form \eqref{eq:productmeasures} with  $\phi_{j,m}$ and $ \psi_{j,m}$ as in \eqref{eq:opchoice}.
Assume that for $k, \ell \in \mathbb Z$ with $|k-\ell | \leq 1$ we have \begin{equation}\label{eq:cormainthcond1}
\lim_{n\to \infty} \left(\mathcal J_{m}\right)_{n+k,n+l} = a_{|k-l|}^{(m)},
\end{equation}
for some $a_0^{(m)} \in \mathbb R$ and $a_1^{(m)} >0$ and 
\begin{equation}
\label{eq:cormainthcond2}
\lim_{n\to \infty} \frac{c_{n+k,m}}{c_{n+\ell,m}}= {\rm e}^{\tau_m(\ell-k)}.
\end{equation}
for some  $\tau_1 < \tau_2<\ldots < \tau_m$.
 Then for any function $f:\{1,\ldots,N\} \times \bbR \to \mathbb R$   such that $f(m,x)$ is a polynomial in $x$  we have
\begin{multline}\label{eq:variancesymmetric}
X_n(f)  - \bbE X_n(f)  \to
 N\left(0,\sum_{m_1=1}^N\sum_{m_2=1}^N  \sum_{k=1}^\infty k  {\rm e}^{-|\tau_{m_1}-\tau_{m_2}|k}\hat f^{(m_1)}_k \hat f^{(m_2)}_{k}\right),
\end{multline}
where 
\begin{equation}\label{eq:cordefpk}
\hat  f_k^{(m)}= \frac{1}{\pi } \int_0^\pi  f(m,a_0^{(m)}+2 a_1^{(m)} \cos \theta) \cos k \theta {\rm d} \theta.
 \end{equation}
\end{corollary} 
\begin{proof}
This directly follows from Theorem \ref{th:main1}, the relation \eqref{eq:fromJacobitoJ} and a rewriting of the limiting variance. For the latter, note that by a rescaling and a symmetry argument \eqref{eq:defpk} can be written as 
$$ f_k^{(m)}= {\rm e}^{-\tau_m k} \hat f_k^{(m)}.$$
Moreover, by  $  \hat f_k^{(m)}=  \hat f_{-k}^{(m)}$ and the fact that $\tau_{\ell} < \tau_k$ for $\ell <k$, we can symmetrize the limiting variance in Theorem \ref{th:main1} to obtain the one in  \eqref{eq:variancesymmetric}.
\end{proof}
As we will see in Section \ref{sec:example},  Theorem  \ref{th:main1} and Corollary \ref{cor:th:main1} cover several interesting examples. 

In case $N=1$ the determinantal point process reduces to the definition of a biorthogonal ensemble \cite{BorDet}.  In this situation, the above results are already  proved  by the author and Breuer \cite{BD}. In that paper, the approach using recurrence matrices was used for the first time and later used again in a mesoscopic analysis for orthogonal polynomials ensembles \cite{BDmes}. The results in \cite{BD} are a generalization of various earlier works in the determinantal setting and there is a vast amount of references on the subject. We only single out the influential work of   Johansson \cite{Jduke} on Unitary Ensembles (and  extensions to general $\beta$)  and refer to  \cite{BD} for further references.  However, much less is known in the case of $N>1$. To the best of the author's knowledge, the statement above are the first general results for multi-time or multi-layer linear statistics for determinantal point processes.

\begin{remark}
The conditions in both Theorem \ref{th:main0} and \ref{th:main1} can be relaxed. In fact, we only need the limits \eqref{eq:cond1main0}  along a subsequence $\{n_k\}_k$ to conclude \eqref{eq:concthm0} along that same subsequence. Similarly, for the limits in Theorem \ref{th:main1} and Corollary \ref{cor:th:main1}. For the case  $N=1$ and $\mathcal J$ the Jacobi operator associated with the orthogonal polynomials corresponding to the measure $\mu$, this relates the study of possible limit theorems  for the linear statistic to the study of right limits of the Jacobi operator.  For the interested reader we refer to the discussion in \cite{BD}.   However, in the present setup this generality seems less relevant. 
\end{remark}
\begin{remark}
The conditions in  both Theorem \ref{th:main0} and \ref{th:main1} are not sufficient to guarantee that a limit shape exists. That is, we do not know (nor need) the limit of $\frac1n \EE X_n(f)$. 
\end{remark}

\begin{remark}\label{rem:sym}
In Corollary \ref{cor:th:main1} it is easy to see that  the variance is positive. 
 In fact, in that case the variance can also be written in a different form that will be useful to us.  We recall the standard integral 
\begin{equation}\label{eq:fourierpoisson}
{\rm e}^{-k|\tau |} =\frac{1}{\pi}  \int_{-\infty}^{\infty} \frac{k}{k^2+\omega^2} {\rm e}^{-{\rm i} \omega \tau} {\rm d}\tau.
\end{equation}
By inserting this back into \eqref{eq:variancesymmetric} and a simple reorganization we see that the limiting variance can be written as 
\begin{equation}\label{eq:variancesymmetricpositivr}
\frac{1}{\pi} \sum_{k=1}^\infty \int_{-\infty}^\infty\left|\sum_{m=1}^N {\rm e}^{-{\rm i} \tau_m \omega} k f_k^{(m)}\right|^2\frac{{\rm d} \omega}{k^2+\omega^2}.
\end{equation}
This will be of use later on when we explain the connection of the above results  with the Gaussian Free Field. 

 In the general case, the limiting variance is of course also positive, but this is not evident from the expression due to the lack of symmetry. This feature is already present in the $N=1$ case, as discussed in \cite{BD}. 
\end{remark}

In the situation of Corollary \ref{cor:th:main1}  we can formulate natural conditions that allow us to extend Theorem  \ref{th:main1}  so that it holds for more general functions~$f$. In that case we will prove that the variance is continuous with respect to the $C^1$ norm. Hence we can try to extend the theorem to $C^1$ functions by polynomial approximation. For such an approximation it is convenient to work on a compact set. 
\begin{theorem}\label{th:extend}
Suppose all the conditions in Corollary \ref{cor:th:main1} hold and in addition there exists a compact set $E \subset \bbR$ such that either\\
(1) all supports $S(\mu_m^{(n)}) \subset E$ for $n \in \bbN$ and $m=1, \ldots, N$,\\
 or, more generally, \\
 (2) for every $k \in \mathbb N$ and $m=1, \ldots, N$, we have 
$$\int_{\bbR \setminus E}  |x|^k K_{n,N}(m,x,m,x) {\rm d} \mu_m(x)=o(1/n),$$
as $n \to \infty$. \\
Then the conclusion of Corollary \ref{cor:th:main1} also holds for  any $f: \{1,\ldots,N\} \times \bbR \to \mathbb R$  such that for $m \in \{1,\ldots,N\}$ the map $x\mapsto f(m,x)$ is a $C^1$ function  that grows at most polynomially at $\pm \infty$.
\end{theorem}
The proof of this theorem will be given in Section \ref{sec:proofs}.

The conditions in the theorem are  rather mild. In case of unbounded supports,  one can often show  in the classical situations, by standard asymptotic methods such as classical steepest descent arguments or Riemann-Hilbert techniques,  that the second condition is satisfied with exponentially small terms at the right-hand side, instead of only~$o(1/n)$. 
\subsection{Varying $N_n$}
Motivated by the example of non-colliding Brownian bridges in the Introduction, the natural question rises whether we can allow $N_n$ to depend on $n$ and such that $N_n \to \infty$.  Indeed, in that example we wanted to view the discrete sum \eqref{eq:intro:discr}  as a Riemann sum. Hence we will now consider probability measures of the form \eqref{eq:productmeasures} with $N=N_n$ and keep in mind that in many examples we have $T_m=P_{t_{m+1}-t_m}$ for some transition probability function~$P_t$ and sampling times $t_m$.

We start with a partitioning 
$$\alpha=t_0^{(n)}<t_1^{(n)} < t_2^{(n)}< \ldots < t_N^{(n)}< t_{N+1}^{(n)}= \beta,$$
of an interval $I= [\alpha, \beta] \subset \bbR$ such that $$\sup_{m} (t_{m+1}^{(n)}-t_m^{(n)})\to 0,$$ as $n \to \infty$.   And then, for a function on $g:I \times \bbR\to \bbR$‚ we ask for the equivalent statement of Theorems \ref{th:main0}, \ref{th:main1} and \ref{th:extend} for the linear statistic 
\begin{equation}\label{eq:varyinglinstat}
Y_n(g)= \sum_{m=1}^{N_n} \frac{1}{t_{m+1}^{(n)}-t_m^{(n)}}  \sum_{j=1}^n g(t_m^{(n)},x_{j,m}).
\end{equation}
  The first result is that Theorem \ref{th:main0} continues to hold when the limits \eqref{eq:cond1main0}  hold uniformly in $m$.
\begin{theorem}\label{th:main0growing}
Let $\{N_n\}_n$ be a sequence of integers such that $N_n \to \infty$ as $n \to \infty$.  Consider two probability measures of the form \eqref{eq:productmeasures} with $N=N_n$ and satisfying Assumption \ref{assumption} and denote the banded matrices by $\mathbb J_m$ and $\tilde {\mathbb J}_m$ for $m=1, \ldots, N_n$.

Assume that  for any $k, l \in \bbZ$  the set 
$\{(\tilde{ \mathbb  J}_{m})_{n+k,n+l}\}_{m=1,n=1}^{N_n,\infty}$ is bounded and
\begin{equation}\label{eq:growingextendcondition1}
\lim_{n\to \infty} \sup_{m=1,\ldots,N_n} \left|\left(\mathbb J_{m}\right)_{n+k,n+l}-\left(\tilde{\mathbb  J}_{m}\right)_{n+k,n+l}\right| = 0.
\end{equation}
Then for any function $g(t,x)$ such that polynomial $g\mapsto g(t,x)$ is a polynomial  in $x$ we have, for $k \in \bbN$ and $Y_n(g)$ as in \eqref{eq:varyinglinstat}, 
$$\EE\left[\left( Y_n(g)-\EE Y_n(g)\right)^k \right]-\widetilde \EE\left[\left( Y_n(g)-\widetilde \EE Y_n(g)\right)^k \right]\to 0,$$
as $n \to \infty$. 
\end{theorem} 

Also Theorem \ref{th:main1} has an extension to the varying $N_n$ setting.
\begin{theorem}\label{th:main1growing}
Let $\{N_n\}_n$ be a sequence of integers such that $N_n \to \infty$ as $n \to \infty$ and   suppose that for each $n$ we have a probability measures of the form \eqref{eq:productmeasures} with $N=N_n$ and satisfying Assumption \ref{assumption}. 

Assume that  there exists piecewise continuous functions $a_k(t)$  on the interval $I$ such that,  for $k,l \in \bbZ$,
$$\lim_{n\to \infty} \sup_{m=1,\ldots, N_n} \left|(\mathbb  J_m)_{n+k,n+l}-a_{k-l}(t_m^{(n)}) \right| =0.$$
Then for any function $g:I\times \bbR \to \bbR$  such that  $t \mapsto g(t,x)$ is piecewise continuous and $x\mapsto g(t,x)$ is a polynomial, we have that $Y_n(g)$ as defined in \eqref{eq:varyinglinstat} satisfies
$$ Y_n(g)-\EE Y_n(g) \to N\left(0,{\sum_{k=1}^\infty 2\iint_{\alpha<t_1<t_2<\beta} k g_k(t_1) g_{-k}(t_2) {\rm d} t_1 {\rm d} t_2} \right)$$
as $n \to \infty$, with 
$$g_k(t)=\frac{1}{2\pi {\rm i}} \oint_{|z|=1} g\left(t,\sum_{\ell} a_\ell(t)z^\ell\right) \frac{{\rm d}z}{z^{k+1}}.$$ 
\end{theorem}
As before, in the special case that we deal with orthogonal polynomials the latter theorem takes the following form.
\begin{corollary}\label{cor:main1growing}
 Let $a_0(t)$, $a_1(t)$ and $\tau(t)$ be  piecewise continuous functions on an interval $I$  and assume that $\tau(t)$ is increasing. 
 
Suppose that for each $n$ we have a probability measures of the form \eqref{eq:productmeasures} satisfying Assumption \ref{assumption} with $\phi_{j,m}$ and $\psi_{j,m}$ as in \eqref{eq:opchoice} and assume that for $k,l \in \bbZ$, for $k, \ell \in \mathbb Z $ with $|k-\ell| \leq 1$, we have
\begin{equation} \label{eq:limitsop1}
\lim_{n\to \infty} \sup_{m=1,\ldots, N_n} \left|(\mathcal J_m)_{n+k,n+\ell}-a_{|k-\ell|}(t_m^{(n)}) \right| =0,
\end{equation}
 and  
 \begin{equation} \label{eq:limitsop2}
 \lim_{n\to \infty}\sup_{m=1, \ldots,N_n}  \left|\frac{c_{n+\ell,m}}{c_{n+ k,m}}-{\rm e}^{\tau(t_m^{(n)})(k-\ell)} \right|= 0.
 \end{equation}
Then for any function $g:I\times \bbR \to \bbR$  such that  $t \mapsto g(t,x)$ is piecewise continuous and $x\mapsto g(t,x)$ is a polynomial, we have that $Y_n(g)$ as defined in \eqref{eq:varyinglinstat} satisfies
$$ Y_n(g)-\EE Y_n(g) \to N\left(0,{\sum_{k=1}^\infty \iint_{I \times I } {\rm e}^{-|\tau(t_2)-\tau(t_1)|k} k g_k(t_1) g_{-k}(t_2) {\rm d} t_1 {\rm d} t_2} \right)$$
as $n \to \infty$, with 
$$g_k(t)=\frac{1}{\pi } \int_0^\pi  g\left(t,a_0(t)+2 a_1(t) \cos \theta  \right) \cos k \theta{\rm d} \theta.$$ 
\end{corollary}
\begin{proof}
The proof follows from Theorem \ref{th:main1growing} in the same way as Corollary \ref{cor:th:main1} followed from \ref{th:main1}.
\end{proof}
\begin{remark}
By \eqref{eq:fourierpoisson} we can write the variance also as
\begin{equation}\label{eq:fouriertransformvariance}
\frac1\pi\sum_{k=1}^\infty  \int \frac{k^2}{\omega^2+k^2}   \left| \int_I  {\rm e}^{- {\rm i}  \omega \tau(t)} g_k(t){\rm d} t \right|^ 2 {\rm d} \omega.
\end{equation}
This will be useful later on. 
\end{remark}

Again Theorem \ref{th:main1growing} is stated for a function $g(t,x)$ that is a polynomial in $x$. Under similar conditions as in Theorem \ref{th:extend} we can extend this to a larger class of functions. 
\begin{theorem}\label{th:extendgrowing}
Assume that all the conditions of Corollary \ref{cor:main1growing} hold.  In addition, assume that    there exists a compact set $E \subset \bbR$ such that either\\
(1) all supports $S(\mu_m^{(n)}) \subset E$ for $n \in \bbN$ and $m=1, \ldots, N$,\\
 or, more generally, \\
 (2) for every $k \in \mathbb N$ and we have 
$$\sup_{m=1, \ldots,N_n} \int_{\bbR \setminus E}  |x|^k K_{n,N}(m,x,m,x) {\rm d} \mu_m(x)=o(1/n),$$
as $n \to \infty$. \\
Moreover, assume that $t_m^{(n)}$ are such 
$N_n \sum_{m=1}^{N_n} (t_{m+1}^{(n)}-t_m^{(n)})^2$
is bounded in $n$. Then Theorem \ref{th:main1growing} also holds wit for any $g$ such  that  $x\mapsto g(t,x)$ is a $C^1$ function  growing at most polynomially at $\pm \infty$.
\end{theorem}
\subsection{Connection to  Gaussian Free Field} \label{sec:GFF}
Finally, we discuss the relation of the above results with the Gaussian Free Field.  We will focus on the situation of Theorem \ref{th:main1growing} and such that one of the conditions in Theorem \ref{th:extendgrowing} is valid, such that Theorem \ref{th:main1growing} holds for continuously differentiable $g$.  

We will start by  recalling the definition of the Gaussian Free Field without a detailed justification. More details and background can be found in the survey \cite{Shef}. 

Let $\mathcal D$ be a simply connected domain in $\bbR^2$. With this domain we consider the space of test functions $\mathcal H_\nabla$ defined as follows: we start with space  $C^1_0(\mathcal D)$  of all continuously differentiable functions that vanish at the boundary of $\mathcal D$. On that space we define the norm 
$$\|\phi\|_\nabla^2 = \pi \iint_{\mathcal D} |\nabla \phi(w)|^2 {\rm d} m(w),$$
where ${\rm d} m$ stands for the planar Lebesgue measure on $\mathcal D$. The space of test function $\mathcal H_\nabla$ is then defined as the closure of $C^1_0(\mathcal D)$ with respect to this norm. Now that we have the space of test functions, we then define  the Gaussian free field to be the collection of Gaussian random variables $\langle F, \phi\rangle_\nabla$  indexed by $\phi \in \mathcal H_\nabla$ such that $$\langle F,\phi\rangle_\nabla\sim N(0,\|\phi\|_\nabla^2),$$
and such that $\phi\mapsto  \langle F, \phi\rangle_\nabla$ is linear. 

Now let us first focus on the example given in the introduction and let $h_n$ be the height function as defined in \eqref{eq:defheight}.  The statement now is that the fluctuations of $h_n-\EE h_n$ are described by the Gaussian free field \emph{in appropriately chosen coordinates}. That is, there is a simply connected domain $\mathcal D$ and a  homeomorphism $$\Omega : \mathcal D \to \mathcal E: w=(\tau,\theta) \mapsto (t,x),$$
 where $\mathcal E$ is the ellipse \eqref{eq:ellipse}, such that the push-forward of $h_n-\EE h_n$ under the map $\Omega$. That is, 
\begin{equation}
\label{eq:gffnonn} \langle h_n\circ \Omega, \phi\rangle_\nabla - \langle \EE h_n\circ \Omega, \phi \rangle _\nabla \to N(0,\|\phi\|_\nabla^2),
\end{equation}
as $n \to \infty$, for some natural pairing $\langle h_n\circ \Omega, \phi\rangle_\nabla $.  It is important to note that the Gaussian Free Field is a universal object, the coordinate transform  is not and depends on the specific problem at hand. 

The relation with linear statistics is explained as follows (see also \cite{D}), which also gives the precise form of the pairing  $\langle h_n\circ \Omega, \phi\rangle_\nabla $ that we will use.  First, by integration by parts and a change of variables we obtain
\begin{align*}
\pi  \iint_{\mathcal D} \nabla  h_n(\Omega(w)) \cdot \nabla\phi(w) {\rm d}m(w)  
= -\pi \iint_{\mathcal E} h_n(x,t)\Delta \phi(w(t,x)) \frac{ d (\tau,\theta)}{d(x,t)} &{\rm d} x {\rm d} t,
\end{align*}
where ${ d (\tau,\theta)}/{d(x,t)}$ stands for the Jacobian of the map $\Omega^{-1}$. 
We then use the  fact that $h_n(t,x) = \sum_{j=1}^n \chi_{(-\infty,\gamma_j(t)]} (x) $  to rewrite the right-hand side as
\begin{align*}
   -\pi \int_I \sum_{j=1}^n  \int_{-\infty}^{\gamma_j(t)} \Delta \phi(w(t,x)) \frac{ d (\tau, \theta)}{d(x,t)} &{\rm d} x {\rm d}t
\end{align*}
Finally, the pairing $\langle h_n , \phi\rangle_\nabla$ is then defined  by a discretization of the  integral over $t$,
 \begin{align}\label{eq:gffparingdiscreeet}
\langle h_n , \phi\rangle_\nabla=  - \sum_{m=1}^N \frac{1}{t_{m+1}-t_m}  \sum_{j=1}^n \pi \int_{-\infty}^{\gamma_j(t_m^{(n)})}\Delta \phi (w(t_m,x))  \frac{ d (\tau,\theta)}{d(x,t)} {\rm d} x.
 \end{align}
Now note that $\langle h_n , \phi\rangle_\nabla=Y_n(g)$ where $Y_n(g)$ is the linear statistic  as in \eqref{eq:varyinglinstat} with  $$g(t,y)= -\pi \int_{-\infty}^{y}\Delta \phi(w(t,x) ) \frac{ d (\tau, \theta)}{d(x,t)} {\rm d}x.$$ Hence the pairing of the height function with a test function, reduces to a linear statistic  for the point process $\{\gamma_j(t_m)\}_{j=1,m=1}^{n,N}$ and we can apply Proposition \ref{prop:intro} to find its limiting fluctuations, which leads to \eqref{eq:gffnonn} as we will show below. 
 
We will state the result in the more general setup of Theorem  \ref{th:extendgrowing}. That is we consider a probability measure of the form in \eqref{eq:productmeasures} satisfying Assumption \ref{assumption} in the orthogonal polynomial situation \eqref{eq:opchoice}. We then assume that there exists interval $I=[\alpha, \beta]$  and functions $a_0:I \to \bbR$, $a_1: I \to \bbR$ and an increasing function $\tau(t)$ such that we have the limits \eqref{eq:limitsop1} and \eqref{eq:limitsop2} for some partitioning $\{t_m^{(n)}\}$ of $I$. This gives us the random points $\{(m,x_{j,m})\}_{j=1,m=1}^{n,N_n}$ and we ask for the fluctuations of  the height function defined by 
$$h(t_m^{(n)},x)= \# \{ j \mid x_{j,m} \leq x\}.$$
 The coordinate transform  is now constructed as follows. First, define
$$\mathcal E= \left\{(t,x) \mid  -2 a_1(t)  \leq x-a_0(t) \leq 2 a_1(t)\right\}.$$
 Note that $\tau(t)$ is strictly increasing as a function of $t$ and hence it has an inverse $t(\tau)$. Then, with
$$\mathcal D= \{ (\tau, \theta) \mid \tau \in  (\tau(\alpha),\tau(\beta)), \, \theta \in (0,\pi)\},$$
the map
$$\Omega : \mathcal D \to \mathcal E : (\tau,\theta) \mapsto (t(\tau),x(\tau,\theta))= \left(t(\tau),  2  a(t(\tau)) \cos \theta\right),$$
is a bijection  and has the  inverse
$$\Omega^{-1} : \mathcal E \to \mathcal D : (t,x) \mapsto (\tau(t), \theta(t,x))=\left ( \tau(t), \arccos \frac{x}{2a(t)}\right).$$
In this setting, we have that the push-forward by $\Omega$ of the fluctuations of the height function $h_n$ are governed by   the Gaussian Free Field on $\mathcal D$ in the following sense. 
\begin{theorem}\label{prop:GFF} Let $\{(m,x_{j,m})\}_{j,m=1}^{n,N}$ be random from a probability measure of the form \eqref{eq:productmeasures} satisfying the conditions in Theorem \ref{th:extendgrowing} with the parameters as described above.

Let $\phi$ be a twice continuously differentiable  real-valued function with compact support in $\mathcal D$ and consider the pairing 
$$\langle h_n , \phi\rangle_\nabla:=  - \sum_{m=1}^N  \frac{1}{t_{m+1}^{(n)}-t_m^{(n)}}\sum_{j=1}^n \pi \int_{-\infty}^{x_{j,m}}\Delta \phi (w(t,x))  \frac{ d (\tau,\theta)}{d(x,t)} {\rm d} x.$$
Then,  as $n \to \infty$, 
$$\langle h_n , \phi\rangle_\nabla-\EE \left[\langle h_n , \phi\rangle_\nabla\right]\to N(0,\|\phi\|_\nabla^2),$$
in distribution. 
\end{theorem}
\begin{proof} We start by recalling that $\langle h_n , \phi\rangle_\nabla$ is the linear statistic $Y_n(g)$ as in \eqref{eq:varyinglinstat} 
with 
 $$g(t,y)= -\pi \int_{-\infty}^{y} \Delta \phi(w(t,x) ) \frac{ d (\tau, \theta)}{d(x,t)} {\rm d}x.$$
Note that $x\mapsto g(t,x)$  is a continuously differentiable and bounded function. Moreover, the point process satisfies the assumptions in Theorem \ref{th:extendgrowing} so that, as $n \to \infty$, 
$$\langle h_n , \phi\rangle_\nabla-\EE \left[\langle h_n , \phi\rangle_\nabla\right]\to N(0,\sigma(g)^2),$$
in distribution, with (see Remark \ref{rem:sym})
$$
\sigma(g)^2=\frac{1}{\pi} \sum_{k=1}^\infty \int_\bbR \frac{k^2}{\omega^2+k^2} \left|\int_I {\rm e}^{-{\rm i} \tau(t) \omega} g_k(t) {\rm d} t \right|^2 {\rm d} \omega.$$
and 
$$ g_k(t) =\frac{1}{\pi} \int_0^\pi g(t,a_0(t)+2 a_1(t) \cos \theta)  \cos k \theta {\rm d} \theta.$$
It remains to show that we can rewrite the variance so that it matches with the one in the statement. 

We start by  noting that the Jacobian for the map is given by 
$$\frac{{\rm d} (\tau,\theta)}{{\rm d} (t, x)}= \tau'(t) \frac{\partial \theta}{\partial x}.$$
Then, by a change of variables we have
\begin{multline*}
g(t,y)= \pi\int_{a_0(t)-2 a_1(t)}^y \Delta \phi(\tau(t),  \theta(t,x) ) \frac{{\rm d} (\tau,\theta)}{{\rm d} (t, x)} {\rm d} x\\
= \pi\tau'(t)\int_{a_0(t)-2 a_1 (t)}^y \Delta \phi(\tau(t),  \theta(t,x) )
 \frac{\partial \theta}{\partial x} {\rm d} x 
=\pi\tau'(t) \int_0^{\theta(t,y)}\Delta \phi(\tau(t),  \theta) )
 {\rm d} \theta.
\end{multline*}
Clearly, since $\Omega$ and $\Omega^{-1}$ are each others inverse maps, we have $$\theta(t,a_0(t)+2  a_1(t) \cos(t))= \theta.$$
and hence
$$g(t,a_0(t)+2 a_1 (t) \cos \theta)= \pi \tau'(t)\int_0^{\theta} \Delta \phi(\tau(t),\tilde \theta) {\rm d} \tilde \theta.$$

This implies that
\begin{multline*}
k g_k(t) = \frac{1}{\pi} \int_0^\pi g(t,a_0(t)+2 a_1(t) \cos  \theta) k \cos k \theta {\rm d} \theta
\\=
  {\tau'(t)} \int_0^\pi \int_0^{\theta} \Delta \phi(\tau(t),\tilde \theta) {\rm d} \tilde \theta \  k \cos k \theta  {\rm d} \theta
 =   {\tau'(t)} \int_0^\pi  \Delta \phi(\tau(t), \theta)   k \sin k \theta {\rm d} \theta
\end{multline*}
using integration by parts in the last step. We then continue by inserting the last expressions and using the fact that $\phi$ has compact support in $\mathcal D$, 
\begin{multline*}
k\int_I {\rm e}^{-{\rm i} \tau(t) \omega} g_k(t) {\rm d} t
= \int_I  \int_0^\pi{\rm e}^{-{\rm i} \tau(t) \omega}  \Delta \phi(\tau(t), \theta)   k \sin k \theta \, {\rm d} \theta \,\tau'(t)\, {\rm d} t\\
=  \int_\bbR \int_0^\pi{\rm e}^{-{\rm i} \tau \omega}  \Delta \phi(\tau, \theta)   k \sin k \theta \, {\rm d} \theta  \, {\rm d} \tau =\pi \left( \mathcal G \Delta \phi\right)(\omega,k),
\end{multline*}
where $\mathcal G$ is the operator
$$\mathcal G f(\omega,k) = \frac{1}{\pi} \int_\bbR \int_0^\pi f(\tau,\theta) {\rm e}^{-{\rm i} \omega \tau} \sin k \theta \, {\rm d} \theta\, {\rm d} \tau.$$
Since $\{\sqrt{\frac{2}{\pi}} \sin k \theta\}_{k=1}^\infty$ is an orthonormal basis for $\mathbb L_2([0,\pi))$ and since the integral over $\tau$ is the usual Fourier transform (with normalization $(2 \pi)^{-1/2}$), we see that $\mathcal G$ defines a unitary transform from $\mathbb L_2(\mathcal D)$ to $\mathbb L_2(\bbR) \times \ell_2(\bbN)$. It is also easy to check that $\mathcal G \Delta \phi(\omega,k)=- (\omega^2+k^2) \mathcal G \phi(\omega,k)$. We then apply Plancherel's Theorem  to write
\begin{multline*}
\frac{1}{\pi} \sum_{k=1}^\infty \int_\bbR \frac{k^2}{\omega^2+k^2} \left|\int_I {\rm e}^{-{\rm i} \tau(t) \omega} g_k(t) {\rm d} t \right|^2 {\rm d} \omega
=  \pi \sum_{k=1}^ \infty \int_\bbR \frac{\left|(\mathcal G\Delta \phi)(\omega,k)\right|^2}{\omega^2+k^2}  {\rm d} \omega\\
-\pi \sum_{k=1}^ \infty \int_\bbR \overline{\mathcal G\Delta \phi(\omega,k)}\mathcal G \phi(\omega,k)  {\rm d} \omega
=- \pi \iint_{\mathcal D} \phi \overline{\Delta \phi} =  \pi \iint_{\mathcal D} \phi {\Delta \phi}= \|\phi\|_\nabla^2,
\end{multline*}
and this proves the statement. 
\end{proof}

\section{Examples}\label{sec:example}
In this section we will illustrate the main results by discussing several examples.

\subsection{Stationary non-colliding processes}

The first class of examples  is that of non-colliding processes for which the classical orthogonal polynomials ensembles are the invariant measures. The construction we will follow is  a well-known approach using Doob's $h$-transform and the Karlin-McGregor Theorem, see e.g. \cite{K} for a discussion.  An alternative way of defining the  processes is to start with a generator for a single particle process and then define an $n$-particle process by  constructing  a generator on the space of symmetric functions~\cite{BO,O}. 

Suppose we are given a  Markov process for a single particle and let us assume that it has a transition function $P_t(x,y) {\rm d} \mu(y)$ on a subset $E\subset \mathbb R$ that can be written as
\begin{equation}\label{eq:diffusion}
P_t(x,y) {\rm d} \mu(y)= \sum_{j=0}^\infty {\rm e}^{-\lambda_j t} p_j(x)p_j(y) {\rm d} \mu(y),
\end{equation}
where
$$ 0 = \lambda_0 < \lambda_1< \lambda_2< \ldots $$
and  $p_j(x)$ are orthogonal polynomials with respect to ${\rm d} \mu(y)$. That is, $\{p_k\}_k$ is the unique family of polynomials such that  $p_k$ is a polynomial of degree $k$ with positive leading coefficient and 
$$\int p_k(x) p_\ell(x) {\rm d} \mu(x)= \delta_{k\ell}. $$  In other words, we assume that the generator for the Markov process has eigenvectors $p_j$ and eigenvalues $\lambda_j$.  It is standard that the classical orthogonal polynomials appear in this way, as we will see.

We then  construct  a Markov process on the Weyl chamber 
$$\mathcal W_n= \{(x_1, \ldots, x_n) \in \bbR^n \mid x_1<\ldots <x_n\}.$$ 
First we note that by a theorem of Karlin-McGregor  it follows, under general conditions on the Markov process, that the joint probability distribution for the positions $y_j$ after time $t>0$ of particles  that (1) each perform a  single particle process given by $P_t$, (2) start at $x_1, \ldots, x_n$ and  (3) are conditioned not to  collide in $[0,1]$, is given by 
$$\det \left(P_t(x_i,y_j)\right).$$
Then by \eqref{eq:adnr} and \eqref{eq:diffusion} it follows that 
\begin{multline*} \int _{\mathcal W_n} \det \left(P_t(x_i,y_j)\right)\det \left( p_{j-1}(y_k)\right)_{j,k=1}^n  {\rm d} \mu(y_1) \cdots {\rm d} \mu(y_n)\\
={\rm e}^{-t \sum_{j=0}^{n-1}  \lambda_j}  \det \left(p_{j-1}(x_k)\right)_{j,k=1}^n.
\end{multline*}
Moreover,  $\det \left(p_{j-1}(x_k)\right)_{j,k=1}^n= c \prod_{1 \leq i < j \leq n} (x_j-x_i)$ is positive.  Hence it is a positive harmonic function and we can apply Doob's $h$-transform to arrive at the transition function
$$
P_t(\vec x,\vec y)=
{\rm e}^{t \sum_{j=0}^{n-1}  \lambda_j} \det \left(P_t(x_i,y_j)\right)_{i,j=1}^n  \frac{\det \left( p_{j-1}(y_k)\right)_{j,k=1}^n} {\det \left(p_{j-1}(x_k)\right)_{j,k=1}^n}.
$$
This defines the  Markov process on $\mathcal W_n$ that we will be interested in. 
Finally,  it is not hard to show from \eqref{eq:adnr} that the unique invariant measure is given by the orthogonal polynomial ensemble \cite{K},
$$   \left(\det \left(p_{j-1}(x_k)\right)_{j,k=1}^n \right)^2  {\rm d} \mu(x_1) \cdots {\rm d} \mu(x_n).$$
In other words, the above construction provides a way for defining a stochastic dynamics for which the classical orthogonal polynomial ensembles are the invariant measures.

We consider this Markov process in the stationary situation. That is, we fix $t_1<t_2<\ldots< t_N \in \bbR$ and start the Markov process with the invariant measure at $t_1$. Then  we obtain a probability measure for the locations at $t_m$ 
\begin{multline*}
\det \left(p_{j-1}(x_k(t_1)\right)_{j,k=1}^n  \prod_{m=1}^{N-1}\det \left(P_{t_{m+1}-t_m}(x_j(t_m),x_k(t_{m+1}))\right)_{j,k=1}^n \\ \times \det \left(p_{j-1}(x_k(t_N)\right)_{j,k=1}^n  \prod_{m=1}^N \prod_{j=1}^n {\rm d} \mu(x_j(t_m)),
\end{multline*} 
which, after  a symmetrization, is exactly of the form \eqref{eq:productmeasures}. In fact it is an example of an orthogonal polynomial situation as given in \eqref{eq:opchoice} with $c_{j,m}= {\rm e}^{-t_m \lambda_j}$. Before we apply Corollary \ref{cor:th:main1} we recall that the orthogonal polynomials on the real line are subject to a three term recurrence relation 
$$xp_k(x) =a_{k+1} p_{k+1}(x) + b_k p_k(x)+a_k p_{k-1}(x),$$
for some numbers $b_k \in \mathbb R$ and $a_k>0$. We recall that we allow the measure $\mu$ to vary with $n$ so that also $b_k$ and $a_k$ may vary with $n$ and hence we will write $a_k= a_{k}^{(n)}$ and $b_k^{(n)}$.

\begin{theorem}\label{th:OP1}
Suppose that, for some $a>0$ and $n \in \bbR$ we have 
$$\begin{cases}
a_{n+k}^{(n)} \to a, \\
b_{n+k}^{(n)} \to b,
\end{cases}$$
and choose $t_j$ such that, for some $\tau_j$ and sequence $\kappa_n$, 
$$\kappa_n (\lambda_{n+k}-\lambda_{n+l})t_j \to (k-l)\tau_j $$
as $n \to \infty$. Then for any $f: \{1, \ldots, N\} \times \bbR$ such that $f(m,x)$ is a polynomial, we have
$$X_n(f)-\EE X_n(f) \to N\left(0,\sum_{m_1,m_2=1}^N  \sum_{k=1}^\infty k  {\rm e}^{-|\tau_{m_1}-\tau_{m_2}|k}f^{(m_1)}_kf^{(m_2)}_{k}\right),$$
where 
$$ f_k^{(m)}= \frac{1}{ \pi } \int_0^\pi f\left(m,2 a  \cos \theta+b\right) \cos k  \theta {\rm d} \theta.$$
Moreover,   for any $g: I \times \bbR\to \bbR$ such that $g(t,x)$ is a polynomial in $x$ we have that $Y_n(g)$ as defined in \eqref{eq:varyinglinstat} satisfies
$$ Y_n(g)-\EE Y_n(g) \to N\left(0,{\sum_{k=1}^\infty \iint_{I \times I } {\rm e}^{-|\tau(t_2)-\tau(t_1)|k} k g_k(t_1) g_{-k}(t_2) {\rm d} t_1 {\rm d} t_2} \right)$$
as $n \to \infty$, with 
$$g_k(t)=\frac{1}{\pi } \int_0^\pi  g\left(t,a_0+2 a_1 \cos \theta  \right) \cos k \theta{\rm d} \theta.$$ 
\end{theorem}
\begin{proof}
This is a direct consequence of Corollaries \ref{cor:th:main1} and \ref{cor:main1growing}, with $c_{j,m}={\rm e}^{-t_m \lambda_j}$.
\end{proof}
The point is now that for the Markov process related to the classical polynomials we can easily verify the stated condition by looking up the explicit values of the parameters in standard reference works on classical orthogonal polynomials, such as \cite{koekoek}.

\begin{figure}
\begin{center}
\includegraphics[scale=.4]{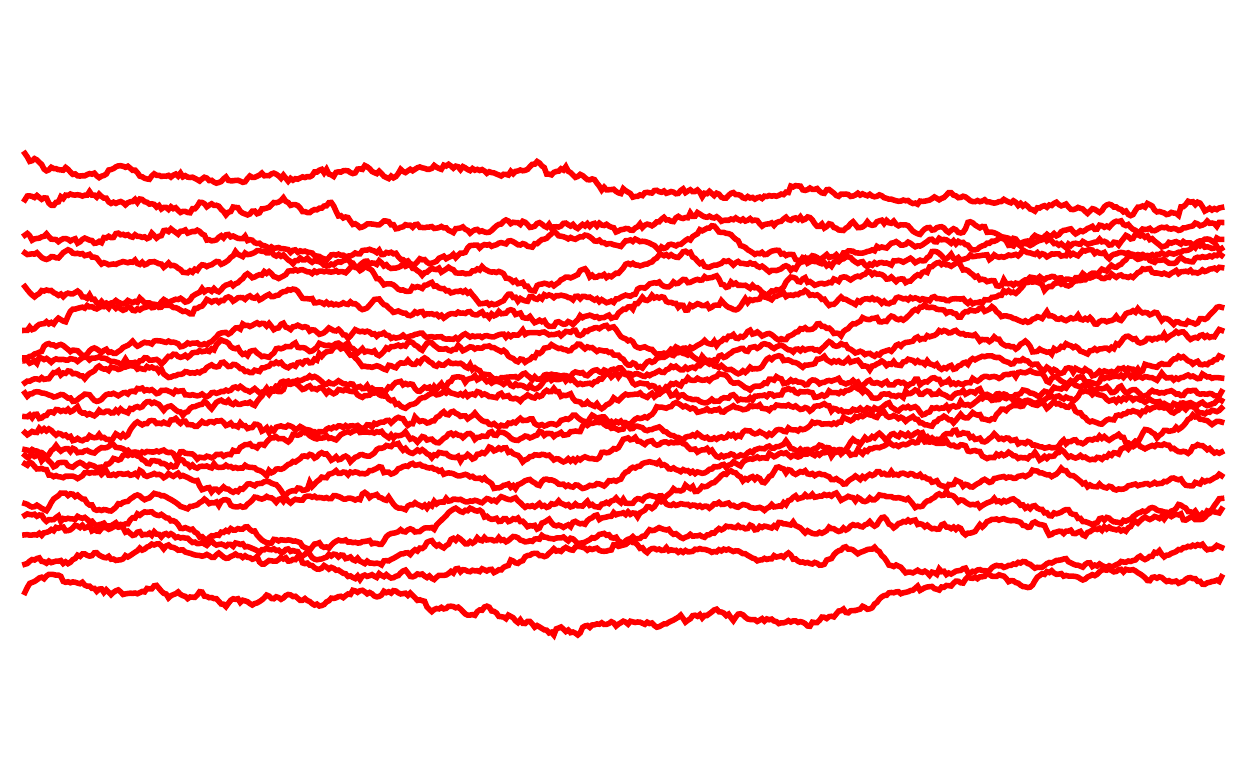}\qquad 
\includegraphics[scale=.4]{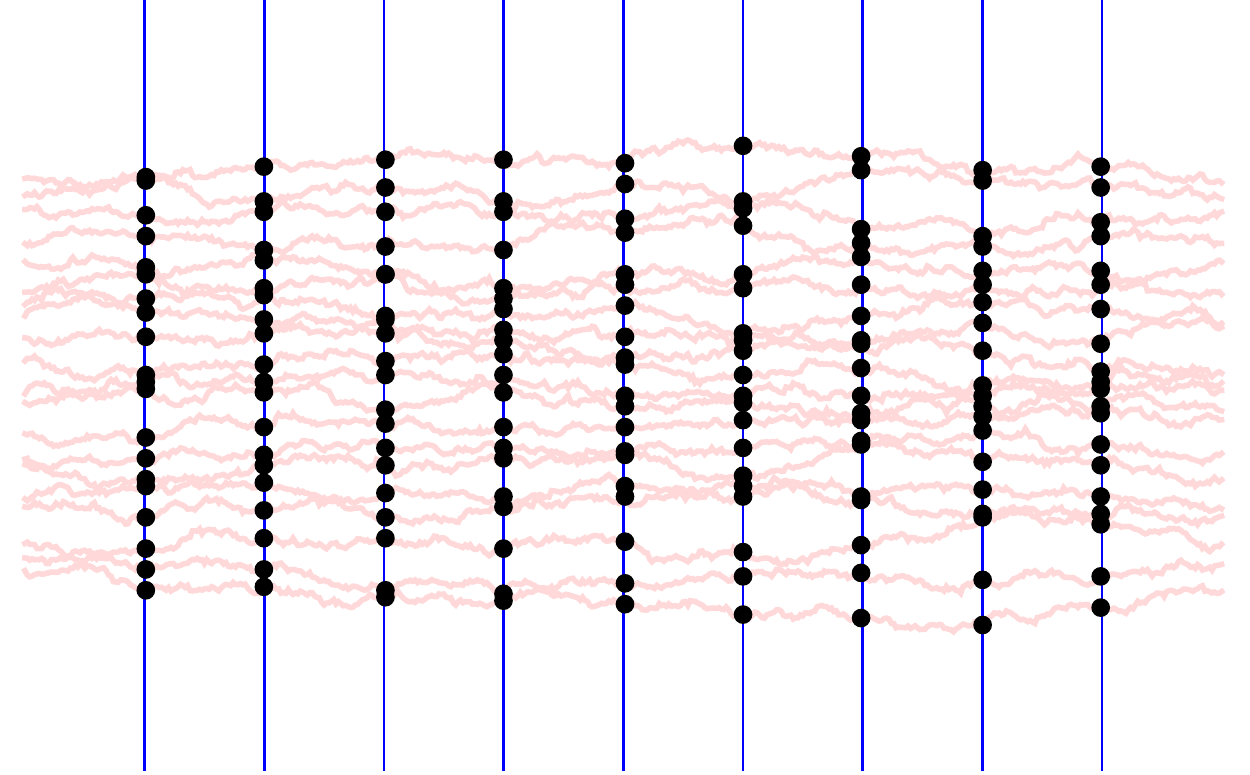}
\end{center}
\caption{The left figure shows a sampling from a stationary non-colliding process generated by the Ornstein-Uhlenbeck process of size $n=20$. At the right we intersect the trajectories at multiple times $t_m$.}
\end{figure}
We will now illustrate the results with some examples. To start with, we consider the classical Hermite, Laguerre and Jacobi polynomials, which are well-known to be  eigenfunctions for a second order differential operator that can be used as the generator for the Markov process. 
\begin{example}[Non-colliding Ornstein-Uhlenbeck processes] \label{eq:OU} Let us start where we take $P_t(x,y)$ according to the Ornstein-Uhlenbeck process. This is the model that was considered by Dyson \cite{Dyson} for $\beta=2$. In that case we have 
$$ P_t(x,y) {\rm d} \mu(y)= \frac{1}{\sqrt{2 \pi (1-{\rm e}^{-2t})}} {\rm e}^{-\frac{({\rm e}^{-t}x-y)^2}{2(1-{\rm e}^{-2t})}}{\rm d} y,$$
as the transition function. By  Mehler's formula for the Hermite polynomials this can be expanded as
$$ P_t(x,y) {\rm d} \mu(y)=\sum_{j=0}^\infty{\rm e} ^{ -j t}    H_j(x) H_j(y)  {\rm e}^{-y^2/2} {\rm d} y, $$
where $H_j(y)$ are the normalized   Hermite polynomials where the  orthogonality is with respect to ${\rm e}^{-y^2/2}{\rm d} y$ on $\bbR$. 

The Hermite polynomials satisfy the recurrence 
$$x H_k(x)= \sqrt{k+1} H_{k+1}(x)+\sqrt k H_{k-1}(x).$$
The recurrence coefficients grow and in order to get a meaningful  limit we need to rescale the process.  Indeed, when $n \to \infty$ the paths at any given time $t$ fill the interval $(-2 \sqrt n,2 \sqrt n)$ and we rescale the space variable and introduce the new variable $\xi$ by 
$$x= \sqrt{n} \xi.$$
Then the rescaled orthonormal  polynomials  are $p_k(\xi)=n^{-1/4} H_k(\sqrt n \xi)$ and for these polynomials we have the following recursion
$$p_k(\xi)= \sqrt\frac {k+1}{n} p_{k+1}(\xi)+ \sqrt\frac{k}{n} p_{k-1}(\xi).$$
One readily verifies that  
 $$a_{n+k,n}^{(n)} \to 1, \qquad b_{n+k,n}^{(n)}=0.$$
Moreover, since $\lambda_j=j$,  we also have 
$$(\lambda_{n+k} -\lambda_{n+\ell}) t_j= (k-\ell)t_j.$$
Therefore, the conditions of Theorem \ref{th:OP1} are satisfied with $a=1,b=0$ and $\tau_j=t_j$. In fact, for the Hermite polynomials one can verify that the second condition in Theorem \ref{th:extend} is satisfied and hence Theorem \ref{th:OP1} also holds for function $f: \{1, \ldots, N\} \times \bbR$ such that $x \to f(m,x)$ is a continuously differentiable function that grows at most polynomially  at $\pm \infty$.  This follows for example after a classical steepest decent analysis on the integral representation of the Hermite polynomials or by a Riemann-Hilbert analysis. We will leave the tedious details to the reader.

Finally, we note that the non-colliding brownian bridges model from the Introduction can be obtained from the above model after the change of variables 
$$
\begin{pmatrix}
t\\
\xi
\end{pmatrix}
 \mapsto \begin{pmatrix}
 \frac{1}{1+ {\rm e}^{-t}}\\
 \frac{\xi}{\cosh t}
\end{pmatrix}. $$ This is discussed in, for example, \cite{Jhahn} and we refer to that paper for more details. This also proves Proposition \ref{prop:intro}. \hfill $\blacksquare$
\end{example}

\begin{example}[Non-colliding squared radial Ornstein-Uhlenbeck processes] In the next example, we replace the Ornstein-Uhlenbeck process with its squared radial version. That is,
$$P_t(x,y) {\rm d} \mu(y)=\frac{1}{1-{\rm e}^{-t}}\left(\frac{y}{x {\rm e}^{-t}}\right)^{r/2} {\rm e}^{-\frac{x {\rm e}^{-t}}{1-{\rm e}^{-t}}-\frac{y}{1-{\rm e}^{-t}}} I_r\left( \frac{2 \sqrt{ {\rm e}^{-t}xy}}{1-{\rm e}^t} \right)$$
on $[0,\infty)$ where $r>-1$ is a parameter and $I_r$ stands for the modified Bessel function of the first kind of order $r$. The squared radial Ornstein-Uhlenbeck process is related to the squared Bessel process in a similar way as the Ornstein-Uhlenbeck process is to Brownian motion. Indeed, the squared Bessel process can be obtained by a change of variables. The latter process has been studied in the literature in the context of non-colliding processes before. In \cite{KOC} it was used to define  a dynamic version of the Laguerre ensemble from Random Matrix Theory. 

To see its connection we note that we can expand the transition function as \cite{Schoutens}
$$ P_t(x,y) {\rm d} \mu(y)=\sum_{j=1}^\infty  {\rm e} ^{ -j  t}  \frac{j!}{\Gamma(j+r+1)}L_j^{(r)} ( x) L_j^{(r)}( x){y^{r}{\rm e}^{- y}} {\rm d} y,$$
where $L^{(r)}_j(x)$ is the  generalized Laguerre polynomial of degree $r$ (with orthogonality with respect to $y^{r} {\rm e}^{-y}{\rm d} y$). These polynomials satisfy the recursion
$$x L_k^{(r)}=-(k+1) L_{k+1}^{(r)}(x)+(2k+r+1)L^{(r)}_k(x)-(k+r) L_{k-1}^{(r)}(x).$$
Note that the recursion coefficients are growing, which means that  we need to rescale the process. Moreover, the Laguerre polynomials are not normalized.  To be able to apply Theorem \ref{th:OP1} we therefore define the normalized  and rescaled polynomials by 
$$ p_k(\xi)=\sqrt{\frac{k!}{\Gamma(k+r+1)}}L_k^{(r)}(n\xi)(-1)^k.$$
The $p_k$ then satisfy the recursion 
$$\xi  p_k(\xi)=\frac{\sqrt{(k+1)(k+r+1)}}{n}  p_{k+1}(\xi)+\frac{(2k+r+1)}{n} p_k(\xi)+ \frac{\sqrt{k(k+r)}}{n} \ p_{k-1}(\xi).$$
Then one readily verifies that 
$$a_{n+k,n}^{(n)} \to 1, \qquad b_{n+k,n}^{(n)} \to 2.$$
Moreover, as in the previous example we have $\lambda_j=j$,  and hence also
$$(\lambda_{n+k} -\lambda_{n+\ell}) t_j= (k-\ell)t_j.$$
Therefore, the conditions of Theorem \ref{th:OP1} are satisfied with $a=1,b=2$ and $\tau_j=t_j$. To the best of our knowledge we believe that this example is a new result that has not appeared before. 

Also in the case it is possible to prove the conditions of Theorems \ref{th:extend} and \ref{th:extendgrowing}. \hfill $\blacksquare$\end{example}

\begin{figure}
\centering
\begin{subfigure}{}
\includegraphics[scale=.4]{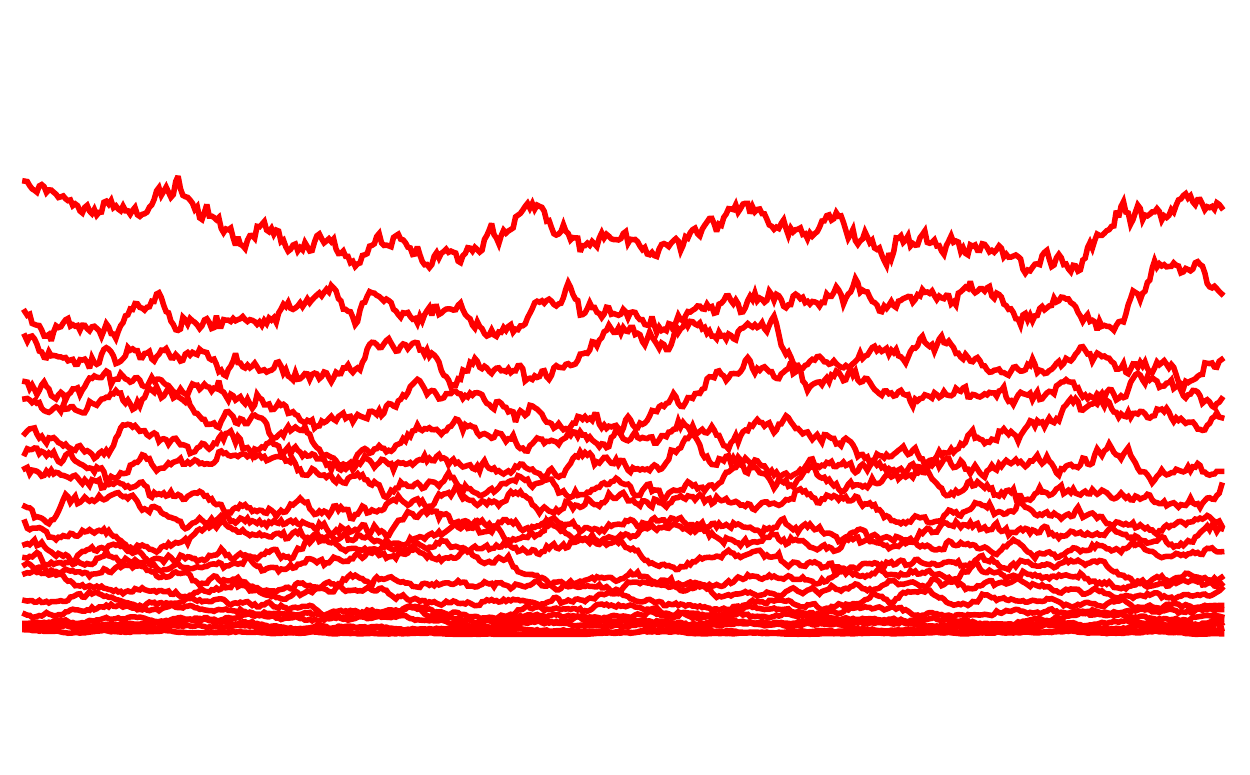}
\end{subfigure}
\begin{subfigure}{}
\includegraphics[scale=.4]{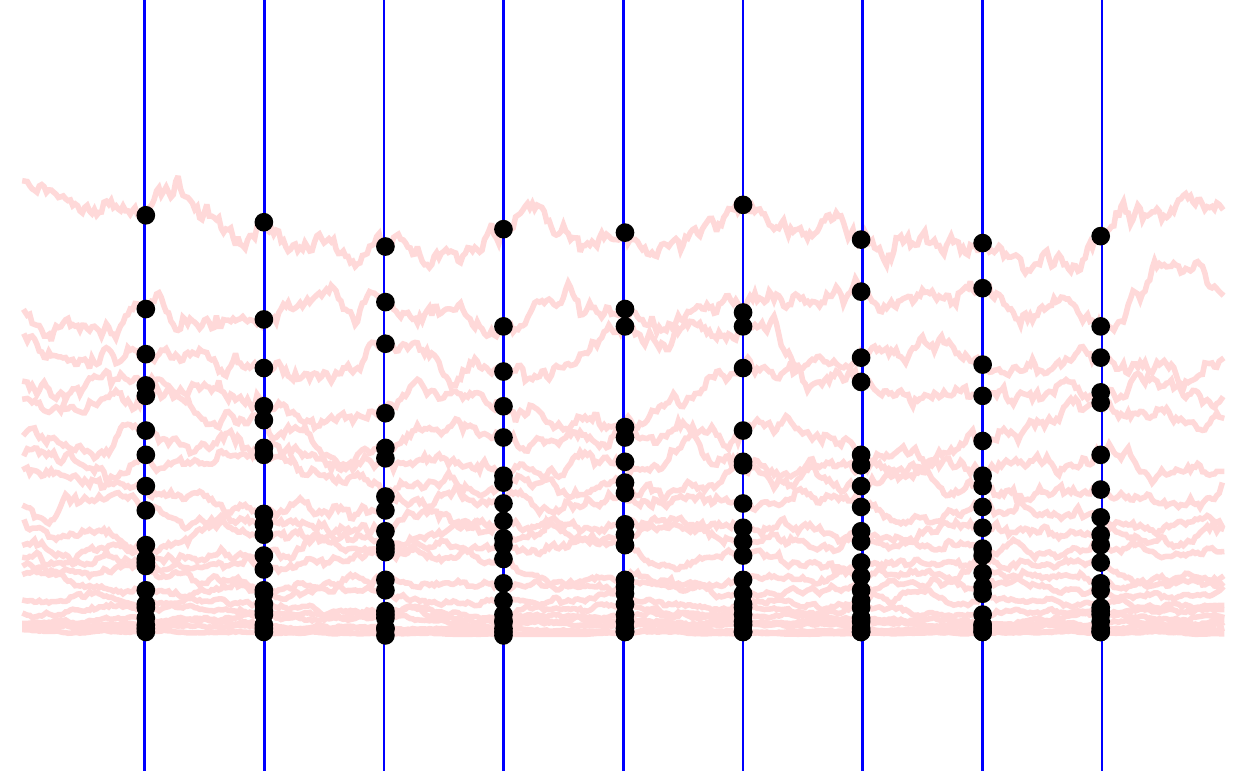}
\end{subfigure}
\caption{The left figure shows a sampling from a stationary non-colliding process generated by the squared radial Ornstein-Uhlenbeck process of size $n=20$. At the right we intersect the trajectories at multiple times $t_m$.}
\end{figure}
\begin{example}[Non-colliding Jacobi diffusions]
The last of the continuous examples is that of Jacobi diffusions, which have also been discussed in \cite{doumerc,G2}. For $\alpha, \beta>-1$, consider the Jacobi diffusion  \cite{Schoutens}
$$ P_t(x,y) {\rm d} \mu(y)=\sum_{j=1}^\infty {\rm e}^{-j(j+\alpha+ \beta+1)} p_j^{\alpha, \beta}(x) p_j^{\alpha,\beta}(y)  y^{\alpha} (1-y)^{\beta} {\rm d} y
$$
on $[0,1]$, where $p_j^{(\alpha,\beta)}(x) $ is the polynomial of degree $j$ with positive leading coefficient satisfying
$$\int_0^1  p_j^{(\alpha,\beta)}(x) p_k^{(\alpha,\beta)}(x) x^{\alpha}(1-x)^\beta{\rm d} x=\delta_{jk}.$$
Also in this case, the recurrence coefficient are explicit. Without giving them we mention that it can easily be computed that 
$$a_{n+k}\to \frac14, \qquad b_{n+k} \to 1/2.$$
In fact, this result is true for any measure $w(x) {\rm d}x$ on $[0,1]$ with  positive weight $w(x)>0$ (and even more general, the Denisov-Rakhmanov Theorem \cite[Th. 1.4.2]{SimonSzego} says that it holds for any general measure for which the essential support is $[0,1]$ and  for which the density  of the absolutely continuous part is strictly positive on $(0,1)$). 

In this case the $\lambda_j=j(j+\alpha+\beta+1)$ is quadratic. For that reason we will consider times 
$t_j= n(\alpha+ \beta+2)\tau_j$ for some fixed $\tau_j$'s, so that we have
$$\kappa_n t_j(\lambda_{n+k}-\lambda_{n+\ell})=\tau_j(k-\ell),$$
with $\kappa_n= \frac{1}{n(\alpha+ \beta+2)}$
and hence both conditions in Theorem \ref{th:OP1} are satisfied with $a=1/4$ and $b=1/2$. \hfill $\blacksquare$

\end{example}

The three examples above can also be obtained from stochastic evolution for matrices. See \cite{doumerc,Dyson,KOC} for more details.

The last three example are continuous in time and space. The next examples are concerned with a  discrete  space variable, based on birth and death processes. These are processes that model a population that can increase or decrease by  one. By tuning the birth/death rates (which may depend on the size of the population) one obtains classical orthogonal polynomials of a discrete variable. We refer to \cite{Schoutens} for more details and background.  In the $n$-particle construction as before we then arrive at stochastic evolution for which the classical discrete orthogonal polynomial ensembles are the invariant measures.  We emphasize that there are other constructions  \cite{OC} that lead to discrete orthogonal polynomial ensembles, such as the Charlier processes. Although there may be relations, these examples should not be confused with each other. 

\begin{example}[Non-colliding Meixner Ensemble] In the first example we start with a birth and death process on $\{0,1, \ldots\}$ with birth $\mu(n+\gamma)$  and death rate $n$, where $n$ is the size of the population. This process has  the transition function 
$$P_t(x,y){\rm d} \mu(y)= \sum_{j=0}^\infty {\rm e}^{-j t}   \frac{\mu^j (\gamma)_j }{j!} M_j(x;\gamma,\mu) M_j(y;\gamma,\mu)  (1-\mu)^\gamma  \frac{\mu^y(\gamma)_y} {y!}$$
on $\{0,1,2,3,\ldots\}$. Here $(\gamma)_y$ denotes the Pochhammer symbol and $M_j$ is the Meixner polynomial of degree $j$. 

The associated $n$-particle generalization  appeared in \cite{BO}. We now show how the conditions of Theorem \ref{th:OP1} are met.

The Meixner polynomials satisfy the following recursion 
\begin{multline*}
x M_k(x; \gamma,\mu)= -\frac{\mu (k+ \gamma)}{1-\mu} M_{k+1}(x,\gamma,\mu)+\frac{k(1 + \mu) +\mu \gamma}{1-\mu} M_k(x, \gamma, \mu)\\ -\frac{k}{1-\mu} M_{k-1}(x,\gamma,\mu).
\end{multline*}
Also in this case, both a recalling and normalization are needed. We define
$$p_k(\xi)= (-1)^k \sqrt{\frac{\mu^k (\gamma)_k (1-\mu)^{\gamma}}{k!}}M_k(\xi n;\gamma,\mu).$$
Then the recursion turns into 
\begin{multline*}
xp_k(\xi)= \frac{\sqrt{ \mu(k+ \gamma)(k+1)}}{n(1-\mu)} p_{k+1}(\xi)+ \frac{k(1 + \mu) +\mu \gamma}{1-\mu}  p_k(\xi)\\ + \frac{\sqrt{\mu(\gamma+k-1)k}}{n(1-\mu)} p_{k-1}(\xi).
\end{multline*}
Now it easily follows that 
$$a_{n+k,n}^{(n)} \to \frac{\sqrt \mu}{1-\mu}, \qquad b_{n+k,n}^{(n)} \to \frac{1+ \mu}{1-\mu}.$$ 
Also $\lambda_j$ so we have $\tau_j=t_j$. This shows that the conditions of Theorem \ref{th:OP1} hold. \hfill $\blacksquare$

\end{example}

\begin{example}[Non-colliding Charlier] In the next example we consider  a birth and death process on $\{0,1, \ldots\}$ with birth $\mu$  and death rate $n$, where $n$ is the size of the population. This process has  the transition function 
$$P_t(x,y){\rm d} \mu(y)= \sum_{j=0}^\infty {\rm e}^{-j t} \frac{\mu^j }{j!}  C_j(x;\mu) C_j(y;\mu) {\rm e}^{-\mu} \frac{\mu^k}{k!} $$
for  $x,y \in \{0,1,2,3,\ldots\}$, where $C_j(x;\mu)$ is the Charlier polynomial of degree~$j$. 

To apply Theorem \ref{th:OP1} for the corresponding $n$-particles process, we recall that the recursion for the Charlier reads 
$$x C_k(x; \mu)=- \mu C_{k+1}(x;\mu)+(k+\mu) C_k(x;\mu)-k  C_{k-1}(x;\mu).$$
As before, we renormalize 
$$p_k(x)= (-1)^k \sqrt{\frac{\mu^k}{k!}} C_k(x;\mu),$$
which gives the new recurrence 
$$xp_k(x)=\sqrt{\mu (k+1)} p_k(x) + (k+\mu)p_k(x)+ \sqrt{\mu k } p_{k-1}(x).$$
Now note that this case is special, since the $a_k$'s and $b_k$'s grow with different rates. This is a well-known feature of the Charlier polynomials. It means that there are two ways to get an interesting limit which we will treat separately. 

In the first one, we shift and rescale  the space variable according to
 $\xi= (x-n)/\sqrt{n},$
 and set $\hat p_k(\xi)= n^{1/4} p_k(n+ \sqrt n x)$. These polynomials satisfy the recurrence
 $$\hat p_k(\xi)= \sqrt{\frac{\mu (k+1)}{n}} \hat p_{k+1}(\xi)  +\frac{k-n+\mu}{\sqrt{n}} p_k(\xi)+ \sqrt{\frac{\mu k}{n} } p_{k-1}(\xi).$$
 Hence we see that 
 $$a_{n+k}^{(n)} \to \sqrt{\mu} \qquad \text{and} \qquad b_{n+k}^{(n)} \to 0.$$
 Combining this with the fact  $\lambda_j=j$ we see that in this way the conditions of Theorem \ref{th:OP1} are met. 
 
 In the second, way we allow $\mu$ to vary with $n$ and write $\mu= \tilde \mu n$. Then we consider the new variable 
  $\xi= x/n,$
  and set $\tilde  p_k(\xi)= \sqrt n  \tilde p_k(x n)$. Now the recurrence becomes
   $$\tilde p_k(\xi)= \frac{\sqrt{\mu k (k+1)}}{n} \tilde p_{k+1}(\xi)  +\frac{k(1+\mu)}{n} \tilde p_k(\xi)+ {\frac{\sqrt \mu k}{n} } \tilde p_{k-1}(\xi).$$
   In this case we have $$a_{n+k}^{(n)} \to \sqrt{\mu} \qquad \text{and} \qquad b_{n+k}^{(n)} \to 1+ \mu.$$
   Combining this again with the fact  $\lambda_j=j$ we see that also in this way the conditions of Theorem \ref{th:OP1} are met.
   
   Finally, we want to mention that this process is different from what is usually referred to as the Charlier process \cite{KOCR,OY}, which is non-colliding Poisson random walks starting from densely packed initial points. In that case we only allow for up jumps. The Charlier Ensemble  appears there as the fixed time distribution, but not as the invariant measure. \hfill $\blacksquare$
\end{example}

\begin{example}[Non-colliding Krawtchouck]
In the final example we consider  a birth and death process on $\{0,1, \ldots,M\}$ with birth rate $p(M-n)$  and death rate $n(1-p)$, where $n$ is the size of the population, $p \in (0,1)$ and $M \in \bbN$.  We then have the transition function 
\begin{multline}
P_t(x,y){\rm d} \mu(y)\\= \sum_{j=0}^M {\rm e}^{-j t} \begin{pmatrix} M \\ j \end{pmatrix} p^j (1-p)^{M-j}K_j(x;M,p) K_j(y;M,p)  \begin{pmatrix} M \\ k \end{pmatrix}p^k 1(-p)^{M-k}
\end{multline}
on $\{0,1,2,3,\ldots,M\}$, where $K_j$ is the Krawtchouk polynomial of degree $j$. 

In order for the $n$-particles process  to make sense, we need  enough available  nodes for all the  paths. That is, we need a $M\geq n$. In fact, when taking the limit $n \to \infty$, we will assume that $M_n$ also goes to infinity such that 
$$\frac{M_n}{n} \to \gamma >1.$$
The Krawtchouk polynomial satisfy the recurrence
\begin{multline*}
xK_k(x;M,p)= p(M-k) K_{k+1}(x;M,p)-\left(p(M-k)+k(1-p)\right) K_k(x;M,p)\\
+k (1-p)K_{k-1}(x;M,p).\end{multline*}
In this case we define the rescaled and normalized polynomials 
$$p_k(x)= (-1)^k \begin{pmatrix} M \\k \end{pmatrix} ^{1/2} \left(\frac{p}{1-p}\right)^{k/2}(-1)^k K_k(x;M,p),$$
and for these polynomials we get the recursion 
\begin{multline*}
x p_k(x)= \sqrt{p(1-p)} \frac{\sqrt{(k+1)(M-k)}}{n} p_{k+1}(x)+ \frac{p M-2p k +k }{n}p_k(x)\\
+ \sqrt{p(1-p)} \frac{\sqrt{k(M-k+1)}}{n}p_{k-1}(x).
\end{multline*}
Hence in this model we have, with $M/n \to \gamma$, 
$$a_{n+k,n}^{(n)} \to \sqrt{p(1-p)\gamma}, \qquad b_{n+k,n}^{(n)} \to p \gamma -2 p +1, $$
and, again $\lambda_j=j$, so that the conditions of Theorem \ref{th:OP1} are satisfied.

The invariant measure for the $n$-particle process is the Krawtchouk ensemble. This ensemble also appears in random domino tilings of an Aztec diamond \cite{Jannals}. However, the multi-time processes here is different from the extended Krawtchouk Ensemble in \cite{Jannals}. It is also different form the process introduced in \cite{KOCR} for which the single time distribution is a Krwatchouk Ensembles.  \hfill $\blacksquare$

\end{example}

\subsection{Non-stationary example}\label{sec:2mm}

We now consider the same construction ideas as the non-colliding Ornstein-Uhlenbeck process of Example \ref{eq:OU},  but instead of having the invariant measure as initial condition, we take the initial points random from a Unitary Ensemble. That is, we take $x_j$ random from the probability measure on $\mathbb R^n$ proportional to 
$$\prod_{1 \leq i < j\leq n } (x_i-x_j)^2 {\rm e}^{-n \sum_{j=1}^n V(x_j)} {\rm d} x_1 \cdots {\rm d} x_n,$$
where $V$ is a polynomial of even degree and positive leading coefficient (so that the above measure is indeed of finite mass). Then if we start the non-colliding Ornstein-Uhlenbeck process from these initial points and look at the positions $\{x_{j,m}\}_{j=1,m=1}^{n,N}$  at times 
$$0= t_1 < \ldots < t_N,$$
then we find that the following joint probability for these locations is proportional to 
\begin{multline} \det \left(x_{k,1}^{j-1}\right)_{j,k=1}^n \prod_{m=1}^N  \det \left(T_m(x_{j,m},x_{k,m+1}) \right)_{j,k=1}^n\\ \times  \det \left(x_{k,m}^{j-1} \right)_{j,k=1}^n  \prod_{m=1}^N \prod_{j=1}^n {\rm d} \mu_m(x_{j,m}),
\end{multline}
where, with $\Delta_m={\rm e}^{- (t_{m+1}^{(n)}-t_m^{(n)})}$, 
$$T_m(x,y)=(2 \pi)^{-1/2} \exp \left(\frac{- n\Delta_m^2 (x^2+y^2) +n\Delta_m  xy}{1-\Delta_m^2 }\right)= \sum_{j=0}^\infty {\rm e}^{-j t} H_{j}(\sqrt n x) H_j(\sqrt n y)$$
and 
$${\rm d}\mu_m(x)= \begin{cases}
{\rm e}^{-n V(x)} {\rm d} x, & m=1\\
{\rm e}^{-n x^2/2} {\rm d } x, & m=2, \ldots,N.
\end{cases} $$
(Note that we have rescaled space immediately). 

The functions in the determinant are not in the right form, as Assumption \ref{assumption} is not yet satisfied. Hence the first thing to do is to rewrite the probability density. We start by defining $p_{j,n}$ to be the normalized orthogonal polynomial (with positive leading coefficient) with respect to ${\rm e}^{-n V(x)} {\rm d} x$.
Now set
$$\phi_{j,1}(x)= p_{j-1,n}(x).$$
To define the $\phi_{j,m}$'s we first expand $p_{j-1,n}$  in terms of Hermite functions
$$p_{j,n}(x)= \sum_{k=0}^j c_{j,k} H_k(x),$$
where $H_n(x)$ are the normalized Hermite polynomials. We then define
$$\psi_{j,N}(x)= \sum_{k=0}^{j-1} c_{j,k}  {\rm e}^{kt_N } H_k(x).$$
It is then straightforward  that 
$$\psi_{j,m}(x)= \mathcal T_m^*\cdots \mathcal T_{N-1}^* \psi_{j,N}(x)=\sum_{k=0}^{j-1} c_{j,k}  {\rm e}^{kt_m } H_k(x).$$
Hence, we also have $\phi_{j,1}=\psi_{j,1}=p_{j-1,n}(x)$ and the biorthogonality condition is satisfied. In the following lemma we show that these indeed are biorthogonal families and that they satisfy a recursion.

\begin{figure}[t]
\begin{center}
\includegraphics[scale=0.26]{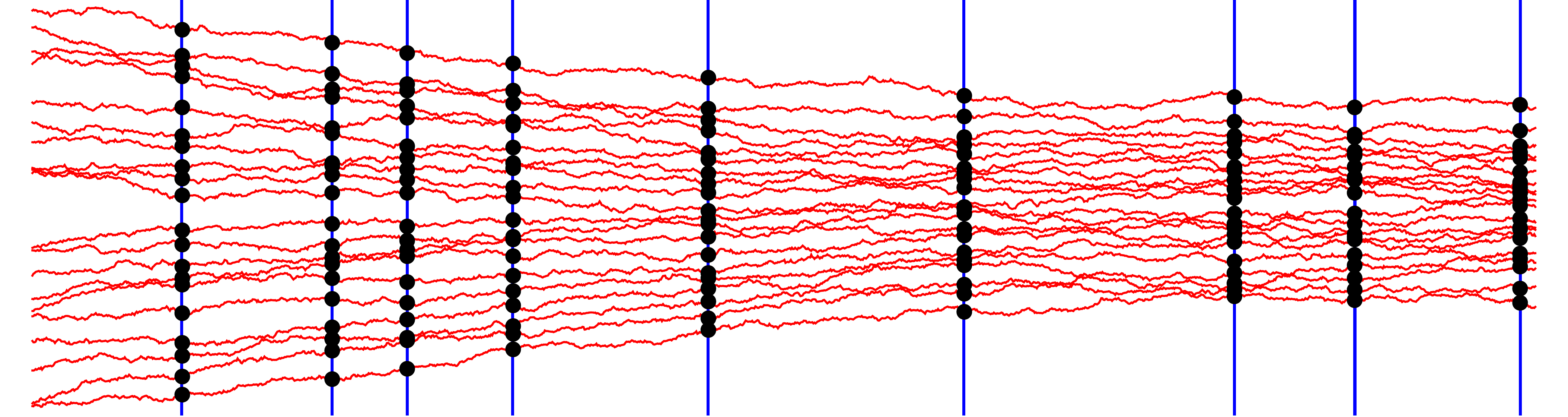}
\end{center}
\caption{The non-colliding Ornstein-Uhlenbeck process started from  arbitrary points at $t=0$.   In the example of Section \ref{sec:2mm} we take those initial points randomly from a Unitary Ensemble.} \label{DBMUE}
\end{figure}

\begin{lemma} \label{lem:nonstat}
The $\phi_{j,m}$ satisfy a recurrence relation $$x\begin{pmatrix}\phi_{1,m}(x)\\
\phi_{2,m}(x) \\
\vdots \end{pmatrix}= \mathbb J_m \begin{pmatrix}\phi_{1,m}(x)\\
\phi_{2,m}(x) \\
\vdots \end{pmatrix} $$
with $$(\mathbb J)_{j,k} = \begin{cases}
{\rm e}^{-t_m}( \mathcal J)_{j,k},& \text{if } j>k,\\
{\rm e}^{-t_m}( \mathcal J)_{j,k}+ 2 \sinh t_m (V'(\mathcal J))_{j,k}, &\text{if } {j \leq k},
\end{cases}
$$
where $\mathcal J$ is the Jacobi matrix associated to the polynomials $p_j$.
\end{lemma}
\begin{proof}
For $m=1$ the statement is trivial since then $\psi_{j,1}=\phi_{j,1}=p_{j-1,n}$ and the the recurrence matrix is $\mathcal J$ by definition, which is also the result when we substitute $t_m=0$ in the statement.

 So it remains to deal with the case $m>1$.  We first  claim that 
 \begin{equation} \label{eq:lastlast} p_{j-1,n}'(x)= n \sum_{k<j} (V'(\mathcal J))_{k,j} p_{k-1,n}(x).
 \end{equation}
 To see this, we note that $p_{j-1,n}'(x)$ is a polynomial of degree $j-2$ and hence it can be expanded in terms of the polynomials $p_{k-1,n}$ for $k=0,\ldots ,j-1$. That the coefficients in the expansion are indeed as stated follows from an integration by parts
 \begin{multline}
\int p_{k-1,n}(x) p_{j-1,n}'(x) {\rm e}^{-n V(x)} {\rm d} x
=n \int V(x) p_{k-1,n}(x) p_{j-1,n}(x) {\rm e}^{-nV(x)}{\rm d} x\\-\int p_{k-1,n}'(x) p_{j-1,n}(x) {\rm e}^{-n V(x)} {\rm d} x= (V'(\mathcal J))_{k,j},
 \end{multline}
 where the second integral in the middle part vanishes by orthogonality and the fact that $k<j$. 
 
  Then, for $m>1$, we have
$$\phi_{j,m}(x)= \frac{1}{\sqrt{2  \pi (1- {\rm e}^{-2t_m})}} \int p_{j-1,n}(y) {\rm e}^{-n \left(V(y)+ \frac{{\rm e}^{-2t_m}(x^2+y^2)-2 {\rm e}^{-t_m}x y}{2(1-{\rm e}^{-2t_m})}\right)} {\rm d} y.$$
Hence, by  integration by parts,
\begin{multline}\label{eq:recurrence2m}
x \phi_{j,m}(x)
=- \frac{2 \sinh t_m}{n}  {\rm e}^{-\frac{n{\rm e}^{-2 t_m} x^2}{2(1-{\rm e}^{-2t_m})}} 
 \int p_{j-1,n}(y)  {\rm e}^{-n V(y)- \frac{n {\rm e}^{-t_m}y^2}{2(1-{\rm e}^{-2t_m}) }} \frac{\partial }{\partial y} {\rm e}^{-\frac{n xy }{  2 \sinh t_m}} {\rm d} y 
 \\
  = 2\sinh t_m    \int \left(-p_{j-1,n}'(y)/n +p_{j-1,n}(y)  V'(y)+p_{j-1,n}(y)\frac{ {\rm e}^{-t_m}y}{2\sinh t_m}\right)\\
 \times {\rm e}^{- n (V(y)+\frac{{\rm e}^{-2t_m}(x^2+y^2)-2 {\rm e}^{-t_m}x y}{2(1-{\rm e}^{-2t_m})})} 
{\rm d} y
\end{multline}
The statement now follows by a rewriting of the latter using the recurrence matrix $\mathcal J$ and using \eqref{eq:lastlast}. \end{proof}

From Lemma \ref{lem:nonstat} it in particular follows that if the recurrence coefficients for the orthogonal polynomials have the required asymptotic behavior, then also the recurrence coefficients for $\phi_{j,m}$ have the required behavior and Theorems \ref{th:main1} and \ref{th:main1growing} apply.  
\begin{proposition}
If the recurrence coefficients $a_{k,n}$ and $b_{k,n}$  for $p_{k,n}$ satisfy $$a_{n+k,n} \to a, \qquad b_{n+k,n} \to b,$$
as $n \to \infty$,  then Theorem \ref{th:main0}  applies where
$$a^{(m)}(z)= 2 \sinh(t_m) (V'(az +b+a/z))_+ + {\rm e}^{-t_m}(az+b+a/z),$$
where $(V'(az +b+a/z))_+$ is the part of the Laurent polynomial $V'(az +b+a/z)$ containing the non-negative powers.

Moreover, Theorem \ref{th:main0growing} also applies with 
$$\sum_{j} a_j(t) z^j= 2 \sinh(t) (V'(az +b+a/z))_+ + {\rm e}^{-t}(az+b+a/z).$$
\end{proposition}
The conditions of the latter proposition are met, when the polynomial $V$ is such that the zeros of $p_n$ accumulate on a single interval \cite{DKMVZ1}. This happens for example when $V$ is convex.

Finally, we note that the above model is a special case of the Hermitian multi-matrix model. In the present setting the limiting distribution of points at a given time $t$ can also be computed using the recurrence coefficients, as was done in \cite{DGK} for the special case where $V$ is an even quartic. This even leads to a vector equilibrium problem for the limiting distribution. 
\subsection{Lozenge tilings of a hexagon}
The last example that we will treat is that of lozenge tilings of an $abc$-hexagon. See  Figure \ref{fig:hex}. This well-studied model can also be viewed as a model of discrete non-intersecting paths.
It was proved in \cite{Petrov}  (in a more general context) that the height function associated to the paths indeed has Gaussian Free Field fluctuations. We will show here that it also follows from our general results. We first give the two equivalent descriptions of the model in terms of tilings and in terms of non-intersecting paths, starting with the latter.

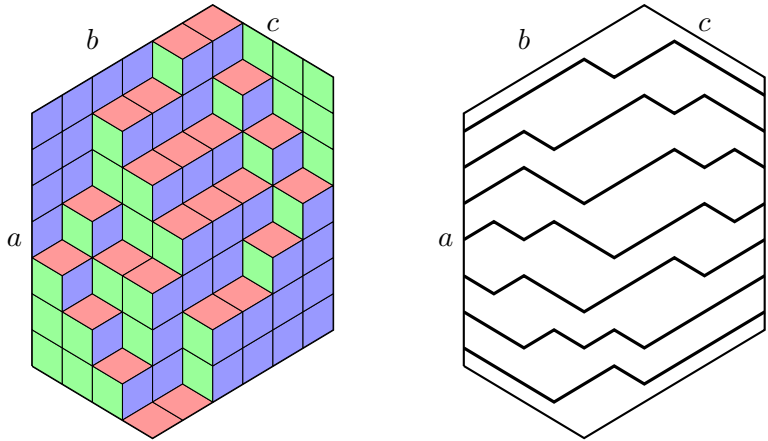
\begin{figure}[t]
\begin{center}
\begin{tikzpicture}[xscale=0.4,yscale=0.24]

\draw[ thick] (0,0)--(0,14)--(6,20)--(10,16)--(10,2)--(4,-4)--(0,0);

{
\draw (0,7) node[left] {$a$};
\draw (2,17) node[above] {$b$};
\draw (8,18) node[above] {$c$};

\draw (0,0) \lozd;
\draw (0,2) \lozd;
\draw (0,4) \lozd;
\draw (0,6) \lozu;
\draw (0,8) \lozu;
\draw (0,10) \lozu;
\draw (0,12) \lozu;

\draw (0,6) \lozr;

\draw (1,-1) \lozd;
\draw (1,1) \lozd;
\draw (1,3) \lozu;
\draw (1,7) \lozd;
\draw (1,9) \lozu;
\draw (1,11) \lozu;
\draw (1,13) \lozu;

\draw (1,3) \lozr;
\draw (1,9) \lozr;

\draw (2,-2) \lozd;
\draw (2,0) \lozu;
\draw (2,4) \lozd;
\draw (2,6) \lozu;
\draw (2,10) \lozd;
\draw (2,12) \lozd;
\draw (2,14) \lozu;

\draw (2,0) \lozr;
\draw (2,6) \lozr;

\draw (2,14) \lozr;

\draw (3,-3) \lozu;
\draw (3,1) \lozd;
\draw (3,3) \lozd;

\draw (3,7) \lozd;
\draw (3,9) \lozd;
\draw (3,11) \lozu;
\draw (3,15) \lozu;

\draw (3,-3) \lozr;
\draw (3,5) \lozr;
\draw (3,11) \lozr;
\draw (3,15) \lozr;

\draw (4,-2) \lozu;
\draw (4,0) \lozu;
\draw (4,2) \lozu;
\draw (4,6) \lozd;
\draw (4,8) \lozu;
\draw (4,12) \lozu;
\draw (4,16) \lozd;

\draw (4,-2) \lozr;
\draw (4,8) \lozr;
\draw (4,12) \lozr;
\draw (4,18) \lozr;

\draw (5,-1) \lozd;
\draw (5,1) \lozd;
\draw (5,3) \lozu;
\draw (5,5) \lozu;
\draw (5,9) \lozu;
\draw (5,13) \lozu;
\draw (5,15) \lozu;

\draw (5,3) \lozr;
\draw (5,9) \lozr;
\draw (5,13) \lozr;
\draw (5,19) \lozr;

\draw (6,-2) \lozu;
\draw (6,0) \lozu;
\draw (6,4) \lozu;
\draw (6,6) \lozu;
\draw (6,10) \lozu;
\draw (6,14) \lozd;
\draw (6,16) \lozu;

\draw (6,4) \lozr;
\draw (6,10) \lozr;
\draw (6,16) \lozr;

\draw (7,-1) \lozu;
\draw (7,1) \lozu;
\draw (7,5) \lozd;
\draw (7,7) \lozu;
\draw (7,11) \lozd;
\draw (7,13) \lozu;
\draw (7,17) \lozd;

\draw (7,7) \lozr;
\draw (7,13) \lozr;

\draw (8,0) \lozu;
\draw (8,2) \lozu;
\draw (8,4) \lozu;
\draw (8,8) \lozd;
\draw (8,10) \lozu;
\draw (8,14) \lozd;
\draw (8,16) \lozd;

\draw (8,10) \lozr;;

\draw (9,1) \lozu;
\draw (9,3) \lozu;
\draw (9,5) \lozu;
\draw (9,7) \lozu;
\draw (9,11) \lozd;
\draw (9,13) \lozd;
\draw (9,15) \lozd;
};
\end{tikzpicture}
\hspace*{1cm} 
\begin{tikzpicture}[xscale=0.4,yscale=0.24]
\draw[ thick] (0,0)--(0,14)--(6,20)--(10,16)--(10,2)--(4,-4)--(0,0);
{\draw[very thick] (0,1)--(3,-2)--(5,0)--(6,-1)--(10,3);
\draw[very thick] (0,3)--(2,1)--(3,2)--(4,1)--(5,2)--(6,1)--(10,5);
\draw[very thick] (0,5)--(1,4)--(2,5)--(4,3)--(7,6)--(8,5)--(10,7);
\draw[very thick] (0,7)--(1,8)--(2,7)--(3,8)--(5,6)--(8,9)--(9,8)--(10,9);
\draw[very thick] (0,9)--(2,11)--(4,9)--(7,12)--(8,11)--(9,12)--(10,11);
\draw[very thick] (0,11)--(2,13)--(3,12)--(6,15)--(7,14)--(8,15)--(10,13);
\draw[very thick] (0,13)--(4,17)--(5,16)--(7,18)--(10,15);}
\draw (0,7) node[left] {$a$};
\draw (2,17) node[above] {$b$};
\draw (8,18) node[above] {$c$};
\end{tikzpicture}
\end{center}

\caption{A lozenge tiling of an $abc$-hexagon (l) and the equivalent representation in terms of non-intersecting paths (r). }\label{fig:hex}
\end{figure}

 Fix $a,b,c \in \bbN$ and without loss of generality we assume that $b \leq c$. Then we consider a collection $\vec \gamma$ of $c$ zig-zag paths $\gamma_j : \{0,1,\ldots,b+c\} \to  \bbZ $ for $j=1, \ldots, a$ such that
\begin{enumerate}
\item they start $\gamma(0) =2j-1$ and end at $\gamma(b+c)=c-b+2j-1$ for $j=1,\cdots, a$. 
\item At each step by one to the right, the path goes up or down by one, i.e, $\gamma_j(k+1)= \gamma_j(k) \pm 1$. 
\item  the paths never cross $\gamma_{j}(k) < \gamma_{j+1}(k)$. 
\end{enumerate}
Note that due to the conditions on the starting and endpoints, each path will consist of $b$ down steps and $c$ up steps. We then take the uniform measure on all such $\vec\gamma$. This is equivalent to say that we consider $a$ random walkers with given starting and ending points conditioned never to intersect.

A different representation is that of lozenge tiling of the hexagon.   Indeed, if we take an $abc$-hexagon with corners $(0,0)$, $(0,2a)$, $(c,2a+c)$, $(b+c,2a+c-b)$, $(b+c,c-b)$ and $(b,-b)$ and tile this hexagon with lozenges of 
\begin{center}
type I \tikz[xscale=.2,yscale=.15]{\draw (0,0) \lozd;}, type II \tikz[xscale=.2,yscale=.15]{\draw (0,0) \lozu;} and type III \tikz[xscale=.2,yscale=.15]{\draw (0,0) \lozr;}.
\end{center} To make the connection with the above path model, we associate to each tiling of the hexagon a collection of paths by drawing a down step on a lozenge of type I and an up step on a type III lozenge going through the centers of the lozenges. That is, 
\begin{center}
type I \tikz[xscale=.2,yscale=.15]{\draw (0,0) \lozd; \draw[very thick] (0,1)--(1,0);}, type II \tikz[xscale=.2,yscale=.15]{\draw (0,0) \lozu; \draw[very thick] (0,1)--(1,2);} and type III \tikz[xscale=.2,yscale=.15]{\draw (0,0) \lozr;}.
\end{center} 
 It is then easy to see that this indeed defines a collection of zig-zag paths that do not intersect and start and end from the given points. Moreover, by taking the uniform measure on all possible tiling, we obtain the uniform measure on all zig-zag paths.

In \cite{G1,Jhahn} it was proved that the locations of the paths $\{(k,\gamma_j(k))\}_{k=1,j=1}^{b+c-1,a}$ form a determinantal point process with a  kernel  constructed out of the Hahn polynomials.  We recall that the Hahn polynomials $q_{k,M}^{(\alpha, \beta)}$ are the orthonormal polynomials with respect to the weight 
$$w_M^{(\alpha, \beta)}(x)= \frac{1}{x!(x+\alpha)!(M+\beta-x)!(M-x)!},$$
on $\{0,1,\ldots,M\}$, i.e.
$$\sum_{x=0}^M q_{k,M}^{(\alpha, \beta)}(x) q_{\ell,M}^{(\alpha, \beta)}(x) w_N^{(\alpha,\beta)}= \delta_{k \ell}.$$ 
They have the explicit representation 
$$q_{k,M}^{(\alpha,\beta)}= \frac{(-M-\beta)_k(-M)_k}{k!d_{k,M}^{(\alpha,\beta)} }3F_2\left(\begin{matrix} -k,k-2M-\alpha-\beta-1,-x\\-M-\beta,-M \end{matrix} \ ; \ 1 \right)$$ 
with 
$${d_{k,M}^{(\alpha,\beta)}}^2 =\frac{(\alpha+ \beta+M+1-k)_{M+1}}{(\alpha+\beta+2kM+1-2k)k! (\beta+M-k)!(\alpha+M-k)!(M-k)! },$$
and $(\alpha)_M=\alpha(\alpha+1) \cdots (\alpha+M-1)$ denotes the usual Pochhammer symbol. From \cite[\S 9.5]{koekoek} it follows that the normalized  Hahn polynomials have the recurrence
\begin{equation}
\label{eq:rechahn} x q_{k,M}^{(\alpha,\beta)}(x)= a_{k+1} q_{k+1,M}^{(\alpha,\beta)}(x)+b_kq_{k,M}^{(\alpha,\beta)}(x)+a_k
q_{k-1,M}^{(\alpha,\beta)}(x),
\end{equation}
where 
\begin{equation}
\label{eq:akhex}
 a_k= \sqrt{\frac{(M-k+1)k(M-k+1+\alpha)(M-k+1+\beta)(M-k+1+\alpha+\beta)(2M-k+2+\alpha+ \beta)}{(1+2M-2k+\alpha+ \beta)(2+2M-2k+\alpha+\beta)^2(3+2M-2k+\alpha+ \beta)}},
 \end{equation}
and 
\begin{multline} \label{eq:bkhex}
b_k=\frac{(2M+\alpha+ \beta+k-1)(M+\beta-k)(M-k)}{(2M-2k+\alpha+ \beta)(2M-2k+\alpha+\beta+1)}\\+ \frac{k(2M+ \alpha+ \beta+1-k)(M-k+\alpha+1)}{(2M-2k+\alpha+ \beta+2)(2M-2k+ \alpha+ \beta+1)}
\end{multline}

Now we come back to the tiling process. We first need some notations in which we follow \cite{Jhahn}.  Set $\alpha_r= |c-r|$, $\beta_r=|b-r|$, $L_r=b-b_r$ and
$$M_r= \begin{cases}
r+a-1, & 0 \leq r \leq b, \\
b+a-1, & b \leq r \leq c,\\
a+b+c-1-r,& c \leq r \leq  b+c.
\end{cases}$$
Then, as shown in  \cite{G1,Jhahn}, the locations of the paths  $\{(m,\gamma_j(m)\}$ (or, equivalently, the centers of the tiles of type II) form a determinantal point process on $\{0,1,\ldots,B+C\} \times \bbZ$ with kernel  
\begin{multline}
K_a(r,L_r+2x+1,s,L_s+2y+1)\\
=\begin{cases}
 \sum_{k=0}^{a-1}\sqrt{\tfrac{(a+s-1-k)!(a+b+c-r-1-k)!}{(a+r-1-k)!(a+b+c-s-1-k)!}}q_{k,M_r}^{(\alpha_r,\beta_r)}(x)q_{k,M_s}^{(\alpha_s,\beta_s)}(y), &r \leq s, \\
- \sum_{k=a}^{\infty}\sqrt{ \tfrac{(a+s-1-k)!(a+b+c-r-1-k)!}{(a+r-1-k)!(a+b+c-s-1-k)!}} q_{k,M_r}^{(\alpha_r,\beta_r)}(x)q_{k,M_s}^{(\alpha_s,\beta_s)}(y), & r>s,
\end{cases}.
\end{multline}
The question of interest is what happens with the system as the hexagon becomes large. That is, we introduce a  big parameter $n$ and  scale $a,b$ and $c$  such that 
$$a/n \to 1, \quad b/n \to B >0, \quad c/n \to C>0.$$
Then we take $\rho_m\in (0,B+C)$ for $m=1, \ldots, N$ and set $r_r=[n \rho_r]$ where $[\cdot]$ denotes the integer part. We also rescale the process along the vertical axis by $n$ (hence we replace $x$ by $[n x]$). 

We then set 
$$p_{j,m}(x)=q_{j,M_{r_m}}^{\alpha_{r_m}, \beta_{r_m}}([n x ])$$
and 
$$c_{j,m}= \sqrt{\frac{(a+b+c-r_m-1-j)!}{(a+r_m-1-j)!}},$$
and consider the probability measure \eqref{eq:productmeasures}  with $\phi_{j,m}= c_{j,m} p_{j-1,m} $ and $ \psi_{j,m}=1/c_{j,m} p_{j-1,m}$ as in  \eqref{eq:opchoice} (with $T_m(x,y)= \sum_{j=0}^\infty c_{j,m+1}/c_{j,m} p_{j,m}(x) p_{j,m+1}(y)).$

Denote the recurrence coefficients for $p_{j,m}$ by $a^{(n)}_{j,m}$ and $b_{j,m}^{(n)}$. Then, from \eqref{eq:akhex}, \eqref{eq:bkhex} and the choice of the parameters is not hard to show that there exists functions $F_1$ and $F_2$ as $n \to \infty$,
$$a_{n+k,m}^{(n)}\to F_1(\rho_m;B,C), \qquad b_{n+k,m}^{(n)}\to F_2(\rho_m;B,C),$$
for any $k \in \bbN$. In other words,  condition \eqref{eq:cormainthcond1} is satisfied. Moreover,  we also 
we easily verify \eqref{eq:cormainthcond2} and find
$$\lim_{n \to \infty} \frac{c_{l,m}}{c_{k,m}} = {\rm  e} ^{\tau_m(k-l)}, \qquad \text{with } \tau_m= \frac12 \ln \frac{1+B+C-\rho_m}{1+\rho_m}.$$
Hence we see that Theorem \ref{cor:th:main1} applies. Also note that after rescaling with $n$ the hexagon will always be contained in a fixed compact set for every $n$, hence also Corollary \ref{th:extend} applies.

Similarly, in the same way  one can verify that Corollary \ref{cor:main1growing}, Theorem  \ref{th:extendgrowing}  and Theorem \ref{prop:GFF} apply. We leave the precise statement to the reader.

\begin{figure}
\begin{center}
\begin{subfigure}{}
\includegraphics[scale=0.4]{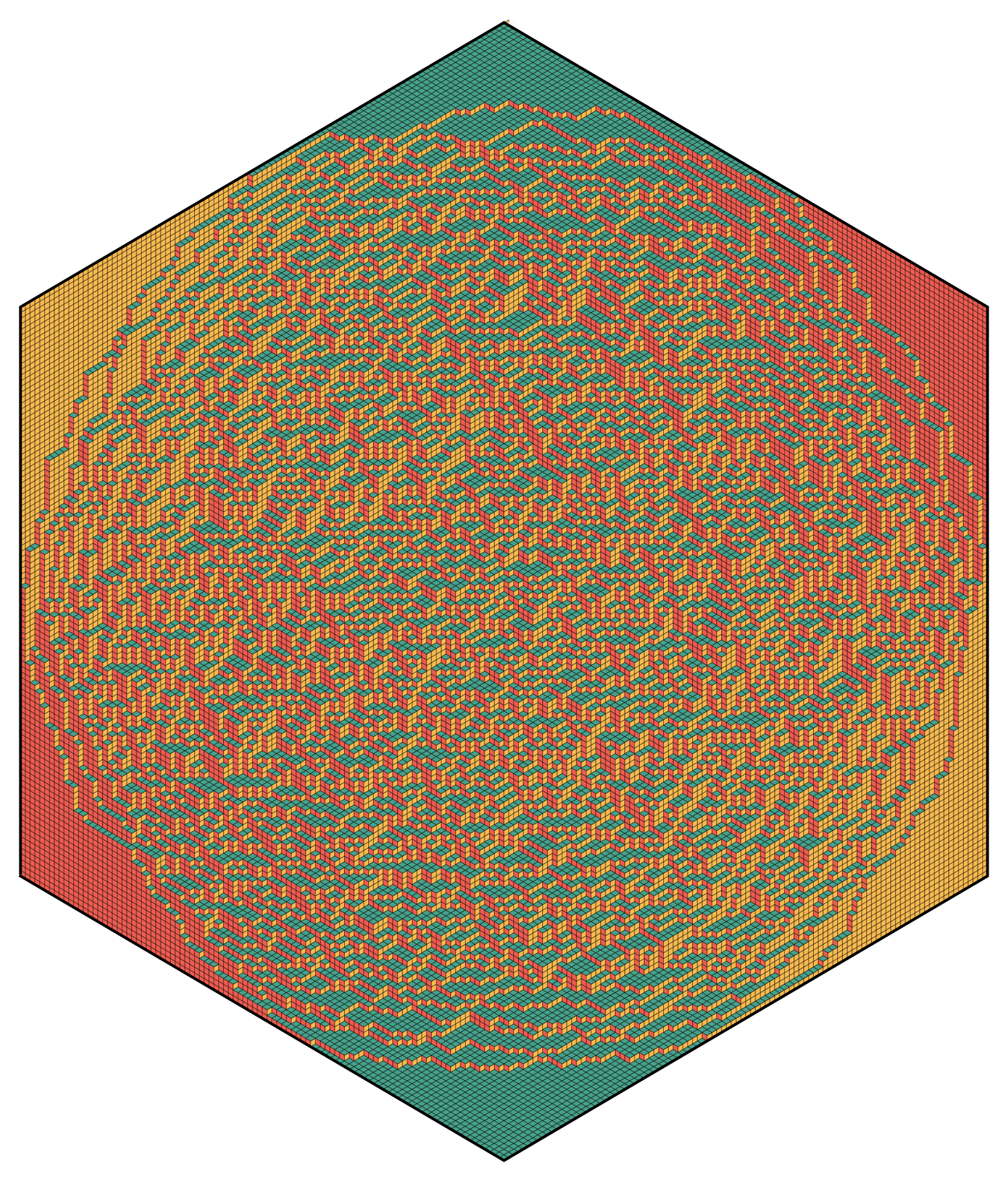}
\end{subfigure}
\begin{subfigure}{}
\includegraphics[scale=0.4]{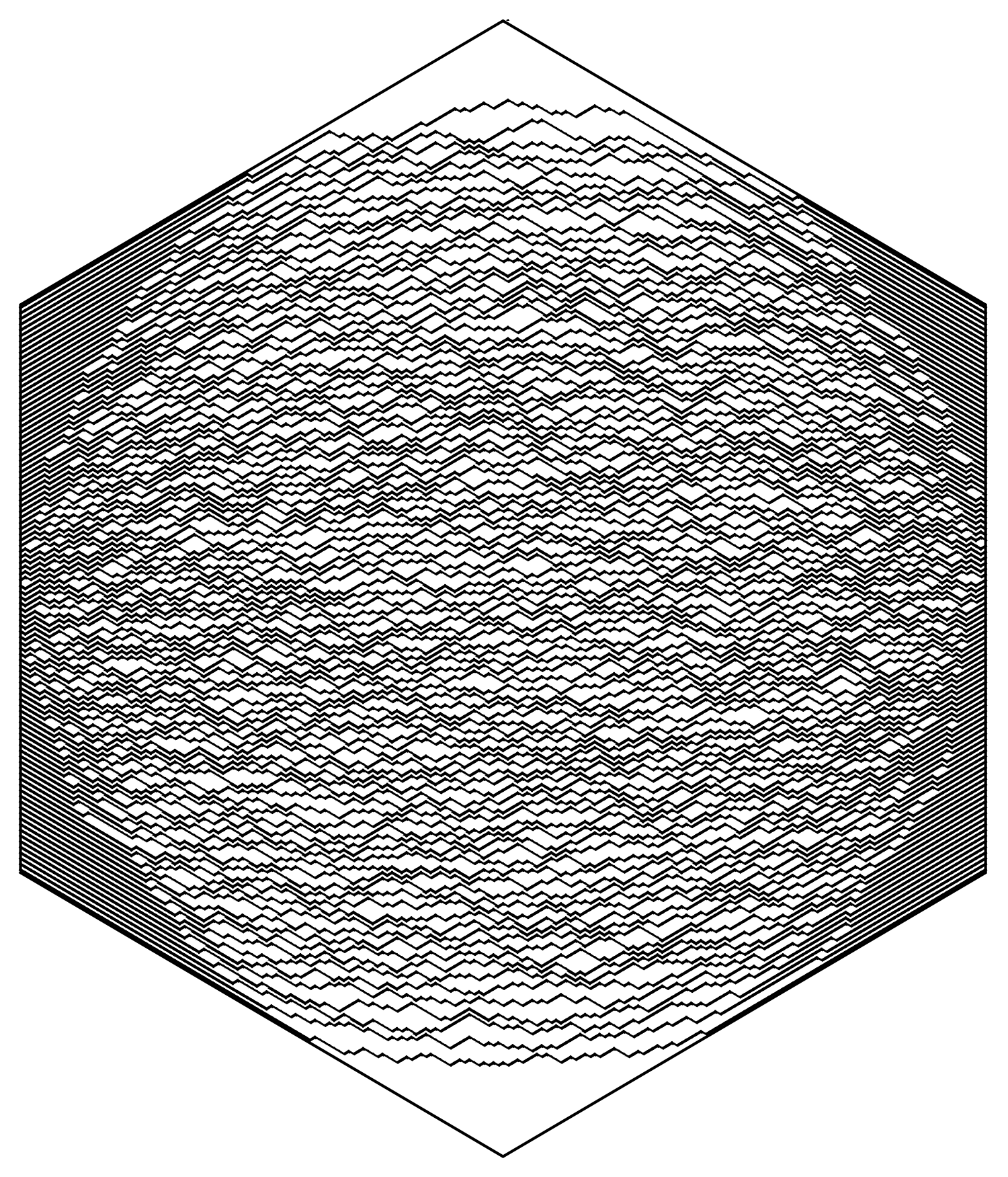}
\end{subfigure}\end{center}
\caption{A sampling of a random tiling of a large regular hexagon (l) and the alternative representation in terms of non-intersecting paths (r). The disorded regime, circle inside the hexgon, and the frozen corners are clearly visible.\protect \footnotemark } \label{fig:hexagon}
\end{figure}

\section{From linear statistics to the recurrence matrices} \label{sec:rec}
In this section we show how the moments and cumulants of the linear statistics are connected to the matrices $\mathbb J_m$.

The determinantal structure of the process means that we can express the moments of  linear statistics in terms of the kernel $K_{n,N}$ \eqref{eq:generalformKn}.  Indeed, it is standard that from  \eqref{eq:detstruct} and some computations one can show 
$$\mathbb E X_n(f) = \sum_{m=1}^N \int f(m,x) K_{n,N} (m,x,m,x){\rm d} \mu_m(x),$$
implying that $K_{n,N}(m,x,m,x)$ is the mean density. Moreover, 
\begin{multline} \label{eq:variancelinstat} \Var X_n(f)= \sum_{m=1}^N \int f(m,x)^2 K_{n,N}(m,x,m,x) {\rm d} \mu_m(x)\\
- \sum_{m_1,m_2=1}^N \iint f(m_1,x_1) f(m_2,x_2) K_{n,N} (x_1,m_1,x_2,m_2)K_{n,N} (x_2,m_2,x_1,m_1) {\rm d} \mu_{m_1} (x_1) {\rm d} \mu_{m_2} (x_2).
\end{multline}
and similar expressions hold for the higher terms.

\footnotetext{These figures are produced using a code that was kindly provided to the author by Leonid Petrov.}

Although we use  these expressions in the proofs of Theorem \ref{th:extend} and \ref{th:extendgrowing},  the general strategy in this paper is based on  a different approach.  The key identity is the following lemma, connecting the moments to the recurrence matrices $\mathbb J_m$.
\begin{lemma}\label{lem:moment}
Suppose that $f(m,x)$ is a polynomial in $x$ and $\mathbb J_m$ is bounded for $m=1, \ldots, N$, then  
\begin{equation}
\label{eq:pendantic0} \bbE \left[ {\rm e} ^{\lambda X_n(f)} \right]= \det \left(\left({\rm e} ^{\lambda f (1,\mathbb J_1)}{\rm e} ^{\lambda f (2,\mathbb J_2)} \cdots {\rm e} ^{ \lambda f (N,\mathbb J_N)}\right)_{i,j}  \right)_{i,j=1}^n.\end{equation}
In case one of the  $\mathbb J_m$ is an unbounded matrix, the equality is understood as an equality between formal power series by expanding each exponential. More precisely, with $R_M(x) = \sum_{k=0}^M x^k/k!$  consider the expansion 
\begin{multline} \label{eq:pendantic}\det \left(\left(R_M(\lambda f (1,\mathbb J_1))R_M(\lambda f (2,\mathbb J_2)) \cdots R_M(\lambda f (N,\mathbb J_N))\right)_{i,j}  \right)_{i,j=1}^n\\= \sum_{k=0}^\infty D_{k,M}(f) \frac{\lambda^k}{k!}, 
\end{multline}
then we have $\mathbb E\left[(X_n(f))^k\right]=D_{k,M}(f),$
for $k \leq M$.
\end{lemma}
\begin{remark}
Before we come to the proof we note that since the $\mathbb J_m$'s are banded matrices and $f$ is a polynomial, the product of the matrices $R_m(f(\mathbb J_m,m))$ is well-defined so that the determinant at  the left-hand side of \eqref{eq:pendantic} makes sense.
\end{remark} 
\begin{proof}
It is enough to prove \eqref{eq:pendantic}. In case all $\mathbb J_m$ are bounded we then obtain \eqref{eq:pendantic0} by taking the limit $M \to \infty$ in a straightforward way.

We first note that the first $M+1$ terms of the (formal) expansions 
$$\bbE \left[ {\rm e} ^{\lambda X_n(f)} \right]=\bbE \left[ \prod_{j=1,m}^{n,N} {\rm e} ^{\lambda f(m,x_{m,j})} \right]$$ 
and 
$$\bbE \left[ \prod_{j=1,m}^{n,N} R_M(\lambda f(m,x_{m,j})) \right]$$
are equal. To prove the lemma it  thus suffices to prove that the last expectation equals the left-hand side of \eqref{eq:pendantic}. To this end, we note
\begin{multline*}\bbE \left[ \prod_{j=1,m}^{n,N} R_M(\lambda f(m,x_{m,j})) \right]=\frac{1}{(n!)^N}\int \cdots \int  \left(\prod_{j=1,m}^{n,N} R_M(\lambda f(m,x_{m,j})) \right) \det \left(\phi_{j,1}(x_{1,k})\right)_{j,k=1}^n \\
\times  \prod_{m=1}^{N-1} \det\left( T_m(x_{m,i},x_{m+1,j}) \right)_{i,j=1}^n \det \left(\psi_{j,N}(x_{N,k})\right)_{j,k=1}^n
\prod_{m=1}^N \prod_{j=1}^n {\rm d} \mu_m(x_{m,j})\\
=\frac{1}{(n!)^N} \int \cdots \int  \det \left(   \phi_{j,1}(x_{1,k})\right)_{j,k=1}^n
 \prod_{m=1}^{N-1} \det\left( R_M(\lambda f(m,x_{m,i}))  T_m(x_{m,i},x_{m+1,j}) \right)_{i,j=1}^n  \\ \times  \det \left( R_M(\lambda f(N,x_{N,k})) \psi_{j,N}(x_{N,k})\right)_{j,k=1}^n
\prod_{m=1}^N \prod_{j=1}^n {\rm d} \mu_m(x_{m,j}),
\end{multline*}

For convenience we set some notation  $A_m=R_M(\lambda f(m,\mathbb J_m)).$ Now the statement is a special case (where $s_j=j$) of the more general claim 
\begin{multline} \label{eq:claimAs}
\frac{1}{(n!)^N} \int \cdots \int  \det \left( \phi_{s_j,1}(x_{1,k})\right)_{j,k=1}^n  \prod_{m=1}^{N-1} \det\left( R_M(\lambda f(m,x_{m,i}))  T_m(x_{m,i},x_{m+1,j}) \right)_{i,j=1}^n  \\ \times  \det \left( R_M(\lambda f(N,x_{N,k}))  \psi_{j,N}(x_{N,k})\right)_{j,k=1}^n
\prod_{m=1}^N \prod_{j=1}^n {\rm d} \mu_m(x_{m,j})\\= \det \left(\left(A_1 \cdots A_N\right)_{s_i,s_j} \right)_{i,j=1}^n,
\end{multline}
for any $s_1 <\ldots< s_n$.

The proof of \eqref{eq:claimAs} goes by induction to $N$. 

The case of $N=1$ is a direct consequence of  Andreief's identity in \eqref{eq:adnr}
\begin{multline*}\frac{1}{n!} \int \cdots \int  \det \left(  \phi_{s_j,1}(x_{1,k})\right)_{j,k=1}^n 
  \det \left(R_M(\lambda f(1,x_{1,k}))\psi_{s_j,1}(x_{1,k})\right)_{j,k=1}^n
 \prod_{j=1}^n {\rm d} \mu_1(x_{1,j})\\
= \det \left(\int R_M(\lambda f(1,x))  \phi_{s_j,1}(x)\psi_{s_i,1}(x) {\rm d} \mu_1(x)\right)_{i,j=1}^n = \det \left(\left(A_1 \right)_{s_i,s_j}\right)_{i,j=1}^n.
\end{multline*}

For $N>1$ we use  Andrei\'ef's identity to write  \begin{multline} \label{eq:claimAspr}
\frac{1}{n!} \int \cdots \int \det \left( \phi_{s_j,1}(x_{1,k})  \right)_{j,k=1}^n  \det\left( R_M( \lambda f(x_{1,i},1))  T_{1}(x_{1,i},x_{2,j}) \right)_{i,j=1}^n 
 \prod_{j=1}^n {\rm d} \mu_1(x_{1,j}) \\
= \det \left( \int  T_{1}(x,x_{2,i})  R_M( \lambda f(x,1)) \phi_{s_j,N}(x) {\rm d} \mu_1(x) \right).
\end{multline}
By using the recurrence and the fact that $$\int  T_{1}(x,x_{2,i})  \phi_{k,1}(x){\rm d}\mu_1(x)=\phi_{k,2}(x_{2,i})
$$
we find  that the right-hand side of \eqref{eq:claimAspr} can be written as
$$\det \left( \sum_{k} A_{s_j,k} \phi_{k,2}(x_{2,i}) \right)_{i,j=1}^n=\sum_{l_1 < l_2 < \cdots <l_n} \det \left(A_{s_j,l_i}\right)_{i,j=1}^n \det \left( \phi_{l_j,2}(x_{2,i}) \right)_{i,j=1}^n $$
where we used Cauchy-Binet in the last step. By inserting the latter with \eqref{eq:claimAspr} back into the left-hand side of \eqref{eq:claimAs} and using  the induction hypothesis we find that the left-hand side of \eqref{eq:claimAs} can be written as
\begin{multline*}
\sum_{l_1 < l_2 < \cdots <l_n} \det\left(\left(A_1\right)_{s_k,l_j}\right)_{j,k=1}^n \det\left((A_2 \ldots A_{N})_{l_j,s_i}\right)_{i,j=1}^n  \\
=\det\left((A_1 \ldots A_{N})_{s_i,s_k}\right)_{i,k=1}^n,
\end{multline*}
where we used Cauchy-Binet in the last step again.
This proves the claim in \eqref{eq:claimAs} and hence the statement. 
\end{proof}

This lemma also has a convenient consequence. Since all $\mathbb J_m$ are banded, each entry 
$$
\left(R_M(\lambda f (1,\mathbb J_1))R_M(\lambda f (2,\mathbb J_2)) \cdots R_M(\lambda f (N,\mathbb J_N))\right)_{i,j} 
$$
for $i,j =1, \ldots, n$ only depends on some entries of the individual $\mathbb J_m$'s. By writing out the matrix product it is not hard to see that these entries  do not depend on any $(\mathbb J_m)_{rs}$ for $m=1, \ldots, N$ and $r,s>S$ for some sufficiently large $S$.  Hence if we define the cut-offs
$$(\mathbb J_{m,S})_{j,k}= \begin{cases}
(\mathbb J_{m})_{j,k} & j,k \leq S\\
0, & \text{otherwise,}
\end{cases}
$$
and expand
\begin{equation} \label{eq:expansionsdks}
\det \left(\left({\rm e} ^{\lambda f (1,\mathbb J_{m,S})}{\rm e} ^{\lambda f (2,\mathbb J_{2,S})} \cdots {\rm e} ^{ \lambda f (N,\mathbb J_{N,S})}\right)_{i,j}  \right)_{i,j=1}^n=\sum_{k=0}^\infty  \frac{\lambda^k}{k!} \tilde {\mathcal D}_{k,S}(f),
\end{equation}
then for each $k \in \bbN$ we have $\mathbb E \left[(X_n(f))^k \right]= \tilde{ \mathcal D}_{k,S}(f)$ for sufficient large $S$ (which may depend on $k$). The benefit is that the matrix in the determinant consists of a product of bounded operators and the series is convergent.  Hence we do not have to worry about formal series and this will be convenient for technical reasons.

Instead of the moments, it will be more convenient to work with the cumulants $\mathcal C_k(X_n(f))$. These are special combinations of the linear statistic, determined by the (formal) generating function
\begin{equation}\label{eq:cumulantgeneratinggunction}
\log \mathbb E\left[\lambda X_n(f)\right]= \sum_{k=1}^\infty \frac{\lambda^k}{k!}\mathcal C_k(X_n(f)).
\end{equation}
Note that $\mathcal C_1(X_n(f))= \mathbb E X_n(f)$ and $\mathcal C_2(X_n(f))=\Var X(f)$. The $k$-th cumulant can be expressed in terms of the first $k$ moments and vice versa. Since the first terms in the expansion on the right-hand side of \eqref{eq:expansionsdks} are the moments, we can take the logarithm at both sides and immediately obtain the following lemma. 
\begin{lemma}\label{lem:cumulant}
Let $\mathcal C_{k,S}(f)$ be the coefficients in the series
$$\log \det \left(\left({\rm e} ^{\lambda f (1,\mathbb J_{m,S})}{\rm e} ^{\lambda f (2,\mathbb J_{2,S})} \cdots {\rm e} ^{ \lambda f (N,\mathbb J_{N,S})}\right)_{i,j}  \right)_{i,j=1}^n=\sum_{k=0}^\infty  \frac{\lambda^k}{k!} \mathcal C_{k,S}(f),$$
then $\mathcal C_k(X_n(f))= \mathcal C_{k,S}(f)$ for sufficiently large $S$. 
\end{lemma}

Using this  representation,  we will give  useful expressions for all the cumulants. We will do this in the next section in a more general setup. 

\section{Expansions of Fredholm determinant}\label{sec:fred}
In this  section we will look at the expansion given in Lemma \ref{lem:cumulant}, where we replace the $f(m,\mathbb J_{m,S})$'s  in the  determinant  by general \emph{banded} and \emph{bounded} operators $A_m$.
\subsection{Preliminaries}
We start by recalling traces and determinants for trace class operators. We refer to \cite{GK,Simon} for more details. 

For a compact operator $A$ on a separable Hilbert space $\mathcal H$ we define the singular values 
$$\sigma_1(A) \geq \sigma_2(A) \geq \sigma_3(A) \geq  \ldots >0,$$ as the square roots of the eigenvalues of the self-adjoint compact operators $A^* A$. The space of trace class operators is then defined as  the Banach space
$$\mathcal B_1(\mathcal H)= \{A \mid \sum_{k=1}^\infty \sigma_k(A) < \infty\},$$
equipped with the trace norm
$$\|A\|_1=\sum_{k=1}^\infty \sigma_k(A).$$ The space of Hilbert-Schmidt operators is then defined as  the Hilbert  space
$$\mathcal B_2(\mathcal H)= \{A \mid \sum_{k=1}^\infty \sigma_k(A)^2 < \infty\},$$
equipped with the  Hilbert-Schmidt norm
$$\|A\|_2=\left(\sum_{k=1}^\infty \sigma_k(A)^2\right)^{1/2}.$$
We also denote the operator norm by $\|A\|_\infty$ and the space of bounded operators by $\mathcal B_\infty(\mathcal H)$. 

The following identities are standard. For any $A \in \mathcal B_\infty(\mathcal H) $ and $B \in \mathcal B_1(\mathcal H)$ we have 
$$\|AB\|_1,\|BA\|_1  \leq \|A\|_\infty \|B\|_1.$$
Similarly, for any $A \in \mathcal B_\infty(\mathcal H) $ and $B \in \mathcal B_2(\mathcal H)$ we have 
$$\|AB\|_2,\|BA\|_2  \leq \|A\|_\infty \|B\|_2.$$
For any $A,B \in \mathcal B_2(\mathcal H)$ we have
$$\|AB\|_1 \leq \|A\|_2\|B\|_2.$$
The trace class operators $\mathcal B_1(\mathcal H)$ are precisely the operators for which we can define the trace, denote by $\Tr A$, by naturally extending the trace for finite rank operators. We note that 
$$|\Tr A|= \|A\|_1.$$
For any  trace class operator $ A \in \mathcal B_1 (\mathcal H)$ we can also define the operator determinant
$\det (I+A)$ by natural extension from the finite rank operators. Here we note that  
$$\left|\det (I+A)-\det(I+B)\right| \leq \|A-B\|_1 \exp(\|A\|_1+ \|B\|_1+1).$$
A particular relation between the trace and the determinant that we will use is the following 
\begin{equation}
\label{eq:logdet}
\log \det (I+A) = \sum_{j=1}^\infty \frac{(-1)^{j+1}}{j} \Tr A^j,
\end{equation}
valid for any $A \in \mathcal B_1(\mathcal H)$ for which $\|A\|_\infty  <1$ (ensuring the convergence of the right-hand side).
\subsection{A cumulant-type expansion}
Let $A_1,\ldots,A_N$ be bounded operator on $\ell_2(\mathbb N)$ (in the coming analysis we will identity bounded operators on $\ell_2(\mathbb N)$ with their semi-infinite  matrix representations with respect to the canonical basis).  We will also use the notation  $P_n$ for the projection operators on $\ell_2(\mathbb N)$ defined by 
$$P_n : (x_1,x_2, \ldots ) \mapsto (x_1, \ldots,x_n, 0, 0,\ldots),$$
and $Q_n=I-P_n$.
Then 
\begin{equation}
\label{eq:defoperatorcumulant}
\det \left(I+ P_n\left({\rm e}^{\lambda A_1} {\rm e}^{\lambda A_2} \cdots {\rm e}^{\lambda A_N}-I\right)P_n\right),
\end{equation}
is a well-defined and entire function of $\lambda$. 
By taking $A={\rm e}^{\lambda A_1} {\rm e}^{\lambda A_2} \cdots {\rm e}^{\lambda A_N}-I$ in \eqref{eq:logdet} for sufficiently small $\lambda$  we define $C_k^{(n)}(A_1,\ldots, A_N)$ by 
$$ \log  \det \left(I+ P_n\left({\rm e}^{\lambda A_1} {\rm e}^{\lambda A_2} \cdots {\rm e}^{\lambda A_N}-I\right)P_n\right)= \sum_{m=1}^\infty \lambda^k C_k^{(n)} (A_1,\ldots,A_N),$$
which is valid for small $\lambda$.  In  Lemma \ref{lem:cumulant} we have shown that the relation between the cumulant $\mathcal C_k(X_n(f))$ and the general coefficient $C_k^{(n)}(A_1,\ldots,A_N)$ is given by 
\begin{equation}
\label{eq:fromCtoC}
\mathcal C_k(X_n(f)) =  C_k^{(n)}(f(1,\mathbb J_{1,S}), \ldots, f(N,\mathbb J_{N,S})),
\end{equation}
for sufficiently large $S$. We will use this connection only in  Section \ref{sec:proofs} when we give the proofs of the main results. In this section  we focus on general properties of $C_k^{(n)}(A_1,\ldots,A_N)$.  To start with, an easy consequence of the above is the following. 

\begin{lemma}
We have
\begin{multline}\label{eq:firsttrace}
C_k^{(n)} (A_1,\ldots,A_N)\\ = \frac{1}{2 \pi {\rm i}} \oint_{|z|= \rho} \Tr \log  \left(I+ P_n\left({\rm e}^{\lambda A_1} {\rm e}^{\lambda A_2} \cdots {\rm e}^{\lambda A_N}-I\right)P_n\right) \frac{{\rm d}\lambda}{\lambda^{k+1}},
\end{multline}
where $0<\rho<(2   \sum_{j=1}^N \|A_j\|_\infty)^{-1}$. 
\end{lemma}
\begin{proof}
The only remaining is the choice of $\rho$. To this end, we note that 
$$\left\|P_n\left({\rm e}^{\lambda A_1} {\rm e}^{\lambda A_2} \cdots {\rm e}^{\lambda A_N}-I\right)P_n\right\|_\infty \leq |\lambda|  \sum_{j=1}^N \|A_j\|_\infty \exp\left(|\lambda | \sum_{j=1}^N \|A_j\|_\infty\right),$$
and that $\frac12 {\rm e} ^{\frac12} <1$. Hence the integrand at the right-hand side of \eqref{eq:firsttrace} is well-defined and analytic for $|\lambda |<\rho$. This proves the statement.
\end{proof}
By expanding the logarithm we obtain another useful expression. 
\begin{lemma}\label{lem:cumulantexpp}
We have
\begin{multline}
C_k^{(n)}(A_1,\ldots,A_N)\\=
\sum_{j=1}^k \frac{(-1)^{j+1} }{j} \sum_{\overset{\ell_1+ \cdots + \ell_j=k}{\ell_i\geq 1}}\,\sum_{ (r_{s,v})\in R_{\ell_1,\ldots \ell_j}} 
\frac{\Tr \prod_{s=1}^j \left(P_n A_1^{r_{s,1}} \cdots A_N^{r_{s,N}}  P_n\right)}{r_{1,1}! \cdots r_{1,N} ! r_{2,1}! \cdots r_{2,N}! \cdots r_{j,1}! \cdots r_{j,N}!}
\end{multline}
where 
$$R_{\ell_1,\ldots \ell_j}= \left\{ (r_{s,v})_{s=1, v=1}^{j,N} \mid r_{s,v} \in \{0,1,2,\ldots\}, \ \sum_{v=1}^N r_{s,v} =\ell_s\right\}$$
\end{lemma}
\begin{proof}
We note the following expansion which is valid for any bounded operators $A_1, \ldots,A_N$ and sufficiently small $\lambda$,
\begin{align*}
 \log &\left((1+P_n({\rm e}^{\lambda A_1}{\rm e}^{\lambda A_2} \cdots {\rm e}^{\lambda A_N}-I)P_n\right)=\sum_{j=1}^\infty \frac{(-1)^{j+1}}{j} \Tr \left(P_n ({\rm e}^{\lambda A_1}{\rm e}^{\lambda A_2} \cdots {\rm e}^{\lambda A_N}-I)P_n\right)^j\\
&=\sum_{j=1}^\infty \frac{(-1)^{j+1}}{j} \left(\sum_{\overset{r_1,\ldots,r_N=0}{r_1+\ldots +r_N\geq 1}}^\infty \lambda^{r_1+ \cdots r_N} \Tr P_n A_1^{r_1} \cdots A_N^{r_N}P_n\right)^j\\
&=\sum_{j=1}^\infty \frac{(-1)^{j+1}}{j} \left(\sum_{\ell=1}^\infty \lambda^{\ell} \sum_{{r_1,\ldots,r_N=\ell}}\frac{\Tr P_n A_1^{r_1} \cdots A_N^{r_N}P_n}{r_1!\cdots r_N!}\right)^j\\
&=\sum_{j=1}^\infty \frac{(-1)^{j+1}}{j} \sum_{\ell_1,\ldots ,\ell_j=1}^\infty \lambda^{\ell_1+ \cdots + \ell_j}
\sum_{ (r_{s,v})\in R_{\ell_1,\ldots \ell_j}} \frac{ \Tr \prod_{s=1}^j \left(  P_n A_1^{r_{s,1}} \cdots  A_n^{r_{s,N}}  P_n\right)}{r_{1,1}! \cdots r_{1,N} ! r_{2,1}! \cdots r_{2,N}! \cdots r_{j,1}! \cdots r_{j,N}!}\\
&=\sum_{j=1}^\infty \frac{(-1)^{j+1}}{j} \sum_{k=j}^\infty \lambda^k \sum_{\ell_1+\ldots+\ell_j=k}
\sum_{ (r_{s,v})\in R_{\ell_1,\ldots \ell_j}} \frac{  \Tr \prod_{s=1}^j \left( P_n A_1^{r_{s,1}} \cdots  A_n^{r_{s,N}}  P_n\right)}{r_{1,1}! \cdots r_{1,N} ! r_{2,1}! \cdots r_{2,N}! \cdots r_{j,1}! \cdots r_{j,N}!}\\
&=\sum_{k=1}^\infty  \lambda^k \sum_{j=1}^k \frac{(-1)^{j+1}}{j} \sum_{\ell_1+\ldots+\ell_j=k}
\sum_{ (r_{s,v})\in R_{\ell_1,\ldots \ell_j}} \frac{\Tr \prod_{s=1}^j \left( P_n A_1^{r_{s,1}} \cdots  A_n^{r_{s,N}}  P_n\right)}{r_{1,1}! \cdots r_{1,N} ! r_{2,1}! \cdots r_{2,N}! \cdots r_{j,1}! \cdots r_{j,N}!}.
\end{align*}
This proves the statement.
\end{proof}
In the proofs of the main theorems  it will be important to have the following continuity result. 
\begin{lemma}\label{lem:continuitycumu}
Let $A_1,\ldots, A_N$ and $B_1,\ldots, B_N$ be semi-infinite matrices, then 
\begin{multline*}
\left|C_k^{(n)} (A_1,\ldots, A_N)-C_k^{(n)} (B_1,\ldots, B_N)\right|\\
\leq 
\frac{2 {\rm e} }{(2-\sqrt {\rm e} )^2} (2 \max(\sum_{j=1}^N \|A_j\|_\infty, \sum_{j=1}^N \|B_j\|_\infty))^{k-1}\sum_{j=1}^N \|A_j -B_j\|_1.
\end{multline*}
\end{lemma}
\begin{proof}
We start by writing 
\begin{multline} \label{eq:cumulantsdifferenceintegral}
C_k^{(n)} (A_1,\ldots, A_N)-C_k^{(n)} (B_1,\ldots, B_N)\\
\frac{1}{2 \pi {\rm i}} \oint_{|z|= \rho} \Tr \left(\log  \left(I+ P_n\left({\rm e}^{\lambda A_1} {\rm e}^{\lambda A_2} \cdots {\rm e}^{\lambda A_N}-I\right)P_n\right) \right.\\
\left. - \log  \left(I+ P_n\left({\rm e}^{\lambda B_1} {\rm e}^{\lambda B_2} \cdots {\rm e}^{\lambda B_N}-I\right)P_n\right)\right) \frac{{\rm d}\lambda}{\lambda^{k+1}},
\end{multline}
We estimate the integrand using
\begin{equation}\label{eq:continuitytracelog}
\left|\Tr \left(\log(I+A)-\log(1+B) \right)\right| \leq \frac{\|A-B\|_1}{(1-\|A\|_\infty)(1-\|B\|_\infty)}.
\end{equation}
If we take $\rho=(2 \max(\sum_{j=1}^N \|A_j\|_\infty, \sum_{j=1}^N \|B_j\|_\infty))^{-1}$. Then
\begin{align}\label{eq:estA}
&\left\|\left({\rm e}^{\lambda A_1 } {\rm e}^{\lambda A_2} \cdots {\rm e}^{\lambda A_N}-I\right)\right\|_\infty \leq \frac{\sqrt {\rm e}}{2}\\
&\left\|\left({\rm e}^{\lambda B_1 } {\rm e}^{\lambda B_2} \cdots {\rm e}^{\lambda B_N}-I\right)\right\|_\infty \leq \frac{\sqrt {\rm e}}{2}\label{eq:estB}
\end{align}
for $\lambda= \rho$. Moreover, 
\begin{multline}\label{eq:lang}
\left\|\left({\rm e}^{\lambda A_1 } {\rm e}^{\lambda A_2} \cdots {\rm e}^{\lambda A_N}-I\right)- \left({\rm e}^{\lambda B_1 } {\rm e}^{\lambda B_2} \cdots {\rm e}^{\lambda B_N}-I\right)\right\|_1
\\
=\left\|\sum_{j=1}^N {\rm e }^{\lambda A_1} \cdots {\rm e}^{\lambda A_{j-1}} \left({\rm e}^{\lambda A_j}-{\rm e}^{\lambda B_j}  \right){\rm e }^{\lambda B_{j+1}} \cdots {\rm e}^{\lambda B_N}\right\|_1\\
\leq \sum_{j=1}^N \|{\rm e }^{\lambda  A_1}\|_\infty \cdots \| {\rm e}^{\lambda A_{j-1}}\|_\infty \left\|{\rm e}^{\lambda A_j}-{\rm e}^{\lambda B_j}  \right\|_1 \|{\rm e }^{\lambda B_{j+1}}\|_\infty \cdots \|{\rm e}^{\lambda B_N}\|_\infty\\
\leq |\lambda | \sum_{j=1}^N \|A_j -B_j\|_1 
\exp |\lambda | \left(\sum_{j=1}^N (\|A_j\|_\infty + \|B_j\|_\infty)\right) \\ \leq\frac{ {\rm e}\sum_{j=1}^N \|A_j -B_j\|_1 }{2 \max\left(\sum_{j=1}^N \|A_j\|_\infty, \sum_{j=1}^N \|B_j\|_\infty\right)}
\end{multline}
By substituting \eqref{eq:estA}, \eqref{eq:estB} and  \eqref{eq:lang}  into \eqref{eq:continuitytracelog} and using the result and the value of $\rho$ to estimate the integral \eqref{eq:cumulantsdifferenceintegral} we obtain the statement. \end{proof}

\subsection{A comparison principle}
In the next step we will prove a comparison principle for $C_k^{(n)}(A_1,\ldots,A_N)$ in case the $A_j$ are banded matrices. 

We start with an easy lemma. 
\begin{lemma}\label{lem:easy}
Let $N\in \bbN$ and $A_1,\ldots, A_N$ banded matrices with $(A_j)_{rs}=0$ if $|r-s|>a_j$. Then $A_1 \cdots A_N$ is a banded matrix such that $\left(A_1 \cdots A_N\right)_{rs}=0$ if $|r-s| >a_1+\cdots+a_N$ and  $\left(A_1\cdots A_N\right)_{rs}$ only depends on entries $(A_j)_{k\ell}$ with $|k-r|,|\ell-s| \leq a_1+ \cdots +a_N$ for $j=1,\ldots,N$.
\end{lemma}
\begin{proof}
Write 
$$\left(A_1 \cdots A_N \right)_{rs}= \sum_{v_1,\ldots,v_{N-1}} (A_1)_{rv_1} (A_{2})_{v_1v_2} \dots  (A_N)_{v_{T-1}s}.$$
By the assumption of the lemma, each term in the sum can only be non-zero if $|v_{j-1}-v_{j}|\leq a_j$ for $j=1,\ldots,N$ (where we have set $v_0=r$ and $v_{N}=s$ for notational convenience). But then by the triangular inequality, we see that the only possibility of obtain a non-zero value is in case $|r-s|\leq a_1+ \cdots+a_N$, which proves the first part of the statement. Moreover, we only have contribution of entries $(A_j)_{v_{j-1}v_{j}}$ with 
$$|r-v_{j-1}|\leq |r-v_1| + \cdots |v_{j-2}-v_{j-1}| \leq a_1+ \cdots +a_N,$$
and
$$|r-v_{j}|\leq |r-v_{j}| + \cdots |v_{N}-v_{N-1}| \leq a_1+ \cdots + a_N,$$
which proves the second part of the lemma.
\end{proof}

The following is the core of the proof of the main results of this paper. 
\begin{proposition}\label{prop:comparison}
Let $N \in \mathbb N $ and $A_1,\ldots, A_N$ be banded matrices such that $(A_j)_{rs}=0$ if $|r-s|>a_j$. Set $a= \sup_{j=1,\ldots,N} a_j$. Then 
$$C_k^{(n)}(A_1,\ldots,A_N)=C_k^{(n)}(R_{n,a(k+1)} A_1 R_{n,a(k+1)},\ldots,R_{n,2a(k+1)} A_N R_{n,2a(k+1)})$$
where
$$R_{n,\ell}=P_{n+\ell}-P_{n-\ell}.$$
\end{proposition}

\begin{proof}
Note that we have
\begin{multline*}
\lambda \Tr  \left( P_n A_1 P_n+ \cdots + P_n A_N P_n\right)
=\log \det \left((I+\left({\rm e}^{\lambda P_n A_1 P_n}{\rm e}^{\lambda P_n A_2 P_n}\cdots {\rm e}^{\lambda P_n A_N P_n}-I\right)\right)
\end{multline*}
By expanding the right-hand side in the same way as in the proof of Lemma \ref{lem:cumulantexpp} and comparing terms at both sides, we find the identity
$$\sum_{j=1}^k \frac{(-1)^j }{j} \sum_{\overset{\ell_1+ \cdots + \ell_j=k}{\ell_i\geq 1}}\,\sum_{ (r_{s,v})\in R_{\ell_1,\ldots \ell_j}} 
\frac{\Tr \prod_{s=1}^j\prod_{u=1}^N (P_n A_u P_n)^{r_{s,u}}}{r_{1,1}! \cdots r_{1,N} ! r_{2,1}! \cdots r_{2,N}! \cdots r_{j,1}! \cdots r_{j,N}!}=0,$$
for $k\geq 2$.  But then we can write 
\begin{multline*}
C_k^{(n)}(A_1,\ldots,A_N)\\=
\sum_{j=1}^m \frac{(-1)^{j+1} }{j} \sum_{\overset{\ell_1+ \cdots + \ell_j=k}{\ell_i\geq 1}}\,\sum_{ (r_{s,v})\in R_{\ell_1,\ldots \ell_j}} 
\frac{\Tr\left( \prod_{s=1}^j \left(P_n A_1^{r_{s,1}} \cdots A_N^{r_{s,N}}  P_n\right)- \prod_{s=1}^j\prod_{u=1}^N (P_n A_u P_n)^{r_{s,u}}\right)}{r_{1,1}! \cdots r_{1,N} ! r_{2,1}! \cdots r_{2,N}! \cdots r_{j,1}! \cdots r_{j,N}!},
\end{multline*}
for $m\geq 2$. We prove the theorem by  showing that the each summand only depends on some entries of $A_j$ that are all centered around the $nn$-entries.

Note that by a telescoping series, we have 
$$
 A_p^{r_{\ell,p}}-(P_n A_p P_n)^{r_{\ell,p}}=\sum_{q=0}^{r_{\ell,p}-1} A_p^{r_{\ell,p}-q-1}(A_p-P_n A_pP_n)  (P_n  A_p P_n)^{q},
$$
 and
  \begin{multline*}
  P_n \left(\prod_{u=1}^{N} A_u^{r_{\ell,u}}\right)P_n-\prod_{u=1}^{N}(P_nA_uP_n)^{r_{\ell,u}}\\
  = \sum_{p=1}^N P_n \left(\prod_{u=1}^{p-1} A_u^{r_{\ell,u}}\right)\left( A_p^{r_{\ell,p}}-(P_n A_p P_n)^{r_{\ell,p}}\right)\left(\prod_{v=p+1}^{N} (P_n A_v P_n)^{r_{\ell,v}}\right) P_n\\  
= \sum_{p=1}^N \sum_{q=0}^{r_{\ell,p}-1}P_n  \left(\prod_{u=1}^{p-1} A_u^{r_{\ell,u}}\right)  A_p^{r_{\ell,p}-q-1}(A_p-P_n A_p P_n) (P_n A_pP_n)^{q}\\ \times \left(\prod_{v=p+1}^{N} (P_n A_v P_n)^{r_{\ell,v}}\right) P_n.
 \end{multline*}
 We use the fact that $P_n^2$ to rewrite this to
  \begin{multline*}
 P_n \left(\prod_{u=1}^{N} A_u^{r_{\ell,u}}\right)P_n-\prod_{u=1}^{N}(P_nA_uP_n)^{r_{\ell,u}}\\=
   \sum_{p=1}^N \sum_{q=0}^{r_{\ell,p}-1} P_n \left(\prod_{u=1}^{p-1} A_u^{r_{\ell,u}}\right)  A_p^{r_{\ell,p}-q-1}(A_p P_n-P_n A_p P_n) (P_nA_p P_n)^{q}\\ \times \left(\prod_{v=p+1}^{N} (P_nA_v P_n)^{r_{\ell,v}}\right).
 \end{multline*}
 Finally, by another telescoping series we find 
\begin{multline*}\prod_{s=1}^j \left( P_n A_1^{r_{s,1}} \cdots  A_n^{r_{s,N}}  P_n\right)-\prod_{s=1}^j \left(P_n (A_1P_n)^{r_{s,1}} \cdots (A_N P_n)^{r_{s,N}}  \right)\\
=\sum_{\ell=1}^j \sum_{p=1}^N \sum_{q=0}^{r_{\ell,p}-1}\left(\prod_{s=1}^{\ell-1} \left(P_n A_1^{r_{s,1}} \cdots  A_N^{r_{s,N}}  P_n 
\right)P_n\left( \prod_{u=1}^{p-1} A_u^{r_{\ell,u}}\right)\right)\\
\times \left(  A_p^{r_{\ell,p}-q-1}(A_p P_n-P_n A_p P_n) (P_n A_p P_n)^{q}\right)\\
\left(\prod_{v=p+1}^{N} (P_n A_v P_n)^{r_{\ell,v}}\right)
\left(\prod_{s=\ell+1}^{j} \left(P_n (A_1 P_n)^{r_{s,1}} \cdots (P_n A_N P_n)^{r_{s,N}}  \right)\right)\\
=\sum_{\ell=1}^j \sum_{p=1}^N \sum_{q=0}^{r_{\ell,p}-1} \Tr Q_1 (A_p P_n-P_n A_p P_n)Q_2
\end{multline*}
with 
\begin{align*}
Q_1&=\left(\prod_{s=1}^{\ell-1} \left(P_n A_1^{r_{s,1}} \cdots A_N^{r_{s,N}}  P_n 
\right)\left( \prod_{u=1}^{p-1} A_u^{r_{\ell,u}}\right)\right)  A_p^{r_{k,p}-q-1}\\
Q_2&=  (P_n A_p P_n)^{q}\left(\prod_{v=p+1}^{N} (P_nA_v P_n)^{r_{k,v}}\right)
\left(\prod_{s=\ell+1}^{j} \left(P_n (A_1 P_n)^{r_{s,1}} \cdots (P_n A_N  P_n)^{r_{s,N}}  \right)\right)
\end{align*}
Let us compute 
\begin{equation}\label{eq:tracecompa}\Tr Q_1 (A_p P_n-P_n A_p P_n)Q_2= \sum_{r_0=1}^n \sum_{r_1,r_2} (Q_1)_{r_0r_1} (A_p P_n-P_n A_p P_n)_{r_1r_2}(Q_2)_{r_2r_0}.\end{equation}
The fact of the matter is that becuase of the band structure the matrix
$$A_p P_n-P_n A_p P_n$$
is of finite rank and the non-zero entries are concentrated around the $nn$-entry. Hence we can restrict the sum to terms with $|r_{1,2}-n| \leq   a$.
Now,  since $Q_1$ and $Q_2$ are a product of band matrices, they  are themselves also  band matrices. The number of terms in the product is at most $k$ (ignoring the $P_n$) and the bandwith of each terms is at most $a$.   Hence $(Q_1)_{r_0r_1}=0$ if $|r_0-r_1| > a k$ and $(Q_2)_{r_2r_0}=0$ if $|r_2-r_0|> a k$. By combining the latter observations, we see that the trace in \eqref{eq:tracecompa} only depends on $(Q_1)_{r_0 r_1}$ and $(Q_1)_{r_2 r_0}$ with $|r_0-n| \leq a(k +1)$ and $|r_{1,2}-n|\leq a$. By Lemma \ref{lem:easy}, we then also see that these entries only depend on entries $(A_m)_{rs}$ with $|r-n| \leq a(k +1)$ and $|s-n|\leq a(k+1)$ for $m=1,\ldots,N$.  Concluding, we have that $C_m^{(n)}(A_1,\ldots,A_N)$ only depends on $(A_m)_{rs}$ with $|r-n| \leq 2a(k+1)$ and $|s-n|\leq 2a(k+1)$, for $m=1,\ldots, N$. This proves the statement \end{proof}

\begin{corollary}\label{cor:comp}
Let $N \in \mathbb N $ and $A_1,\ldots, A_N, B_1, \dots B_N$ be banded matrices such that $(A_j)_{rs}=0$ if $|r-s|>a_j$ and  $(B_j)_{rs}=0$ if $|r-s|>b_j$. Set $c= max\{a_j,b_j\mid j=1,\ldots,N\}$. If 
\begin{multline}
\left |C_k^{(n)}(A_1, \ldots, A_N) -C_k^{(n)}(B_1, \ldots, B_N)\right|\\\leq 
\left ( \max(\sum_{j=1}^N \|R_{n,2c(k+1)}A_jR_{n,2c(k+1)} \|_\infty, \sum_{j=1}^N \|R_{n,2c(k+1)}B_jR_{n,2c(k+1)}\|_\infty)\right)^{k-1}\\
\times \frac{2^{k+2}c(k+1) {\rm e} }{(2-\sqrt {\rm e} )^2}\sum_{j=1}^N \|R_{n,2c(k+1)}(A_j-B_j)R_{n,2c(k+1)}\|_\infty,
\end{multline}
where $R_{n,2 c(k+1)}$ is as in Proposition \ref{prop:comparison}.
\end{corollary}
\begin{proof}
By combining Lemma \ref{lem:continuitycumu} and  Proposition \ref{prop:comparison}  we obtain
\begin{multline*}
\left |C_k^{(n)}(A_1, \ldots, A_N) -C_k^{(n)}(B_1, \ldots, B_N)\right|\\\leq 
 (2 \max(\sum_{j=1}^N \|R_{n,2c(k+1)}A_jR_{n,2c(k+1)} \|_\infty, \sum_{j=1}^N \|R_{n,2c(k+1)}B_jR_{n,2c(k+1)}\|_\infty))^{k-1}\\
\times \frac{2 {\rm e} }{(2-\sqrt {\rm e} )^2}\sum_{j=1}^N \|R_{n,2c(k+1)}A_jR_{n,2c(k+1)} -R_{n,2c(k+1)}B_jR_{n,2c(k+1)}\|_1.
\end{multline*}
Now the statement follows by noting that the ranks of  $$R_{n,2c(k+1)}A_jR_{n,2c(k+1)} \textrm{    and    } R_{n,2c(k+1)}B_jR_{n,2c(k+1)}$$ are $4c(k+1)+1$ and for any finite rank operator $R$ with rank $r(R)$ we have $\|R\|_1\leq r(R)\|R\|_\infty$. 
\end{proof}
Note that the latter corollary is a pure universality result. Whatever the limits are, they must be the same. It particularly implies that we only need to compute a special case to conclude a general result. This is what we will do in the next paragraph.
\subsection{Special case of banded Toeplitz operators}
We now compute the limiting values of $C_m^{(n)}(A_1,\ldots, A_N)$ in case the $A_j$ are banded Toeplitz operators. We will first recall various basic notions and properties we will need.  For further details and background on Toeplitz operators we refer to the book \cite{BS}.

For a Laurent polynomial $a(z)=\sum_{j=-q}^p a_j z^j$, the Toeplitz operator $T(a)$ is defined by the semi-infinite matrix 
$$\left(T(a)\right)_{jk}= a_{j-k}, \qquad j,k=1,\ldots,$$
viewed as an operator on $\ell_2(\bbN)$.  Of importance to us will also be the Hankel operator defined by the semi-infinite matrix 
$$\left(H(a)\right)_{jk}= a_{j+k-1}, \qquad j,k=1,\ldots.$$
Note that $H(a)$ is of finite rank.  The Toeplitz and Hankel operators are related by 
$$T(a b)= T(a) T(b)+ H(a) H(\tilde b),$$
with $\tilde b(z)=b(1/z)$. An important consequence of this formula that we will frequently use is
\begin{equation}
\label{eq:commutatorToeplitz}
[T(a),T(b)]= H(b) H(\tilde a)-H(a) H(\tilde b).
\end{equation}
Finally, we mention that 
$$\|T(a)\|_\infty \leq \|a\|_{\mathbb L_\infty}, \qquad \text{ and } \qquad \|H(a)\|_\infty \leq \|a\|_{\mathbb L_\infty}.$$

\subsubsection{The case $N$ fixed}
The main purpose of this paragraph is to prove the following proposition. 
\begin{proposition}\label{prop:fredholmtoeplitzfixedn}
Let $N \in \bbN$ and  $a(m,z)=\sum_{\ell} a_\ell(m) z^\ell$ for $m=1,\ldots, N$, be Laurent polynomials in $z$. For $m=1,\ldots, N$ we denote the Toeplitz operator with symbol $a(m,z)$ by  $T(a(m))$. Then 
\begin{multline}\label{eq:toeplitzidentity}
\lim_{n\to \infty} 
\det \left(I+ P_n \left({\rm e}^{T(a(1)}{\rm e}^{T(a(2))}\cdots {\rm e}^{T(a(N))}-I\right)P_n\right) {\rm e}^{-n \sum_{m=1}^N a_0(m)}\\
= \exp\left(  \sum_{m_1=1}^N \sum_{m_2 =m_1+1}^N \sum_{\ell=1}^\infty \ell a_\ell {(m_1)} {a}_{-\ell} {(m_2)}  +\frac12 \sum_{m=1}^N\sum_{\ell=1}^\infty\ell a_\ell{(m)} {a}_{-\ell}{(m)}\right).
\end{multline}
\end{proposition}
Before we come to the proof we first mention that the following immediate corollary.
\begin{corollary}\label{prop:cumulantstoeplitzfixedn}
Under the same assumptions and notations as in Proposition \ref{prop:fredholmtoeplitzfixedn} we have 
$$
\lim_{n\to \infty} C_k^{(n)}(T(a(1),\ldots, T(a(N)) =
0,$$
for $k\geq 3$, and 
\begin{multline*}
\lim_{n\to \infty} C_k^{(n)}(T(a(1),\ldots, T(a(N))\\= \sum_{m_1=1}^N \sum_{m_2 =m_1+1}^N \sum_{\ell=1}^\infty \ell a_\ell {(m_1)} {a}_{-\ell} {(m_2)}  +\frac12 \sum_{m=1}^N\sum_{\ell=1}^\infty\ell a_\ell{(m)} {a}_{-\ell}{(m)}
\end{multline*}
for $k=2$.
\end{corollary}

The proof of Proposition \ref{prop:fredholmtoeplitzfixedn} that we will present here relies on the following beautiful identity due to Ehrhardt.

\begin{lemma}{\cite[Cor. 2.3]{E}}\label{prop:ehrhardt}
Let $A_1, \ldots, A_N$ be bounded operators  such that 
\begin{enumerate}
\item $A_1+ \cdots + A_N=0$,
\item $[A_i,A_j] $ is trace class for $1 \leq i <j \leq N$.
\end{enumerate}
Then 
$ {\rm e}^{A_1} {\rm e} ^{A_2} \cdots {\rm e}^{A_N} -I$
is of trace class and 
$$\det {\rm e}^{A_1} {\rm e} ^{A_2} \cdots {\rm e}^{A_N}= \exp \frac12 \sum_{1 \leq i <j \leq n} \Tr [A_i,A_j],
$$
where the left-hand side is a Fredholm determinant. 
\end{lemma}
\begin{remark}
In the special case $n=3$ the identity reads
$$\det {\rm e}^{-A} {\rm e} ^{A+B} {\rm e}^{-B}= \exp\left( -\frac12  \Tr [A,B]\right),$$
which is an identity that can be used that lies behind the Strong Szeg\H{o} Limit Theorem. It has also been used in \cite{BD} in the context of Central Limit Theorem for linear statistics for biorthogonal ensembles. 
\end{remark}

We now come to the 
\begin{proof}[Proof of Proposition  \ref{prop:fredholmtoeplitzfixedn}]
We start by defining $s(z)= \sum_{j=1}^N a(j,z)$ and split 
$s(z)= s_+(z)+s_-(z)$ where $s_+$ is the polynomial part of $s(z)$, i.e.  $s_+(z)= \sum_{k\geq 0} s_k z^k$. Then we note that
$T(s_+)$ is lower triangular and $T(s_-)$ is strictly upper triangular. 
Hence 
$$P_n T(s_+)P_n= P_n T(s_+), \qquad \text{and} \qquad P_n T(s_-)P_n= T(s_-) P_n.$$
By expanding the exponential and iterating the latter identities we therefore have  (where we recall that $Q_n=I-P-n$)
$${\rm e}^{P_n T(s_+) P_n} = P_n {\rm e}^{ T(s_+)}P_n+ Q_n = P_n {\rm e}^{ T(s_+)}+ Q_n,$$
and 
$${\rm e}^{P_n T(s_-) P_n} = P_n {\rm e}^{ T(s_-)}P_n+ Q_n =  {\rm e}^{ T(s_-)}P_n+ Q_n.$$
Moreover,\begin{multline} \label{eq:hulp}
\left(Q_n+ P_n {\rm e}^{T(s_+)} P_n\right)\left(Q_n+P_nB P_n\right)\left(Q_n+ P_n {\rm e}^{T(s_-)} P_n\right)\\
=Q_n+ P_n {\rm e}^{T(s_+)} P_n B P_n  {\rm e}^{T(s_-)} P_n\\
=Q_n+ P_n {\rm e}^{T(s_+)} B   {\rm e}^{T(s_-)} P_n=I+ P_n \left({\rm e}^{T(s_+)} B   {\rm e}^{T(s_-)}-I \right)P_n,
\end{multline}
for any operator $B$.   We then write 
$${\rm e}^{-n \sum_{j=1}^N a^{(j)}_0}= {\rm e}^{-\Tr P_n T(s_+) P_n}=\det {\rm e}^{- P_n T(s_+)P_n}= \det \left(Q_n+ P_n {\rm e}^{-T(s_+)} P_n\right),$$
and 
$$1= {\rm e}^{-\Tr P_n T(s_+) P_n}=\det {\rm e}^{ -P_n T(s_-)P_n}= \det \left(Q_n+ P_n {\rm e}^{-T(s_-)} P_n\right).$$
By combining this with  \eqref{eq:hulp} and taking $B= {\rm e}^{-T(s_+(z))} {\rm e}^{T(a{(1)})}{\rm e}^{T(a{(2)})}\cdots {\rm e}^{T(a{(N)})} $, we see that we can rewrite the left-hand side of \eqref{eq:toeplitzidentity} as 
 \begin{equation}\label{eq:samvaryingfixed}
 \det \left( I+ P_n \left({\rm e}^{-T(s_+)} {\rm e}^{T(a{(1)})}{\rm e}^{T(a{(2)})}\cdots {\rm e}^{T(a{(N)})}   {\rm e}^{-T(s_-)}-I \right)P_n\right).
 \end{equation}
 The idea is now to invoke Lemma \ref{prop:ehrhardt}. To this end, we first  note that 
 $$T(s)+ \sum_{j=1}^N T(a^{(j)})=0,$$
 and 
 that for $1\leq j,k \leq N$ we have that 
 $$[T(a_+{(j)}), T(a_-{(k)})]= H(a_+{(j)}) H(\tilde a_-{(k)}),$$
 is of finite rank and hence of trace class, from which it follows that also the second condition  of the proposition is satisfied and that
$${\rm e}^{-T(s_+)} {\rm e}^{T(a{(1)})}{\rm e}^{T(a{(2)})}\cdots {\rm e}^{T(a{(N)})}   {\rm e}^{-T(s_-)}-I$$
is of trace class. Hence, 
\begin{multline}\label{eq:traceclassconvergence}
P_n\left({\rm e}^{-T(s_+)} {\rm e}^{T(a{(1)})}{\rm e}^{T(a{(2)})}\cdots {\rm e}^{T(a{(N)})}   {\rm e}^{-T(s_-)}-I\right)P_n \\
\to 
{\rm e}^{-T(s_+)} {\rm e}^{T(a{(1)})}{\rm e}^{T(a{(2)})}\cdots {\rm e}^{T(a{(N)})}   {\rm e}^{-T(s_-)}-I,
\end{multline}
in trace norm. 
 By continuity of the Fredholm determinant we can therefore take the limit $n \to \infty$ and obtain 
\begin{multline}\label{eq:limittoepltizfredholm}
\lim_{n\to \infty}  \det \left( I+ P_n \left({\rm e}^{-T(s_+)} {\rm e}^{T(a{(1)})}{\rm e}^{T(a{(2)})}\cdots {\rm e}^{T(a{(N)})}   {\rm e}^{T(s_-)}-I \right)P_n\right)\\
=  \det {\rm e}^{T(s_+)} {\rm e}^{T(a{(1)})}{\rm e}^{T(a{(2)})}\cdots {\rm e}^{T(a{(N)})}   {\rm e}^{-T(s_-)}.
\end{multline}
Moreover, by the same proposition, 
\begin{equation}\label{eq:limitwithS1} \det {\rm e}^{-T(s_+)} {\rm e}^{T(a{(1)})}{\rm e}^{T(a{(2)})}\cdots {\rm e}^{T(a{(N)})}   {\rm e}^{-T(s_-)}= \exp \frac12\Tr S_1,
\end{equation}
with 
$$
S_1 = -\sum_{j=1}^N[ T(s_+),T(a^{(j)})]+[T(s_+),T(s_-)]
+\sum_{j=1}^N \sum_{k >j} [T(a{(j)}),T(a{(k)})- \sum_{j=1}^N [T(a{(j)}),T(s_-)].
$$
By splitting $a(m)=a_+(m)+a_-(m)$  and using the definition of $s_\pm$, we can rewrite this  to 
\begin{multline*}
S_1= \sum_{j=1}^N \sum_{k >j} [T(a{(j)}),T(a{(k)})]- \sum_{j=1}^N \sum_{k=1}^N [T(a{(j)}),T(a_-{(k)})]
\\
= \sum_{j=1}^N \sum_{k >j} [T(a{(j)}),T(a_+{(k)})]- \sum_{j=1}^N \sum_{k \leq j} [T(a{(j)}),T(a_-{(k)})]\\
=\sum_{j=1}^N \sum_{k >j} H(a{(k)}) H (\tilde {a}{(j)}) +\sum_{j=1}^N \sum_{k \leq j}H(a{(j)})H( \tilde {a}{(k)}),
\end{multline*}
where in the last step we used \eqref{eq:commutatorToeplitz}. 
By taking the trace we find 
\begin{multline*}
\Tr S_1= \sum_{j=1}^N \sum_{k =j+1}^N \sum_{\ell=1}^\infty \ell a_\ell{(k)}  {a}_{-\ell}{(j)}   +\sum_{j=1}^N \sum_{k = 1}^j\sum_{\ell=1}^\infty\ell a_\ell{(j)} {a}_{-\ell}{(k)}\\=  2\sum_{j=1}^N \sum_{k =j+1}^N \sum_{\ell=1}^\infty \ell a_\ell{(k)}  {a}_{-\ell}{(j)}   +\sum_{j=1}^N\sum_{\ell=1}^\infty\ell a_\ell{(j)} {a}_{-\ell}{(j)}
\end{multline*}
Hence the statement follows after inserting the latter expression for $\Tr S_1$ into \eqref{eq:limitwithS1} and combining the result with \eqref{eq:limittoepltizfredholm}
\end{proof}
 \subsubsection{The case $N_n \to \infty$}
 We now come to the case that $N=N_n$ is depending on $n$ in such a way that $N_n \to \infty$ as $n \to \infty$. In this paragraph we prove the following proposition. 

\begin{proposition}\label{prop:fredholmtoeplitzvaryingn}
Let $a(t,z)=\sum_{j=-q}^p a_j(t)z^j$ be a $t$ dependent Laurent polynomial for which the $a_j(t)$ are piecewise continuous on an interval  $[\alpha,\beta]$.  Let $\{N_n\}_{n \in \bbN}$ be a sequence such that $N_n \to \infty$ as $n \to \infty$. Moreover, for each $n$ let $$\alpha=t_1^{(n)} < t_2^{(n)} < \ldots < t_{N_n}^{(n)}< t_{N_n+1}^{(n)}=\beta$$ be a partitioning of $[\alpha,\beta]$ for which the mesh $\sup_{j=0, \ldots, N_n} (t_{j+1}-t_j) \to 0$ as $n \to \infty$.  Then 
\begin{multline*}
\det \left(I+ P_n \left(\prod _{m=1}^{N_n} {\rm e}^ { (t_{m+1}-t_m) T(a(t^{(n)}_{m} ))} -I \right)P_n\right) {\rm e}^{-\sum_{m=1}^{N_n} (t_{m+1}-t_m) a_0(t_m^{(n)})}
\\
=\exp\left({\sum_{\ell=1}^\infty \iint_{\alpha<t_1<t_2<\beta} \ell a_\ell(t_1) a_{-\ell}(t_2) {\rm d} t_1 {\rm d} t_2}\right)
\end{multline*}
\end{proposition}

\begin{corollary}\label{prop:cumulantstoeplitzvaryingn}
Under the same assumptions and notations as in Proposition \ref{prop:fredholmtoeplitzvaryingn} we have 
$$
\lim_{n\to \infty} C_k^{(n)}(T(a(1),\ldots, T(a(N_n)) =
0,$$
for $k\geq 3$, and 
$$
\lim_{n\to \infty} C_k^{(n)}(T(a(1),\ldots, T(a(N))\\= 2{\sum_{\ell=1}^\infty \iint_{\alpha<t_1<t_2<\beta} \ell a_\ell (t_1) a_{-\ell}(t_2) {\rm d} t_1 {\rm d} t_2},
$$
for $k=2$.
\end{corollary}

The proof of Proposition \ref{prop:fredholmtoeplitzvaryingn} goes along the same lines as the proof of Proposition \ref{prop:fredholmtoeplitzfixedn}. The main difficulty  is that \eqref{eq:traceclassconvergence} (with $N$ and the symbols depending on $n$) is no longer immediate and requires a proof. Hence we can not deduce  \eqref{eq:limittoepltizfredholm}. We overcome this issue by proving the following.

\begin{lemma}\label{lem:hulp}
Let $\{N_n\}_n$ be a sequence of integers. For each $n\in \mathbb N $ 
let $A_j^{(n)}$ for $j=1,\ldots, N_n$ be a family of Toeplitz operators satisfying the following conditions
\begin{itemize}
\item $A_1^{(n)}+ \cdots+ A_{N_n}^{(n)}=0$
\item the $A_j^{(n)}$'s are banded with width $c$ which is independent of $j$ and  $n$.
\item  $\sum_{j=1}^{N_n} \|A_j^{(n)} \|_\infty <r$, for some constant $r$ independent of $n$. 
\end{itemize}
 Then
\begin{multline}\lim_{n\to \infty}  \left\|P_n\left({\rm e}^{A_1^{(n)}} \cdots {\rm e}^{A_{N_n}^{(n)} }-I\right)P_n- \left({\rm e}^{A_1^{(n)}} \cdots {\rm e}^{A_{N_n}^{(n)} }-I\right)\right\|_1 = 0.
\end{multline}
\end{lemma}

\begin{proof}
Note that 
\begin{multline*}
P_n\left({\rm e}^{A_1^{(n)}} \cdots {\rm e}^{A_{N_n}^{(n)} }-I\right)P_n- \left({\rm e}^{A_1^{(n)}} \cdots {\rm e}^{A_{N_n}^{(n)} }-I\right)\\
=-Q_n\left({\rm e}^{A_1^{(n)}} \cdots {\rm e}^{A_{N_n}^{(n)} }-I\right)P_n- \left({\rm e}^{A_1^{(n)}} \cdots {\rm e}^{A_{N_n}^{(n)} }-I\right)Q_n 
\end{multline*}
Hence it suffices to prove that the two terms at the right-hand side separately converge to zero in trace norm. Here we will only show that for the second term, i.e. we show that 
\begin{equation}\label{eq:qnbla}
\lim_{n\to \infty}  \left\|\left({\rm e}^{A_1^{(n)}} \cdots {\rm e}^{A_{N_n}^{(n)} }-I\right)Q_n\right\|_1 \to 0.
\end{equation}
The arguments for first term are analogous and left to the reader. 

We first claim that  for $A_j$ and $r,c$  satisfying the stated conditions, we have 
\begin{equation}
\label{eq:claimvaryingN} \left\|\left({\rm e}^{A_1} \cdots {\rm e}^{A_{N} }-I\right)Q_n\right\|_1\leq {\rm e}^{2 r} \sum_{2 \leq s < t \leq N} \left\|[A_s,A_t]\right\|_1 \sum_{m_2+m_3 \geq n/(2c)}\frac{  r^{m_2+m_3}}{m_2! m_3!}, \end{equation}
for any $n \geq 2c $ and $N \in \mathbb N$. 

The proof of this statement goes by induction to $N$. 

We start with $N=3$. In that case we recall a result from \cite[Lem. 4.2]{BDmes}  that if $A_1+ A_2+A_3=0$ then 
$${\rm e}^{A_1}{\rm e}^{A_2} {\rm e}^{A_3} -I = \sum_{m_1=0}^\infty \sum_{m_2=1}^\infty \sum_{m_3=0}^\infty \sum_{j=0}^{m_2-j-1} \frac{ A_1^{m_1} A_2^j [A_2,A_3]A_2^{m_2-j-1} A_3^{m_3}}{m_1! m_2! m_3! (m_1+m_2+m_3+1)},$$
(this follows easily after differentiating ${\rm e}^{t A_1}{\rm e}^{ tA_2} {\rm e}^{t A_3}$ with respect to $t$, expanding the exponentials and then integrate again for $t$ from $0$ to $1$).  Since $A_j$   are banded with bandwidth $c$ we also have 
$A_2^{m_2-j-1} A_3^{m_3}$
is banded with bandwidth $c (m_2+m_3-j-1)$ but that means that 
$$A_2^{m_2-j-1} A_3^{m_3}Q_n= Q_{n-c(m_2+m_3-j-1)} A_2^{m_2-j-1} A_3^{m_3}Q_n$$
and thus
$$ [A_2,A_3]A_2^{m_2-j-1} A_3^{m_3}Q_n= [A_2,A_3] Q_{n-c(m_2+m_3-j-1)} A_2^{m_2-j-1} A_3^{m_3}Q_n$$
Since by assumption $A_2$ and $B_3$ are banded Toeplitz matrices of bandwith $c$, we see by \eqref{eq:commutatorToeplitz} that $[A_1,A_2]$ is sum of  product of two Hankel matrices with Laurent polynomials as their symbols. All the non-trivial entries for these Hankel operators are in the upper left $c\times c$ block and hence $$[A_2,A_3]Q_{m} =0$$ for  $m \geq c$.
 Hence, for $n \geq 2 c$,  we can restrict the sum to terms with $m_2+m_3 \geq n/2c$ and write
\begin{multline*}
 \left\|\left({\rm e}^{-A} {\rm e}^{A+B } {\rm e}^{-B}-I\right)Q_n\right\|_1\\
 \\
 =\left\|\sum_{m_1=0}^\infty\sum_{\overset{m_2 \geq 1,m_3 \geq 0,}{ m_2+m_2 \geq n/2c}}  \sum_{j=0}^{m_2-j-1} \frac{ (A_1)^{m_1} (A_2)^j [A_2,A_3](A_2)^{m_2-j-1} (A_3)^{m_3}Q_n}{m_1! m_2! m_3! (m_1+m_2+m_3+1)}\right\|_1\\ 
 \leq  \sum_{m_1=0}^\infty \sum_{\overset{m_2 \geq 1,m_3 \geq 0,}{ m_2+m_2 \geq n/2c}} \sum_{j=0}^{m_2-j-1} \frac{ \|A_1\|_\infty^{m_1} \|A_2\|_\infty^j \| [A_1,A_2]\|_1 \|A_2\|_\infty ^{m_2-j-1} \|A_3\|_\infty^{m_3}}{m_1! m_2! m_3! (m_1+m_2+m_3+1)}\\
\leq {\rm e}^{r} \left\|[A_2,A_3]\right\|_1 \sum_{m_2+m_3 \geq n/(2c)}\frac{  r^{m_2+m_3}}{m_2! m_3!} ,
\end{multline*}
and this proves the statement for $N=3$ (with a slightly better bound).

Now suppose the statement is true for  for  $N-1 \geq 2$.  We then first write
\begin{multline*}
{\rm e}^{A_1} {\rm e}^{A_2} \cdots {\rm e}^{A_{N}} -I= {\rm e}^{A_1} {\rm e}^{A_1} {\rm e}^{A_2} \cdots {\rm e}^{A_{N-2}} {\rm e}^{A_{N-1}+A_{N}}\left({\rm e}^{-A_{N-1}-A_{N}}{\rm e}^{A_{N-1}}{\rm e}^{A_{N}}-I\right)  \\+{\rm e}^{A_1} {\rm e}^{A_2} \cdots {\rm e}^{A_{N-2}} {\rm e}^{A_{N-1}+A_{N}}-I.
\end{multline*}
The first term at the right-hand side can be estimate by as in the case $N=3$ giving 
\begin{multline}
 \left\|{\rm e}^{A_1} {\rm e}^{A_1} {\rm e}^{A_2} \cdots {\rm e}^{A_{N-2}} {\rm e}^{A_{N-1}+A_{N}}\left({\rm e}^{-A_{N-1}-A_{N}}{\rm e}^{A_{N-1}}{\rm e}^{A_{N}}-I\right)\right\|_1 \\
 \leq {\rm e}^{\sum_{j=1}^{N} \|A_j\|_\infty} \left\|{\rm e}^{-A_{N-1}-A_{N}}{\rm e}^{A_{N-1}}{\rm e}^{A_{N}}-I\right\|_1 \\
\leq {\rm e}^{2r} \left\|\left[A_{N-1},A_{N}\right]\right\|_1 \sum_{m_2+m_3 \geq n/(2c)}\frac{ r^{m_2+m_3} }{m_2! m_3!}.\label{eq:inductionpart1}
\end{multline}
Moreover, since 
$$\sum_{j=1}^{N-2} \|A_j\|_\infty  + \|{A_{N-1}+A_{N}} \|_\infty <\sum_{j=1}^{N} \|A_j\|_\infty  \leq r$$
and $A_{N-1}+A_N$ by linearity is also a Toeplitz operator with bandwidth $c$, we have by the induction hypothesis that 
we find that 
\begin{multline}
\left\|{\rm e}^{A_1} {\rm e}^{A_2} \cdots {\rm e}^{A_{N-2}} {\rm e}^{A_{N-1}+A_{N}} -I\right| \\
\leq {\rm e}^{2 r} \left(\sum_{1 \leq s < t \leq  N_{n}-2} \|[A_s,A_t ] \|_1 + \sum_{s=1}^{N-2} \left\|\left[A_s,A_{N-1}+A_{N} \right]\right\|_1\right)\sum_{m_2+m_3 \geq n/(2c)}\frac{ r^{m_2+m_3} }{m_2! m_3!}
\\
\leq {\rm e}^{2 r} \left(\sum_{1 \leq s < t \leq  N_{n}-1} \|[A_s,A_t ] \|_1 + \sum_{s=1}^{N-2} \left\|\left[A_s,A_{N} \right]\right\|_1\right)\sum_{m_2+m_3 \geq n/(2c)}\frac{ r^{m_2+m_3} }{m_2! m_3!}.\label{eq:inductionpart2}
\end{multline}
Hence, by combining the \eqref{eq:inductionpart1} and \eqref{eq:inductionpart2} we obtain the claim  \eqref{eq:claimvaryingN} for any $N \in \mathbb N$.  

To finish the proof we recall again that by \eqref{eq:commutatorToeplitz} we have that $[A_s,A_t]$ is sum of  product of two Hankel matrices with Laurent polynomials as their symbols. All the non-trivial entries for these Hankel operators are in the upper left $c\times c$ block and hence they are of rank $c$.   Since we also have $\|H(a)\|_\infty \leq \|T(a)\|_\infty$, we therefore find
$$\|[A_s,A_t]\|_1\leq 2c^2 \|A_s\|_\infty \|A_t\|_\infty.$$
But then we have 
$$\sum_{2 \leq s < t \leq N} \left\|[A_s,A_t]\right\|_1 \leq 2 c^2 \sum_{2 \leq s < t \leq N}\|A_s\|_\infty \|A_t\|_\infty <r^2 c^2.$$  
After inserting this into  \eqref{eq:claimvaryingN} we obtain 
$$
\left\|\left({\rm e}^{A_1} \cdots {\rm e}^{A_{N} }-I\right)Q_n\right\|_1\leq 2r^2c^2{\rm e}^{2 r}\sum_{m_2+m_3 \geq n/(2c)}\frac{  r^{m_2+m_3}}{m_2! m_3!}, $$
for any $n \geq 2 c$, $N \in \mathbb N$ and operators $\{A_j\}_{j=1}^N$ satisfying the conditions of the proposition (with respect to $c$ and $r$). By setting $N=N_n$, $A_j=A^{(n)}_j$  and taking $n \to \infty$ we obtain \eqref{eq:qnbla}. This finishes the proof. 
 \end{proof}
 \begin{proof}{Proof of Proposition \ref{prop:fredholmtoeplitzvaryingn}}
We argue exactly the same as in the proof of \ref{prop:fredholmtoeplitzfixedn} until \eqref{eq:samvaryingfixed} giving 
 \begin{equation}
 \det \left( I+ P_n \left({\rm e}^{A_0^{(n)}} {\rm e}^{A_1^{(n)}}\cdots {\rm e}^{A_{N+1}^{(n)}}  -I \right)P_n\right).
 \end{equation}
 where 
 $$A_m^{(n)}= (t_{m+1}^{(n)}-t_m^{(n)}) T(a(t_m^{(n)})), \qquad m=1, \ldots, N$$
 and
 $$A_0^{(n)} =- \sum_{m=1}^N(t_{m+1}^{(n)}-t_m^{(n)}) T(a_+(t_m^{(n)})), \text{  and  } A_{N+1}^{(n)} =- \sum_{m=1}^N(t_{m+1}^{(n)}-t_m^{(n)}) T(a_-(t_m^{(n)})).$$
 Hence it is clear that $A^{(n})_0+ \ldots + A^{(n)}_{N+1}=0$. Moreover, each $A_m^{(n)}$ is a banded Toeplitz matrix with band-width $\max(p,q)$. Finally, we check the condition on the norm of the matrices. To this end, 
 \begin{multline}
 \sum_{m=1}^{N} \|A_m^{(n)}\|_\infty=  \sum_{m=1}^N(t_{m+1}^{(n)}-t_m^{(n)})\| T(a(t_m^{(n)}))\|_\infty \\
 = \sum_{m=1}^N(t_{m+1}^{(n)}-t_m^{(n)})\| a(t_m^{(n)})\|_\infty  \leq (\beta-\alpha) \sup_{t} \|a(t)\|_\infty \leq  (\beta-\alpha) \sum_{j=-q}^p \sup_t |a_j(t)|.
 \end{multline}
 The latter is finite, since we assume that the $a_j$ are piecewise continuous. 
 A similar argument shows that also $\|A_0^{(n)}\|_\infty, \|A_{N+1}^{(n)} \|\leq \sum_{j=-q}^p \sup_t |a_j(t)|.$ Hence the last condition of Lemma \ref{lem:hulp} is also satisfied. It then follows from that lemma, that 
 \begin{multline}
\lim_{n\to \infty}  \det \left( I+ P_n \left({\rm e}^{A_0^{(n)}} {\rm e}^{A_1^{(n)}}\cdots {\rm e}^{A_{N+1}^{(n)}}  -I \right)P_n\right)\\
    = \lim_{n\to \infty} \det \left( {\rm e}^{A_0^{(n)}} {\rm e}^{A_1^{(n)}}\cdots {\rm e}^{A_{N+1}^{(n)}} \right).
  \end{multline}
  The right-hand side is the exponential of a trace that we can compute in the same way as we have done in the proof of Proposition \ref{prop:fredholmtoeplitzfixedn} which gives 
  \begin{multline}
\lim_{n\to \infty}  \det \left( I+ P_n \left({\rm e}^{A_0^{(n)}} {\rm e}^{A_1^{(n)}}\cdots {\rm e}^{A_{N+1}^{(n)}}  -I \right)P_n\right)\\
 = \lim_{n\to \infty} \exp\left(  2\sum_{m_1=1}^{N_n} (t_{m_1+1}^{(n)}-t_{m_1}^{(n)}) \sum_{m_2 =m_1+1}^{N_n} (t_{m_2+1}^{(n)}-t_{m_2}^{(n)}) \sum_{\ell=1}^\infty \ell a_\ell{(t_{m_1}^{(n)})}  {a}_{-\ell}{(t_{m_"}^{(n)}}\right.\\
 \left.   +\sum_{m=1}^{N_n} (t_{m+1}^{(n)}-t_{m}^{(n)})^2 \sum_{\ell=1}^\infty\ell a_\ell{(m)} {a}_{-\ell}{(m)}\right).
  \end{multline}
In the limit $n \to \infty$, the double sum in the  exponent converges to a Riemann-Stieltjes integral and the single sum tends to zero. We thus proved the statement.  \end{proof}
\section{Proofs of the main results}\label{sec:proofs}
\subsection{Proof of Theorem \ref{th:main0}}
\begin{proof}[Proof of Theorem \ref{th:main0}]
To prove that the difference of the moments converges to zero, it is sufficient to prove that for $k\geq 2$ we have 
$$\mathcal  C_k(X_n(f))-\tilde {\mathcal C}_k(X_n(f)) \to 0,$$
as $n \to \infty$, where the $\mathcal C_k$'s are the cumulants defined in \eqref{eq:cumulantgeneratinggunction}. By Lemmas \ref{lem:moment} and \ref{lem:cumulant} and this means that we need to prove that, for given $k \geq 2$, 
\begin{equation}\label{eq:cumulantsS}
 C_k^{(n)}(f(1,\mathbb J_{1,S}), \ldots, f(N,\mathbb J_{N,S}))-C_k^{(n)}(f(1,\tilde {\mathbb J}_{1,S}), \ldots f(N,\tilde {\mathbb J}_{N,S})) \to 0,
 \end{equation}
as $n \to \infty$ for $S$ sufficiently large. The key ingredient is Corollary \ref{cor:comp}.

  Since $x\mapsto f(m,x)$ is a polynomial and $\tilde J_m$ and $\tilde{\hat J}_m$ are banded, it is not hard to see that  \eqref{eq:cond1main0}  implies that, for given  $k, \ell \in \mathbb Z$,
$$\left(f(m,\mathbb J_{m})\right)_{n+k,n+ \ell}-\left(f(m,\tilde {\mathbb J}_{m})\right)_{n+k,n+ \ell}\to 0,$$
as $n \to \infty$. 

Now, let $k \in \bbN$. Note that for any  $S \in \bbN$ we have that both $f(m,\mathbb J_{m,S})$ and $f(m,\tilde {\mathbb J}_{m,S})$ are banded matrices with a bandwidth $\rho$ that is independent of $n$.  With $R_{n,2\rho(k+1)}$ as in Corollary \ref{cor:comp} we then have 
\begin{equation} \label{eq:proofthm0A}
\left\|R_{n,2\rho(k+1)}\left(f(m,\mathbb J_{m,S})-f(m,\tilde {\mathbb J}_{m,S})\right)R_{n,2\rho(k+1)}\right\|_\infty\to 0,
\end{equation}
as $n \to \infty$ for sufficiently large $S$. We also have by the first condition in Theorem \ref{th:main0} that there exists an $M>0$ such that
\begin{equation}\label{eq:proofthm0B}
\left\|R_{n,2\rho(k+1)}f(m,\tilde {\mathbb J}_{m,S})R_{n,2\rho(k+1)}\right\|_\infty<M,
\end{equation}
for $n \in \bbN$. Hence, by \eqref{eq:proofthm0A}, we can  also choose $M$ to be large enough so that we have
\begin{equation}\label{eq:proofthm0C}
\left\|R_{n,2\rho(k+1)}f(m,\mathbb {J}_{m,S})R_{n,2\rho(k+1)}\right\|_\infty<M,\end{equation}
for $n \in \bbN$. The statement now follows by \eqref{eq:cumulantsS} and  inserting \eqref{eq:proofthm0A}, \eqref{eq:proofthm0B} and \eqref{eq:proofthm0C} for sufficiently large $S$ into the conclusion of Corollary \ref{cor:comp}. \end{proof}
\subsection{Proof of Theorem \ref{th:main1}}
\begin{proof}[Proof of Theorem \ref{th:main1}]
To prove that $X_n(f)-\EE X_n(f)$ converges to a normally distributed random variable, it is sufficient to prove that the $k$-th cumulant converges to zero if $k \geq 3$ and to the stated value of the variance if $k=2$.  By \eqref{eq:fromCtoC} this means that we need to show that, for given $k>2$, 
$$\lim_{n\to \infty} C_k^{(n)}(f(1,\mathbb J_{1,S}), \cdots f(N, \mathbb J_{N,S})) =0,$$
for some sufficiently large $S$, and that   the variance $ C_k^{(n)}$ converges to the stated value.  

Let $k \in \bbN$, The assumptions of the theorem imply 
that 
$$\lim_{n\to \infty}  \left(\left(\mathbb J_m)\right)_{n+k,n+l}- \left(T(a(m))\right)_{n+k,n+l}\right)=0,$$
with $a(z)=\sum_{j=-q}^p a_j^{(m)} z^j$ where $a_j$  are the values in \eqref{eq:mainthcond1}. Since $x\mapsto f(x,m)$  is a polynomial it is not hard to see that also 
$$
\lim_{n\to \infty} \left( \left(f(m,\mathbb J_{m,S})\right)_{n+k,n+l}-\left(T(b(m))\right)_{k,l}\right)=0,
$$
for a sufficiently large $S$ and  with $b(m)= \sum_{k} f_k^{(m)} z^k $ and $f_k^{(m)}$ as in \eqref{eq:defpk}.  The statement then follows after applying Corollaries \ref{cor:comp} and  \ref{prop:cumulantstoeplitzfixedn}.\end{proof}

\subsection{Proof of Theorem \ref{th:extend}}
We now show how the results can be extended to allow for more general functions in the case \eqref{eq:opchoice}. We recall   that  the variance of the linear statistics is given in terms of the kernel $K_{n,N}$ as in \eqref{eq:variancelinstat}. Note that  for fixed $m$ the kernel in this situation reads  $$K_{n,N}(m,x,m,y)=
\sum_{j=1}^n p_{j,m_1}(x)p_{j,m_1}(y),$$
and hence  $K_{n,N}(m,x,m,y)= K_{n,N}(m,y,m,x)$. This  symmetry is a  key property  in the coming analysis.  We start with  the following estimate on the variance. 
\begin{lemma}\label{lem:compactsupports}
Let $f(m,x)$ be such that $\frac{\partial f}{\partial x} (m,x)$ is a bounded function.  Then 
$$\Var \sum_{m=1}^N \sum_{j=1}^n f(m,x_j(t_m))  \leq \frac N2 \sum_{m=1}^N \left\|\frac{\partial f}{\partial x} (m,x) \right\|_\infty^2 \|[P_n,J_m]\|_2^2.$$
\end{lemma}
\begin{proof}
Let us first deal with the special case that $N=1$. In that case, the variance of a linear statistic can be written as (with $K_{n,N}(x,1,y,1)$ shortened to $K_{n,N}(x,y))$
\begin{multline*}\Var \sum_{j=1}^n f(x_j(t))= \int f(x)^2 K_{n,N}(x,x) {\rm d} \mu(x)\\
- \iint f(x) f(y) K_{n,N}(x,y) K_{n,N}(y,x)  {\rm d} \mu(x) {\rm d}\mu(y)
\end{multline*}
By using the fact that, by orthogonality, we have  $\int K_{n,N}(x,y) K_{n,N}(y,x) {\rm d} \mu(y)=K_{n,N}(x,x)$ we can rewrite this in symmetric form 
\begin{equation}\label{eq:variance}\Var \sum_{j=1}^n f(x_j(t))= \frac12
 \iint (f(x)- f(y))^2 K_{n,N}(x,y) K_{n,N}(y,x)  {\rm d} \mu(x) {\rm d}\mu(y).
\end{equation}
We now use the fact that $K_{n,N}(x,y)$ is symmetric (and hence $K_{n,N}(x,y)K_{n,N}(y,x) \geq 0$) to estimate this as 
\begin{equation}\Var \sum_{j=1}^n f(x_j(t))\leq \frac12\|f'\|_\infty 
 \iint (x-y)^2 K_{n,N}(x,y) K_{n,N}(y,x)  {\rm d} \mu(x) {\rm d}\mu(y).
\end{equation}
It remains to estimate the double integral. To this end, we note that
$$\frac12
 \iint (x-y)^2 K_{n,N}(x,y) K_{n,N}(y,x)  {\rm d} \mu(x) {\rm d}\mu(y)= \|[P_n,J]\|_2^2,$$
 proving the statement for $N=1$. For general $N$ we use the inequality $\Var \sum_{m=1}^N  X_m \leq N  \sum _{m=1}^N \Var X_m$, valid for any sum of random variables. 
\end{proof}
In the next lemma, we show that we can approximate the variance for linear statistics by linear statistic corresponding to polynomials. 
\begin{lemma}\label{lem:extendB}
Assume there exists a compact set $E \subset \mathbb R$ such that \\
(1) $S(\mu^{(n)}_m) \subset E$ for $m=1, \ldots, N$ and $n \in \mathbb N$, \\
or, more generally,\\
(2) for $k \in \mathbb N$ and $m=1,\ldots, N$ we have $$ \int_{\mathbb R \setminus E}  |x|^k K_{n,N}(m,x,m,x) {\rm d} \mu_m^{(n)}(x) = o(1/n),$$ as $n\to \infty$. 

Then for any  function$f$ such that $x \mapsto f(m,x)$ is a continuously differentiable function with at most polynomial growth at $\pm \infty$, we have
\begin{equation}\label{eq:conclusionlemmaextend}
\Var X_n(f) \leq  \frac N2 \sum_{m=1}^N \sup_{x \in E} \left|\frac{\partial f}{\partial x} (m,x)  \right|^2 \|[P_n,J_m]\|_2^2 + o(1),
\end{equation}
as $n \to \infty$. 
\end{lemma}
\begin{proof}
If the supports of  $\mu^ {(m)}$ all lie within a compact set of $\bbR$ (independent of $n$)  then this is simpy Lemma \ref{lem:compactsupports}.
Now suppose the measure is not compactly supported but condition (2)  is satisfied with some $E$. Then we write again
\begin{equation}\label{eq:extensionB}
 \Var \sum_{m=1}^N \sum_{j=1}^n f(m,x_j(t_m)) \leq N \sum_{m=1}^N \Var \sum_{j=1}^n f(m,x_j(t_m)).
 \end{equation}
Each term in the sum at the right-hand side can be written in the form of \eqref{eq:variance} and we have 
\begin{multline}\Var \sum_{j=1}^n f(m,x_j(t_m))\\
= \frac12
 \iint_{E \times E}  (f(m,x)- f(m,y))^2 K_{n,N}(m,x,m,y) K_{n,N}(m,y,m,x)  {\rm d} \mu(x) {\rm d}\mu(y)\\
+ \frac12 \iint_{(\bbR \setminus E ) \times E}  (f(m,x)- f(m,y))^2 K_{n,N}(m,x,m,y) K_{n,N}(m, y,m,x)  {\rm d} \mu(x) {\rm d}\mu(y)\\
+\frac12 \iint_{\bbR \times (\bbR \setminus E ) }  (f(m,x)- f(m,y))^2 K_{n,N}(m,x,m,y) K_{n,N}(m,y,m,x)  {\rm d} \mu(x) {\rm d}\mu(y). \label{eq:extensionC}
\end{multline}
Now,  since $$K_{n,N}(m,x,m,y)K_{n,N}(m,y,m,x) \leq K_{n,N}(m,x,m,x) K_{n,N}(m,y,m,y),$$ $$\int_\bbR K_{n,N}(m,y,m,y) {\rm d} \mu(y)=n,$$ and  $g$ has at most polynomial growth at infinity, we have
\begin{multline}
\iint_{(\bbR \setminus E ) \times E}  (f(m,x)- f(m,y))^2 K_{n,N}(m,x,m,y) K_{n,N}(m,y,m,x){\rm d} \mu(x) {\rm d}\mu(y) \\
\leq \iint_{(\bbR \setminus E ) \times E}  (f(m,x)- f(m,y))^2 K_{n,N}(m,x,m,x) K_{n,N}(m,y,m,y) {\rm d} \mu(x) {\rm d}\mu(y)
\\= o(1/n) \int K_{n,N}(m,y,m,y) {\rm d} \mu(y)=o(1), 
\label{eq:extensionD}
\end{multline}
as $n \to \infty$. For the same reasons, 
\begin{equation}\label{eq:extensionsE}
\iint_{\bbR \times (\bbR \setminus E ) }  (f(m,x)- f(m,y))^2 K_{n,N}(m,x,m,,y) K_{n,N}(m,y,m,x)  {\rm d} \mu(x) {\rm d}\mu(y)= o(1),
\end{equation}
as $n \to \infty$. Moreover
\begin{multline}
 \iint_{E \times E}  (f(m,x)- f(m,y))^2 K_{n,N}(m,x,m,y) K_{n,N}(m,y,m,x)  {\rm d} \mu(x) {\rm d}\mu(y)\\
 \leq 
 \sup_{x \in E} \left|\frac{\partial f}{\partial x} (m,x)  \right|^2\iint_{E \times E}  (x-y)^2 K_{n,N}(m,x,m,y) K_{n,N}(m,y,m,x)  {\rm d} \mu(x) {\rm d}\mu(y)\\
\leq  \sup_{x \in E} \left|\frac{\partial f}{\partial x} (m,x)  \right|^2\iint_{\bbR \times \bbR}  (x-y)^2 K_{n,N}(m,x,m,y) K_{n,N}(m,y,m,x)  {\rm d} \mu(x) {\rm d}\mu(y)\\
= \sup_{x \in E} \left|\frac{\partial f}{\partial x} (m,x)  \right|^2\left\|[P_n,J_m]\right\|_2^2.
\label{eq:extensionsF}
 \end{multline}
By inserting \eqref{eq:extensionD}, \eqref{eq:extensionsE} and \eqref{eq:extensionsF} into \eqref{eq:extensionC} and the result into \eqref{eq:extensionB}, we obtain the statement. 
\end{proof}
We need one more lemma. 
\begin{lemma}\label{lem:extendC}
Let $a_0^{(m)} \in \bbR$ and $a_1^{(m)}>0$ for $m=1, \ldots, N$. Then for any real valued function $f$ on $\{1,\ldots,N\} \times \bbR$ such that $x\mapsto f(m,x)$ is continuously differentiable, we have that 
$$\sigma(f)^2= \sum_{k=1}^\infty \sum_{m_1,m_2=1}^\infty 
{\rm e}^{-|\tau_{m_1}-\tau_{m_2}|k } f_k^{(m_1)} f_k^{(m_2)},$$
is finite, where 
$$f_k^{(m)}= \frac{1}{2 \pi} \int_0^{2 \pi} f\left(m,a_0^{(m)}+  2 a_1^{(m)} \cos  \theta\right) {\rm e}^{-{\rm i} k \theta} {\rm d} \theta.$$
Moreover,  for some constant $C$ we have 
$$\sigma(f)^2 \leq C N \sum_{m=1}^N \sup_{x\in F_m} \left | \frac{ \partial f}{\partial x}(m,x)\right |^2,$$
where $F_m= [a_0^{(m)}-2 a_1^{(m)},a_0^{(m)}+2 a_1^{(m)}]$. 
\end{lemma}
\begin{proof}
We start by recalling Remark \ref{rem:sym} on the symmetric case.  By the same arguments leading to \eqref{eq:variancesymmetricpositivr}, we see that we can write 
$$\sigma(f)^2=\sum_{k=1}^\infty \int_{-\infty}^\infty\left|\sum_{m=1}^N {\rm e}^{-{\rm i} \tau_m \omega} k f_k^{(m)}\right|^2\frac{{\rm d} \omega}{k^2+\omega^2}, $$
showing that $\sigma(f)^2$ is indeed positive (if finite). To show that it is finite, we note by applying Cauchy-Schwarz on the sum over $m$ and using 
\begin{equation}
\label{eq:piiiii}
\int\frac{{\rm d} \omega}{k^2+ \omega^2} \leq \int\frac{{\rm d} \omega}{1+ \omega^2} = \pi,
\end{equation}
it easily follows that 
\begin{equation}\label{eq:sigmaineq} \sigma(f)^2\leq C_0 N \sum_{m=1}^N \sum_{k=1}^\infty k^2 |f_k^{(m)}|^2.
\end{equation}
for some constant $C_0>0$. 
Then by integration by parts we find 
$$k f_k^{(m)} = \frac{1}{2\pi } \int_0^{2 \pi} \frac{ \partial f}{\partial x} (m,a_0^{(m)} +2 a_1^{(m)} \cos \theta)  a_1^{(m)} \sin  \theta \  {\rm e}^{-{\rm i} k \theta} {\rm d} \theta.$$
By applying Parseval and using the compactness of $F$ to estimate the $\mathbb L_2$ norm by the $\mathbb L_\infty$ norm, we find there exists a constant $C$ such that 
$$\sum_{k=1}^\infty k^2 |f_k^{(m)}|^2 \leq C\sup_{x\in  F_m } \left | \frac{ \partial f}{\partial x}(m,x)\right |^2.$$
By inserting this back into \eqref{eq:sigmaineq} we obtain the statement. 
\end{proof}

\begin{proof}[Proof of Theorem \ref{th:extend}]
For any function $p$ on $\{1, \ldots, N\} \times \bbR$ such that $x\mapsto p(m,x)$ is a polynomial we can write
\begin{multline}\label{eq:proofextendA}
\left|\EE \left[ {\rm e}^{{\rm i } t X_n(f)} \right]-{\rm e}^{-\frac{\sigma(f)^2t^2}{2}}\right|\leq  
\left|\EE \left[ {\rm e}^{{\rm i } t X_n(f)} \right]-\EE \left[ {\rm e}^{{\rm i } t X_n(p)} \right]\right|
\\
+\left|\EE \left[ {\rm e}^{{\rm i } t X_n(p)} \right]-{\rm e}^{-\frac{\sigma(p)^2t^2}{2}}\right|+\left|{\rm e}^{-\frac{\sigma(p)^2t^2}{2}}-{\rm e}^{-\frac{\sigma(f)^2t^2}{2}}\right|
\end{multline}
We claim that we can choose $p$ such that the first and last term at the right-hand side are small. Let $G$ be a compact set containing both the compact set $E$ and all $F_m$ from Lemma's \ref{lem:extendB} and \ref{lem:extendC} respectively. For $\eps>0$, choose $p$ such that  
\begin{equation}
\label{eq:proofetendB}
\sum_{m=1}^N \sup_{x \in F_m} \left| \frac{ \partial f}{\partial x}(m,x)- \frac{ \partial p}{\partial x}(m,x)\right|^2<\eps,
\end{equation}
and such that $p(m,x)$ is a polynomial in $x$. Also note that by the assumptions of Theorem \ref{th:main1}, there exists a constant $C_0$ such that
\begin{equation}\label{eq:benmoe}
\|[P_n,J_m] \|_2 \leq C_0,
\end{equation}
for  $n\in \mathbb N$ and $m=1, \ldots, N$.

We now recall that for any two real valued random variables $X,Y$ with $\EE X= \EE Y$ we have the identity
\begin{multline}
\left|\EE \left[{\rm e}^{{\rm i} t X}\right]- \EE \left[{\rm e}^ {{\rm i} t Y} \right]\right|=
\left|\EE \left[{\rm e}^{{\rm i} t X}- {\rm e}^ {{\rm i} t Y} \right]\right|\leq
\EE \left[\left| 1-{\rm e}^ {{\rm i} t (X-Y)}\right| \right] \\
\leq  |t| \EE  \left[ \left|X-Y\right|\right] \leq  |t| \sqrt{\Var (X-Y)}
|,
\end{multline}
for $ t \in \bbR$. This means by Lemma \ref{lem:extendB}, \eqref{eq:proofetendB} and \eqref{eq:benmoe}, that 
$$\left|\EE \left[ {\rm e}^{{\rm i } t X_n(f)} \right]-\EE \left[ {\rm e}^{{\rm i } t X_n(p)} \right]\right|\leq |t| \sqrt{\Var X_n(f-p)} \leq |t| C_1 \eps,$$
for $ t \in \bbR$ some constant $C_1$. Moreover, by Lemma \ref{lem:extendC} and \eqref{eq:proofetendB}   we have that 
$$\left|{\rm e}^{-\frac{\sigma(p)^2t^2}{2}}-{\rm e}^{-\frac{\sigma(f)^2t^2}{2}}\right|\leq C_2 \eps$$
for some constant $C_2$ By substituting this back into \eqref{eq:proofextendA} and applying Theorem \ref{th:main1} we obtain
$$
\left|\EE \left[ {\rm e}^{{\rm i } t X_n(f)} \right]-{\rm e}^{-\frac{\sigma(f)^2t^2}{2}}\right|\leq   (C_1+C_2) \eps
\left|\EE \left[ {\rm e}^{{\rm i } t X_n(p)} \right]-{\rm e}^{-\frac{\sigma(p)^2t^2}{2}}\right| \to (C_1+C_2) \eps,
$$
as $n \to \infty$. Since $\eps$ was arbitrary, we have $\left|\EE \left[ {\rm e}^{{\rm i } t X_n(f)} \right]-{\rm e}^{-\frac{\sigma(f)^2t^2}{2}}\right|\to 0$ as $n \to \infty$ and the statement follows.
\end{proof}

\subsection{Proof of Theorem \ref{th:main0growing}} 
\begin{proof}[Proof of Theorem \ref{th:main0growing}]
The proof is a simple extension of the proof of Theorem \ref{th:main0} and is again a consequence of Corollary \ref{cor:comp}. The main difference being that \eqref{eq:proofthm0A}, \eqref{eq:proofthm0B} and \eqref{eq:proofthm0C} are not sufficient, since the the number of terms in Corollary \ref{cor:comp} grows with $n$. Instead, we replace \eqref{eq:proofthm0A} with 
\begin{multline*}
\sum_{m=1}^{N_n} \left\|R_{n,2\rho(k+1)}\left(f(m,\mathbb J_{m,S})-f(m,\tilde {\mathbb J}_{m,S}\right)R_{n,2\rho(k+1)}\right\|_\infty\\
\leq \sum_{m=1}^{N_n} (t_{m+1}^{(n)}-t_{m}^{(n)}) \left\|R_{n,2\rho (k+1)}\left(g(m,\mathbb J_{m,S})-g(m,\tilde {\mathbb J}_{m,S})\right)R_{n,2\rho(k+1)}\right\|_\infty
\\
 \leq (\beta-\alpha) \sup_{m=1,\ldots,N_n} \left\|R_{n,2\rho(k+1)}\left(g(m,\mathbb J_{m,S})-g(m,\tilde {\mathbb J}_{m,S})\right)R_{n,2\rho(k+1)}\right\|_\infty \to 0,
\end{multline*}
as $n \to \infty$ for sufficiently large $S$, by the condition \eqref{eq:growingextendcondition1}  (and the exact same reasoning as in the beginning of the proof of Theorem \ref{th:extend}). We also replace \eqref{eq:proofthm0B} by 
\begin{multline}
\sum_{m=1}^{N_n} \left\|R_{n,2\rho(k+1)}f(m,\tilde {\mathbb J}_{m,S})R_{n,2\rho(k+1)}\right\|_\infty\\
\leq 
\sum_{m=1}^{N_n} (t_{m+1}^{(n)}-t_{m}^{(n)}) \left\|R_{n,2\rho(k+1)}g(m,\tilde {\mathbb{ J}}_{m,S})R_{n,2\rho(k+1)}\right\|_\infty \\
\leq (\beta-\alpha)\sup_{m=1,\ldots,N_n} \left\|R_{n,2\rho(k+1)}g(m,\tilde {\mathbb
J}_{m,S})R_{n,2\rho(k+1)}\right\|_\infty,
\end{multline}
which is bounded in $n$ by the assumption  in the theorem and the fact that $g(m,x)$ is a polynomial in $x$. This also shows that  
\begin{equation}
\sum_{m=1}^{N_n} \left\|R_{n,2\rho(k+1)}f(m,\mathbb {J}_{m,S})R_{n,2\rho(k+1)}\right\|_\infty,
\end{equation}
is also bounded $n$. Hence, by inserting these identities in the conclusion of Corollary \ref{cor:comp} we proved the statement.\end{proof}

\subsection{Proof of Theorem \ref{th:main1growing}}
\begin{proof}[Proof of Theorem \ref{th:main1growing}]
The proof is exactly the same as the proof of Theorem \ref{th:main1}, with the only difference that we now rely on (the proof of) Theorem \ref{th:main0growing} and use Corollary \ref{prop:cumulantstoeplitzvaryingn} in the final step instead of Corollary \ref{prop:cumulantstoeplitzfixedn}. 
\end{proof}

\subsection{Proof of Theorem \ref{th:extendgrowing}}
The proof of Theorem \ref{th:extendgrowing} also follows the same  arguments. We start with the following adjustment of Lemma \ref{lem:extendB}.
\begin{lemma}\label{lem:extendBgrowing}
Assume there exists a compact set $E \subset \mathbb R$ such that \\
(1) $S(\mu^{(n)}_m) \subset E$ for $m=1, \ldots, N_n$ and $n \in \mathbb N$, \\
or, more generally,\\
(2) for $k \in \mathbb N$, we have 
$$\sup_{m=1,\ldots,N_n} \int_{\mathbb R \setminus E}  |x|^k K_{n,N}(m,x,m,x) {\rm d} \mu_m^{(n)}(x) = o(1/n),$$ as $n\to \infty$. 

Then for any  function $g$ such that $x \mapsto g(t,x)$ is a continuously differentiable function with at most polynomial growth at $\pm \infty$, i.e $g(t,x) = \mathcal O(|x|^M)$ as $|x| \to \infty$, for some $M$ independent of $t\in I$, we have
\begin{equation}
\label{eq:extendassumptiongrowing} \Var X_n(f) \leq  \frac N2 \sum_{m=1}^N (t_{m+1}^{(n)}-t_m^{(n)})^2  \left\|\frac{\partial g}{\partial x}  \right\|_{\mathbb L_\infty(I\times E)}^2 \sup_{m=1,\ldots,N_n} \|[P_n,J_m]\|_2^2 + o(1),\end{equation}
as $n \to \infty$. 
\end{lemma}
\begin{proof}
We use the same arguments as in the proof of Lemma \ref{lem:extendB}. The only possible issue that we need to address is that in \eqref{eq:conclusionlemmaextend} the constant in the $o(1)$ term my depend on $N$ and hence, in the present situation, also on $n$.  However,  by the fact we take the supremum over $m$ in  \eqref{eq:extendassumptiongrowing} we have that the constant in the $o(1)$ terms in \eqref{eq:extensionD} and \eqref{eq:extensionsE} can be chosen such that they do not on $n$.  By following the same proof we see that the constant in the $o(1)$ does not depend on $n$ and hence we have 
$$\Var X_n(f) \leq  \frac{N_n}{2} \sum_{m=1}^{N_n} \sup_{x \in E} \left|\frac{\partial f}{\partial x} (m,x)  \right|^2 \|[P_n,J_m]\|_2^2 + o(1),$$
as $n \to \infty$. Since $f$ is considered to be of the special type \eqref{eq:varyinglinstat}, we rewrite the latter in terms of $g$, $$\Var X_n(f) \leq  \frac {N_n}{ 2} \sum_{m=1}^{N_n}(t_{m+1}^{(n)}-t_m^{(n)})^2 \left\|\frac{\partial g}{\partial x} \right\|_{\mathbb L_\infty(I\times E)}^2 \|[P_n,J_m]\|_2^2 + o(1),$$
as $n \to \infty$. The statement now follows by estimating the summands in terms of their surpema. 
\end{proof}
\begin{lemma}\label{lem:extendCgrowing}
Let $a_0^{(m)}$ and $a_1^{(m)}$ be piecewise continuous function on $I$ for $m=1, \ldots, N_n$.  Then for any real valued function $g$ on $\{1,\ldots,N\} \times \bbR$ such that $x\mapsto g(\tau,x)$ is continuously differentiable, we have that 
$$\sigma(g)^2= \sum_{k=1}^\infty \iint_{I\times I}
{\rm e}^{-|\tau_{1}(t)-\tau_{2}(t)|k } g_k(t_1) g_k(t_2) {\rm d} t_1 {\rm d} t_2,$$
is finite, where 
$$g_k(t)= \frac{1}{2 \pi} \int_0^{2 \pi} g\left(m,a_0(t)+ 2 a_1 (t) \cos t\right) {\rm e}^{-{\rm i} k \theta} {\rm d} \theta.$$
Moreover,  for some constant $C$ we have 
$$\sigma(f)^2 \leq  C \left\| \frac{ \partial g}{\partial x}\right\|_{\mathbb L_\infty (F)}^2,$$
where $F= \{ (t,x) \mid x \in  [a_0(t)- 2 a_{1}(t), a_0(t)+2 a_1(t)], \ t \in I \}.$
\end{lemma}
\begin{proof}
As in \eqref{eq:fouriertransformvariance} we rewrite
$\sigma(g)^2$ as 
$$\sigma(g)^2= \frac1\pi\sum_{k=1}^\infty  \int \frac{k^2}{\omega^2+k^2}   \left| \int_I  {\rm e}^{- {\rm i}  \omega \tau(t)} g_k(t){\rm d} t \right|^ 2 {\rm d} \omega,$$
showing that $\sigma(g)^2$ is positive (if finite). By applying the Cauchy-Schwarz inequality to the integral over $t$, the estimate \eqref{eq:piiiii} and the Plancherel Theorem we have
$$\sigma(g)^2\leq (\beta-\alpha) \pi \sum_{k=1}^\infty  \int_\bbR {k^2}  | g_k(t) |^ 2 {\rm d}t.$$ 
Then by integration by parts we find 
$$kgf_k(t)= \frac{1}{2\pi } \int_0^{2 \pi} \frac{ \partial g}{\partial x} (t,a_0(t)+ 2 a_1(t) \cos  \theta) a_1(t) \sin  \theta  \ {\rm e}^{-{\rm i} k \theta} {\rm d} \theta.$$
By applying Parseval and  estimating the $\mathbb L_2$ norm by the $\mathbb L_\infty$ norm, we find there exists a constant $C$ such that 
$$\sum_{k=1}^\infty k^2 |g_k(t)|^2 \leq C  \left\| \frac{ \partial g}{\partial x}\right\|_{\mathbb L_\infty (F)}^2,$$
and  this finishes the proof.
\end{proof}
\begin{proof}[Proof of Theorem \ref{th:extendgrowing}]
The proof of is now identical to the proof of Theorem \ref{th:extendgrowing}, with only difference being that we use Theorem \ref{th:main1growing}, Lemma's \ref{lem:extendBgrowing} and \ref{lem:extendCgrowing} instead of Theorem \ref{th:main1} and Lemma's \ref{lem:extendB} and \ref{lem:extendC}.  
\end{proof}

\end{document}